\documentclass[acmsmall,screen,dvipsnames]{acmart}
\pdfoutput=1

\setcopyright{rightsretained}
\acmDOI{10.1145/3704855}
\acmYear{2025}
\acmJournal{PACMPL}
\acmVolume{9}
\acmNumber{POPL}
\acmArticle{19}
\acmMonth{1}
\received{2024-07-11}
\received[accepted]{2024-11-07}

\makeatletter
\def\@acmplainindent{0pt}
\def\@acmdefinitionindent{0pt}
\def\@proofindent{\noindent}
\makeatother

\usepackage{hyperref}
\usepackage{enumitem}
\usepackage{cancel}
\usepackage{cleveref}
\usepackage{graphicx}
\usepackage{mathpartir}
\usepackage{mathtools}
\usepackage{stmaryrd}
\usepackage{subcaption}
\usepackage{xcolor}
 \hypersetup{
    colorlinks,
    linkcolor={red!50!black},
    citecolor={blue!50!black},
    urlcolor={blue!80!black}
}
\usepackage{wrapfig}
\usepackage{scalerel}
\usepackage{thmtools} 
\usepackage{thm-restate}
\usepackage[all, cmtip]{xy}
\usepackage{framed}
\usepackage{xr}
\usepackage{dk}

\usepackage{tikz}
\usetikzlibrary{shapes,arrows,positioning}

\def\extended{0}
\usepackage{olmacros}

\ifx\extended\undefined
\fi

\newtheorem{remark}{Remark}



\begin{document}
\title{A Demonic Outcome Logic for Randomized Nondeterminism}

\author{Noam Zilberstein}
\orcid{0000-0001-6388-063X}
\affiliation{%
  \institution{Cornell University}
  \city{Ithaca}
  \country{USA}
}
\email{noamz@cs.cornell.edu}

\author{Dexter Kozen}
\orcid{0000-0002-8007-4725}
\affiliation{%
  \institution{Cornell University}
  \city{Ithaca}
  \country{USA}
}
\email{kozen@cs.cornell.edu}

\author{Alexandra Silva}
\orcid{0000-0001-5014-9784}
\affiliation{%
  \institution{Cornell University}
  \city{Ithaca}
  \country{USA}
}
\email{alexandra.silva@cornell.edu}

\author{Joseph Tassarotti}
\orcid{0000-0001-5692-3347}
\affiliation{%
  \institution{New York University}
  \city{New York}
  \country{USA}
}
\email{jt4767@nyu.edu}

\begin{abstract}
Programs increasingly rely on randomization in applications such as cryptography and machine learning. Analyzing randomized programs has been a fruitful research direction, but there is a gap when programs also exploit nondeterminism (for concurrency, efficiency, or algorithmic design). In this paper, we introduce \emph{Demonic Outcome Logic} for reasoning about programs that exploit {\em both} randomization and nondeterminism. The logic includes several novel features, such as reasoning about multiple executions in tandem and manipulating pre- and postconditions using familiar equational laws---including the distributive law of probabilistic choices over nondeterministic ones. We also give rules for loops that both establish termination and quantify the distribution of final outcomes from a single premise. We illustrate the reasoning capabilities of Demonic Outcome Logic through several case studies, including the Monty Hall problem, an adversarial protocol for simulating fair coins, and a heuristic based probabilistic SAT solver.
\end{abstract}

\maketitle              

\section{Introduction}

Randomization is critical in sensitive software domains such as cryptography and machine learning. While it is difficult to establish correctness of these systems alone, the difficulty is increased as they become distributed, since nondeterminism is introduced by scheduling the concurrent processes.
Verification techniques exist for reasoning about programs that are both randomized and nondeterministic using expectations \cite{wpe} and refinement \cite{polaris}, but there are currently no logics that allow for specifying and reasoning about the multiple probabilistic executions that arise from this combination of effects.


In program logics such as Hoare Logic \cite{hoarelogic}, preconditions and postconditions are propositions about the start and end states of the program. When moving to a probabilistic setting, it is not enough for these propositions to merely describe states, they must also quantify how likely the program is to end up in each of those states, as correctness is a property of the \emph{distribution} of outcomes.
Several logics exist for reasoning about purely probabilistic programs in this way, including Probabilistic Hoare Logic \cite{hartog,corin2006probabilistic,den_hartog1999verifying}, Ellora \cite{ellora}, Outcome Logic \cite{zilberstein2024outcome,outcome,zilberstein2024relatively}, and Quantitative Weakest Hyper Pre \cite{zhang2024quantitative}.
The benefits of reasoning about multiple executions are:

\begin{itemize}[leftmargin=1em,label={}]
\item \textsf{\emph{Outcomes.}}
As opposed to expectation reasoning, program logics can describe multiple outcomes in a single specification, giving a more comprehensive account of the \emph{distribution} of behaviors. This is displayed in \Cref{sec:von-neumann}, where we prove that a program simulates a fair coin by enumerating the outcomes and showing that they are uniformly distributed.

\smallskip
\item \textsf{\emph{Compositionality.}}
Inference rules allow us to reason about programs in a compositional, but also concise way. This is evident in our termination rules (\Cref{sec:loops}), which have fewer premises compared to similar rules in other reasoning systems.
\end{itemize} 
This paper introduces \emph{Demonic Outcome Logic}, a program logic for reasoning about \emph{randomized nondeterministic} programs---programs that have both probabilistic and nondeterministic choice operators. This work builds both on Outcome Logic---which can be used to reason about randomization \emph{or} nondeterminism, but not both together---and a large body of work on the semantics of randomized nondeterministic programs \cite{jifeng1997probabilistic,morgan1996refinement,jacobs2008coalgebraic,tix2009semantic,varacca2002powerdomain,wpe}.  Our contributions can be grouped in four categories, as follows:


\begin{enumerate}[leftmargin=1.75em]
\item\textsf{\emph{From Equations to Propositions.}}
Semantic objects to capture both randomization and nondeterminism are often described in terms of \emph{equations}, stating properties of relevant operators such as idempotence and distributivity. In logic, implications are used to manipulate assertions and facilitate reasoning. In our proposed program logic,  we want to bring these worlds together and have logical implications mirror the equational laws. The challenge is that---as we will see in \Cref{sec:equations,sec:assertions}---the equations do not immediately hold as implications, so a carefully designed assertion language is needed in which the laws indeed hold. 
\smallskip
\item\textsf{\emph{Demonic Outcome Logic.}}
This paper presents the first program logic for reasoning about {\em distributions of outcomes} with both randomization and nondeterminism. Making the logic \emph{demonic}---meaning that the postcondition applies to every nondeterministic possibility---allowed us to create simple and convenient inference rules.

\smallskip
While our logic has similarities to Weakest Pre-Expectation calculi, Demonic Outcome Logic involves some key differences. Demonic Outcome Logic can reason about many executions together, which allows us to specify the distribution of outcomes rather than just quantitative properties of that distribution. This is demonstrated in \Cref{sec:von-neumann}, where a program is specified in terms of multiple distinct outcomes and \Cref{sec:sat-solving}, where case analysis is done over multiple nondeterministic executions. It was necessary to develop new sound rules for this more expressive form of reasoning and idempotence of the logical connectives proved crucial in ways that do not appear in prior work.

\smallskip
\item\textsf{\emph{Loops and Termination.}}
Our rules for reasoning about loops in \Cref{sec:loops} allow us to prove termination, while simultaneously specifying the precise distribution of outcomes upon termination. This goes beyond prior work on expectation based reasoning \cite{mciver2005abstraction,mciver2018new}, where termination is established with a propositional invariant describing only a singe outcome. Our rules also have fewer premises, making them simpler to apply in our experience.

\smallskip
\item\textsf{\emph{Case Studies.}}
We investigate three case studies in \Cref{sec:examples} to demonstrate the applicability of our logic. For example, we present a protocol to simulate a fair coin flip given a coin whose bias is continually altered by an adversary, and show that this program terminates with the outcomes being uniformly distributed. We also prove that a probabilistic SAT solver terminates, even if some of the heuristics involved are adversarially chosen.
\end{enumerate}
We begin in \Cref{sec:overview} by outlining the challenges of reasoning about randomized nondeterminism, and how this informed the design of Demonic Outcome Logic. Next, in \Cref{sec:semantics}, we describe the denotational model that we use for semantics of programs. In \Cref{sec:logic}, we introduce Demonic Outcome Logic and the inference rules for reasoning about sequential programs. We discuss reasoning about loops in \Cref{sec:loops}. We examine three case studies in \Cref{sec:examples} to demonstrate the utility of the logic. Finally, we discuss related work and conclude in \Cref{sec:related,sec:conclusion}.

\section{An Overview of Demonic, Outcome-Based Reasoning}
\label{sec:overview}

The issue of combining randomization and nondeterminism is one of the most difficult and subtle challenges in program semantics. In this section, we outline the desired properties of a logic for that purpose and how the design of that logic intersects with prior work on semantics.

Most applications of randomized nondeterminism take a \emph{demonic} view of nondeterminism, wherein the nondeterminism is controlled by an \emph{adversary}, and the program is correct only if the distribution of outcomes satisfies a certain post-condition, {\em regardless} of how the adversary might have resolved the nondeterminism. One such domain is verification of distributed cryptographic protocols, where the probability that an adversary can guess a secret message must be negligible regardless of how the scheduler interleaves the concurrent processes.

To demonstrate the complex interaction between demonic non-determinism and probabilistic choice, we consider an example in which an adversary tries to guess the outcome of a fair coin flip. The coin flip is represented by the program $x\coloneqq\flip{\frac12}$, whose denotation is a singleton set containing a distribution of outcomes in which $x = \tru$ and $x = \fls$ both occur with probability $\frac12$. The adversarial choice is performed by the program $y\gets\mathbb B$, where $\mathbb B = \{\tru,\fls\}$. Operationally, we presume that the adversary can make this choice in any way it pleases, including by flipping a biased coin. That means that the adversary can force $y$ to be true, it can force $y$ to be false, or it can make both outcomes possible with any probability. Denotationally, the semantics of these programs is a map $\de{C} \colon \Sigma \to \bb{2}^{\D(\Sigma)}$ from states $\Sigma$ to sets of distributions of states, shown below. 
\[
  \de{x \coloneqq \mathsf{flip}\left(\tfrac12\right)}\!\!(\sigma) = \left\{~
    \scriptsize
    \arraycolsep=0em
    \def\arraystretch{1.25}
    \begin{array}{llll}
      \sigma[x \coloneqq \tru &] &\mapsto& \frac12
      \\
      \sigma[x \coloneqq \fls &] &\mapsto& \frac12
    \end{array}
  ~\right\}
\qquad
  \de{y \gets \mathbb B}\!(\sigma) = \left\{~
    {\scriptsize
    \arraycolsep=0em
    \begin{array}{llll}
        \def\arraystretch{1.25}
      \sigma[y \coloneqq \tru &] &~\mapsto~& p
      \\
      \sigma[y \coloneqq \fls &] &~\mapsto~& 1-p
    \end{array}}
    \;\middle|\;
    p \in [0,1]~
  \right\}
\]
Now, we consider two variants of composing these programs, shown below. On the left is a variant in which the adversary picks last and on the right is a variant in which the adversary picks first.
\[
  x \coloneqq \mathsf{flip}\left(\tfrac12\right) \fatsemi y \gets \mathbb B
  \qquad\qquad\qquad
  y \gets \mathbb B \fatsemi x \coloneqq \mathsf{flip}\left(\tfrac12\right)
\]
We wish to know the probability that $x=y$. In the program on the left, the value of $x$ is fixed before the adversary makes its choice, meaning that it can choose a distribution in which $x=y$ with any probability $p\in[0,1]$. However, in the program on the right, the adversary chooses first, and so the later coin flip will ensure that $x=y$ with probability exactly $\frac12$.


We will examine how to prove this fact using program logics. First, in \Cref{sec:equations} we will lay out the semantic properties of random and nondeterministic choices using equations, and we will show how those equations inform propositional reasoning about outcomes. Next, in \Cref{sec:composition}, we will see how to make such propositional inferences about program behavior using our new logic. We will then overview how to reason about more complicated looping programs in \Cref{sec:ov-loops}.

%
%

\subsection{From Equational Laws to Propositional Reasoning}
\label{sec:equations}

Equational theories are a useful tool for defining the behavior of programatic operators in terms of laws that must be upheld. This has been studied extensively in the context of semantics of probabilistic nondeterminism \cite{bonchi2022theory,mio2020monads,mislove2000nondeterminism,tix2000convex,bonchi2021presenting,bonchi2019theory}, where equations are used to describe properties of nondeterminism, random choice, and the interaction between the two.
We will now explore the link between equational theories and propositional reasoning about program outcomes.
As we explain in this section, care has to be taken to craft a model in which the desired properties of program operations can be used not only to establish equality of semantic objects, but also as logical implications.

In the following, the variables $X$, $Y$, and $Z$ denote the outcomes of a program.
The nondeterministic choice operator---denoted $X \nd Y$---is an adversarial choice between the outcomes $X$ and $Y$. It should be idempotent, commutative, and associative, as captured by the following equations:
\[
  X\nd X = X
  \qquad\qquad
  X \nd Y = Y\nd X
  \qquad\qquad
  (X \nd Y)\nd Z = X \nd (Y\nd Z)
\]
That is, choosing between $X$ and $X$ is equivalent to making no choice at all, and the ordering of choices makes no difference.
The probabilistic choice operator $X \oplus_p Y$, where $p\in[0,1]$, represents a random choice where $X$ and $Y$ occur with probability $p$ and $1-p$, respectively. This operator obeys similar laws, with probabilities adjusted appropriately.
\[
  X \oplus_p X = X
  \qquad\qquad
  X\oplus_p Y = Y \oplus_{1-p} X
  \qquad\qquad
  (X \oplus_p Y) \oplus_q Z = X \oplus_{pq} (Y \oplus_{\frac{(1-p)q}{1-pq}} Z)
\]
In addition, the following \emph{distributive law} requires that random choices distribute over nondeterministic ones, much like multiplication distributes over addition in standard arithmetic.
\[
  X \oplus_p (Y \nd Z) = (X \oplus_p Y) \nd (X \oplus_p Z)
\]
This law corresponds to our interpretation of demonic nondeterminism. On the left-hand side of the equation, we first randomly choose to execute either $X$ or $Y\nd Z$, and then---if the second option is taken---the nondeterministic choice is resolved. Applying this axiom as a rewrite rule from left to right would push the nondeterministic choice to the top above the probabilistic choice.

Traditionally, equational theories have been used to decide equality between programs \cite{kat}. Here, we repurpose the equations for propositional reasoning about program outcomes. That is, if $\varphi$, $\psi$, and $\vartheta$ are assertions about outcomes, then $\varphi \nd \psi$ asserts that $\varphi$ and $\psi$ are two possible nondeterministic outcomes, and $\varphi\oplus_p\psi$ asserts that $\varphi$ occurs with probability $p$ and $\psi$ occurs with probability $1-p$. This is inspired by Outcome Logic \cite{outcome}, but there are now two types of outcomes (probabilistic and non-deterministic). We want to rewrite the desired equations above as logical equivalences, e.g. the distributive law would be transformed to:
\[
  \varphi \oplus_p (\psi \nd \vartheta)
  \ \Leftrightarrow\ 
  (\varphi\oplus_p \psi) \nd ( \varphi\oplus_p \vartheta)
\]
However, one has to be careful. For example, as we illustrate below, the idempotence property $\varphi\nd \varphi \Leftrightarrow \varphi$ only holds as an implication in a carefully crafted model. 

One benefit of propositional reasoning vs equational reasoning is the ability to \emph{weaken} assertions. For instance, returning to the coin flip example, the following proposition precisely captures the result of the program in which the adversary chooses first, where $\sure P$ means that $P$ occurs with probability 1 (almost surely).
\begin{equation}\label{eq:full-description}
  ( \sure{y=\tru} \land (\sure{x=\tru} \oplus_\frac12 \sure{x=\fls})) \nd ( \sure{y=\fls} \land (\sure{x=\tru} \oplus_\frac12 \sure{x=\fls}))
\end{equation}
But the precise values of $x$ and $y$ are cumbersome to remember, and obfuscate the property that we want to convey.
Instead, we can weaken the assertion to record only whether $x=y$ or $x\neq y$.
\[
  (\sure{x = y} \oplus_\frac12 \sure{x\neq y}) \nd (\sure{x\neq y} \oplus_\frac12 \sure{x=y})
\]
It is now tempting to use commutativity of $\oplus_\frac12$ and idempotence of $\nd$ to perform the following simplification, concisely asserting the probability that the adversary can determine the value of $x$.
\begin{equation}\label{eq:idem}
  (\sure{x=y} \oplus_\frac12 \sure{x\neq y}) \nd  (\sure{x=y} \oplus_\frac12 \sure{x\neq y})
  \quad\Rightarrow\quad
  \sure{x=y} \oplus_\frac12 \sure{x\neq y}
\end{equation}
However, care had be taken to craft a model in which implication (\ref{eq:idem}) is valid.
Unlike the idempotence equation $X \nd X = X$---which applies when the exact same set of distributions appears on each side---the implication version (\ref{eq:idem}) operates on \emph{approximations} of those sets of distributions. Recall from (\ref{eq:full-description}) that $x=y$ is satisfied by $x = y= \tru$ on the left hand side of the $\nd$, whereas it is satisfied by $x = y =\fls$ on the right, so even though both sets of distributions satisfy $x=y \oplus_{\frac12} x\neq y$, they are not equal. The full details of this example are shown in \Aref{app:idempotence}.

In \Cref{sec:assertions}, we give a full account of how our demonic logic supports idempotence and all of the other properties that were stated equationally above. These properties do not hold by default, but rather required some intentional choices in the design of the logic. In particular, as we will detail in \Cref{sec:assertions}, the assertion language will not include disjunctions or existential quantification. The result is a deductive system that is able to express more intuitive and concise specifications.

\subsection{Program Logics and Compositionality}
\label{sec:composition}

Inspired by Hoare Logic, our goal is to develop a logic where programs are specified in terms of pre- and postconditions using \emph{triples} of the form $\triple\varphi{C}\psi$. Here, $\varphi$ and $\psi$ are outcome assertions from \Cref{sec:equations}, whose models are distributions of states. Since the program $C$ is nondeterministic---and is interpreted as a map into sets of distributions---the postcondition $\psi$ must be satisfied by every distribution in that resulting set.
We call this logic Demonic Outcome Logic, as it supports probabilistic reasoning in the style of Outcome Logic \cite{outcome,zilberstein2024outcome,zilberstein2024relatively} with the crucial addition of demonic nondeterminism.

Compositional reasoning---the ability to analyze a complex program in terms of its subprograms---is the hallmark of program logics. This is exemplified by the inference rule for sequential composition; we infer the behavior of a composite program from the behavior of its constituent parts.
\[
\inferrule{
  \triple{\varphi}{C_1}{\vartheta}
  \\
  \triple{\vartheta}{C_2}{\psi}
}{
  \triple{\varphi}{C_1\fatsemi C_2}{\psi}
}{\ruleref{Seq}}
\]
The soundness of this rule is not a given in the randomized nondeterminism setting. It relies on being able to define the semantics $\de{C_1\fatsemi C_2}$ in terms of $\de{C_1}$ and $\de{C_2}$. As we have already seen at the beginning of \Cref{sec:overview}, the semantics is a map from states to sets of distributions of states: $\de{C} \colon \Sigma \to \bb{2}^{\D(\Sigma)}$.
We compose the semantics of program fragments using a lifted version, known as the \emph{Kleisli extension} $\de{C}^\dagger \colon \bb{2}^{\D(\Sigma)} \to \bb{2}^{\D(\Sigma)}$. If $\bb{2}^{\D(-)}$ were a \emph{monad}, then we would have the compositionality property to guarantee the soundness of the \ruleref{Seq} rule: $ \de{C_1 \fatsemi C_2}^\dagger = \de{C_2}^\dagger \circ \de{C_1}^\dagger$. Unfortunately, as originally shown by \citet{varacca_winskel_2006} (see also \citet{zwart2019no,parlant2020monad}), no such composition operator exists.

But compositionality can be retained by requiring the sets of distributions to be \emph{convex} \cite{morgan1996refinement,mislove2000nondeterminism,tix1999continuous,jacobs2008coalgebraic}. That is, whenever two distributions $\mu$ and $\nu$ are in the set of outcomes, then all convex combinations ($p\cdot \mu + (1-p)\cdot\nu$ for $p\in[0,1]$) are also in the set of possible results. Convexity corresponds to our operational interpretation of nondeterminism---the adversary may flip biased coins to resolve choices \cite[Theorem 6.12]{varacca2002powerdomain}.
Returning to the coin flip example, we can derive the following specifications for the primitive operations.
\[
  \triple{\sure\tru}{y\gets\mathbb B}{\sure{y=\tru} \nd \sure{y=\fls}}
  \qquad
  \arraycolsep=0em
  \def\arraystretch{1.25}
  \begin{array}{r}
  \triple{\sure{y={\phantom{\fls}\makebox[0pt][r]{$\tru$}}}}{x\coloneqq\mathsf{flip}\left(\tfrac12\right)}{\sure{x=y} \oplus_\frac12 \sure{x\neq y}}
  \\
  \triple{\sure{y=\fls}}{x\coloneqq\mathsf{flip}\left(\tfrac12\right)}{\sure{x\neq y} \oplus_\frac12 \sure{x= y}}
  \end{array}
\]
That is, the adversarial choice results in two nondeterministic outcomes, separated by $\nd$. Executing the probabilistic choice in either of those states yields two further probabilistic outcomes.

Since the first command splits the execution into two outcomes, we need one more type of composition in order to stitch these together into a specification for the composite program. That is, we need the ability to decompose the precondition and analyze the program with each resulting sub-assertion individually. We may also wish to do the same for probabilistic choices.
\[
\inferrule{
  \triple{\varphi_1}{C}{\psi_1}
  \\
  \triple{\varphi_2}C{\psi_2}
}{
  \triple{\varphi_1\nd\varphi_2}C{\psi_1\nd\psi_2}
}{\ruleref{ND Split}}
\qquad
\inferrule{
  \triple{\varphi_1}{C}{\psi_1}
  \\
  \triple{\varphi_2}C{\psi_2}
}{
  \triple{\varphi_1\oplus_p\varphi_2}C{\psi_1\oplus_p\psi_2}
}{\ruleref{Prob Split}}
\]
Using \ruleref{Seq}, \ruleref{ND Split}, and the idempotence rule (\ref{eq:idem}), we derive the triple below on the left (shown fully in \Aref{app:coin}). A similar derivation for the reversed version yields the triple on the right.
\[
\begin{array}{lll}
  \ob{\sure\tru} &\qquad\qquad& \ob{\sure\tru}
  \\
  \;\;y\gets \mathbb B\fse
  &&
  \;\; x\coloneqq\mathsf{flip}\left(\tfrac12\right)\fse
  \\
  \;\; x\coloneqq\mathsf{flip}\left(\tfrac12\right) 
  &&
  \;\;y\gets \mathbb B
  \\
  \ob{\sure{x=y} \oplus_\frac12 \sure{x\neq y}}
  &&
  \ob{\sure{x=y} \nd \sure{x\neq y}}
\end{array}
\]
If the adversary picks first, then it can only guess the value of $x$ with probability $\frac12$. But if the coin flip is first, we only know that $x=y$ occurs with some probability. In fact, $x=y\nd x\neq y$ is equivalent to $\tru$, so it certainly does not give us a robust security guarantee, leaving open the possibility that the adversary can guess $x$.



\subsection{Reasoning about Loops}
\label{sec:ov-loops}

Reasoning about loops is challenging in any program logic, and Demonic Outcome Logic is no exception. When reasoning about probabilistic loops, one often wants to prove not only that some property holds upon termination, but also that the program \emph{almost surely terminates}---the probability of nontermination is 0.
An example of an almost surely terminating program is shown below. It is an adversarial random walk, where the agent steps towards 0 with probability $\frac12$, otherwise the adversary moves the agent to an arbitrary position between 1 and 5.
\[
\begin{array}{l}
  \code{while}~ x > 0 ~\code{do} \\
  \quad (x \coloneqq x-1) \oplus_{\frac12} (x\gets \{1, \ldots, 5\})
\end{array}
\]
It may seem surprising that this program almost surely terminates; after all, the adversary can always choose the worst possible option of resetting the position to 5. However, as the number of iterations goes to infinity, the probability of decrementing $x$ five times in a row goes to 1.

Demonic Outcome Logic has a simple inference rule for proving almost sure termination, inspired by a rule of \citet{mciver2005abstraction}, which uses \emph{ranking functions} to quantify how close the loop is to termination. The rule states that the program almost surely terminates if the rank strictly decreases on each iteration with probability bounded away from $0$, while also preserving some invariant $P$.
\[
\inferrule{
  \triple{\sure{P \land e \land e_\mathsf{rank} = n}}C{\sure{P \land e_\mathsf{rank} < n} \oplus_p \sure{P}},\ \ p>0
}{
  \triple{\sure P}{\whl eC}{\sure{P\land\lnot e}}
}{\ruleref{Bounded Rank}}
\]
We prove the soundness of this rule in \Cref{sec:loops}.
To instantiate it for the program above, we use the invariant $P \triangleq 0 \le x\le 5$, the ranking function $e_\mathsf{rank} \triangleq x$, and the probability $p = \frac12$. This means that $x$ is always between 0 and 5 and the value of $x$ strictly decreases with probability $\frac12$ in each iteration of the loop. Applying the inference, we get the following specification for the program.
\[
  \triple{\sure{0 \le x\le 5}}{\whl{x>0}{(x \coloneqq x-1) \oplus_{\frac12} (x\gets \{1, \ldots, 5\})}}{\sure{x = 0}}
\]
This says that the program terminates in a state satisfying $x=0$ with probability 1 (\ie almost surely). Compared to the rule of \citet{mciver2005abstraction}---which is based on \emph{weakest pre-expectations}---our approach has two key advantages. First, in the pre-expectation approach, the preservation of the invariant and the decrease in rank are verified separately, whereas our rule combines the two in a single premise. Second, our rules allow the invariant to have multiple outcomes, allowing them to express a \emph{distribution} of end states, rather than a single assertion. A concrete example of this appears in \Cref{sec:von-neumann}.

%
%

\medskip

\noindent We have now seen the key ideas behind Demonic Outcome Logic, including how equational laws translate to propositional reasoning over pre- and postconditions, the challenges in making the logic compositional, and strategies for analyzing loops to establish almost sure termination. We now proceed by making these ideas formal, starting in \Cref{sec:semantics} where we define the program semantics, and continuing with \Cref{sec:logic,sec:loops} where we define the logic and rules for analyzing loops. In \Cref{sec:examples} we examine case studies to show how the logic is used.

\section{Denotational Semantics for Probabilistic Nondeterminism}
\label{sec:semantics}

In this section, we present the semantics of a simple imperative language with both probabilistic and nondeterministic choice operators, originally due to \citet{jifeng1997probabilistic} and \citet{wpe}. The syntax of the language, below, includes familiar constructs such as no-ops, variable assignment, sequential composition, if-statements, and while-loops, plus two kinds of branching choice. 
\begin{align*}
\mathsf{Cmd} \ni C \Coloneqq&~
\skp && \textsf{(No-op)} \\
|&~ x \coloneqq e && \textsf{(Variable Assignment)}
\\
|&~  C_1 \fatsemi C_2 && \textsf{(Sequential Composition)}
\\
  |&~ C_1 \nd C_2 && \textsf{(Nondeterministic Choice)} 
  \\
  |&~ C_1 \oplus_e C_2 && \textsf{(Probabilistic Choice)}
  \\
  |&~ \iftf e{C_1}{C_2} && \textsf{(Conditional)} \\
  |&~ \whl eC && \textsf{(While Loop)}
\end{align*}
Expressions $e \in \mathsf{Exp}$ range over typical arithmetic and Boolean operations, and we evaluate these expressions in the usual way. 
Nondeterministic choice $C_1\nd C_2$ represents a program that arbitrarily chooses to execute either $C_1$ or $C_2$, whereas probabilistic choice $C_1 \oplus_e C_2$, in which $e$ evaluates to a rational probability $p$, represents a program in which $C_1$ is executed with probability $p$ and $C_2$ with probability $1-p$.
In the remainder of this section, we precisely describe the semantics of the language, building on the informal account given in \Cref{sec:overview}.

\subsection{States, Probability Distributions, and Convex Sets}
\label{sec:prelim}

Before we present the semantics, we review some preliminary definitions. We begin by describing the {\em program states} $\sigma \in \Sigma \triangleq \mathsf{Var} \to\mathsf{Val}$, which are mappings from a finite set of variables $x \in \mathsf{Var}$ to values $v \in \mathsf{Val}$. Values consist of Booleans and rational numbers, making the set of states countable. 
To define the semantics, we will work with discrete probability distributions over states.

\begin{definition}[Discrete Probability Distribution]
Let $\D(X) \triangleq X \to [0,1]$ be the set of discrete probability distributions on $X$. 
The support of a distribution $\mu$ is the set of elements to which $\mu$ assigns nonzero probability $\supp(\mu) \triangleq \{ x \in X \mid \mu(x) > 0 \}$. We only consider proper distributions such that the total probability mass $|\mu| = \sum_{x\in\supp(\mu)}\mu(x)$ is $1$.
\end{definition}

We denote the the Dirac distribution centered at a point $x\in X$ by $\delta_x$, with $\delta_x(x) = 1$ and $\delta_x(y) = 0$ if $x \neq y$. Addition and scalar multiplication are lifted to distributions pointwise:
\[
(\mu_1 + \mu_2)(x) = \mu_1(x) + \mu_2(x)
\qquad\qquad
(p\cdot \mu)(x) = p\cdot \mu(x)
\]
The semantics of our language will be based on nonempty subsets of distributions $\D(\Sigma_\bot)$, where $\bot$ is a special element symbolizing divergence and $\Sigma_\bot = \Sigma \cup \{\bot\}$. We will need three {\em closure properties} to make the semantics well-defined: convexity, topological closure, and up-closure. We begin by describing convexity, which makes the semantics of sequential composition associative, and the overall program semantics compositional (\Cref{sec:composition}).

\begin{definition}[Convex Sets]
A set $S \subseteq {\D(X)}$ of distributions is \emph{convex} if $\mu\in S$ and $\nu\in S$ implies that $p\cdot \mu + (1-p)\cdot \nu \in S$ for every $p\in[0,1]$. 
\end{definition}
In order to formally define the semantics of while loops, we will have to compute certain fixpoints (explained fully in \Cref{sec:fixpoint}), which requires us to restrict our semantic domain to up-closed sets.
\begin{definition}[Up-closed Sets]
A set $S \subseteq \D(\Sigma_\bot)$ is \emph{up-closed} $\mu \in S$ and $\mu \sqsubseteq_\D \nu$ implies $\nu\in S$. The order $\mathord{\sqsubseteq_{\D}}\subseteq \D(\Sigma_\bot) \times \D(\Sigma_\bot)$ is defined as $\mu\sqsubseteq_{\D}\nu
  \ \ \text{iff}\ \ 
  \forall \sigma\in \Sigma.~ \mu(\sigma) \le \nu(\sigma)$. 
The \emph{up-closure} of a set $S$ is the set $\upcl S \triangleq \{ \nu \mid \mu \in S,\ \mu \sqsubseteq_\D \nu \}$. Thus $S$ is up-closed iff $S=\upcl S$.
\end{definition}
Note that $\mu \sqsubseteq_{\D} \nu$ implies that $\nu(\bot) \le \mu(\bot)$, and that $\delta_\bot$ is the bottom of this order, since $\delta_\bot(\sigma) = 0$ for all $\sigma\in\Sigma$, therefore $\delta_\bot(\sigma) \le \mu(\sigma)$ for any $\mu \in \D(\Sigma_\bot)$.
If $\mu(\bot) = 0$, then $\mu$ is already maximal and so $\upset\mu = \{ \mu \}$, but if $\mu(\bot) > 0$, then $\mu$ can be made larger by reassigning probability mass from $\bot$ to proper states, \eg $\delta_\bot \sqsubseteq_\D \delta_\sigma$. As a consequence, $\upset{\delta_\bot} = \D(\Sigma_\bot)$, the set of all distributions. Up-closure means that we cannot be sure whether a program truly diverges, or instead exhibits erratic nondeterministic behavior, which is a common limitation of the Smyth powerdomain \cite{s_ondergaard1992non}.
However, the program logic that we develop in this paper is concerned with proving almost sure termination of programs, so the loss of precision in the semantics when nontermination might occur does not affect the accuracy of our logic.

Finally, we require sets to be closed in the usual topological sense. A subset of $\D(X)$ is \emph{closed} if it is closed in the product topology $[0,1]^X$, where $[0,1]$ has the Euclidean topology. So closure means that a set $S$ contains all of its limit points. This will later help us to ensure that the semantics is Scott continuous by precluding unbounded nondeterminism. More precisely, it will not be possible to define a primitive command $x \coloneqq \bigstar$, which surely terminates and nondeterministically selects a value for $x$ from an infinite set (such as $\mathbb N$). While this is certainly a limitation of the semantics, it is a typical one; an impossibility result due to \citet{apt1986countable} showed that it is not possible to define a semantics that both determines whether a program terminates and also allows unbounded nondeterminism. This corresponds to \cites{Dijkstra76} operational observation that a machine cannot choose between infinitely many branches in a finite amount of time, so any computation with infinitely many nondeterministic outcomes may not terminate.
We now have all the ingredients to define our semantic domain:
\[
 \C(X) \triangleq \big\{ S \subseteq \D(X_\bot) \mid S ~\text{is nonempty, convex, (topologically) closed, and up-closed} \big\}
\]
$\C$ is a functor and, interestingly from a semantics point of view, a monad \cite{jacobs2008coalgebraic}.  For any $f\colon X\to\C(Y)$,  the Kleisli extension $f^\dagger \colon \C(X)\to \C(Y)$ and unit operation $\eta \colon X \to \C(X)$ are defined as follows:
\[
  \eta(x) \triangleq \mathord\uparrow\{ \delta_x \}
  \qquad
  f^\dagger(S) \triangleq
  \left\{
    {\sum\limits_{x\in \supp(\mu)}} \mu(x) \cdot \nu_x
    \;\; \Big| \;\;
    \mu\in S, 
     \forall x\in \supp(\mu).\
    \nu_x \in f_\bot(x)
  \right\}
\]
Where for any function $f \colon X \to \C(Y)$, we define $f_\bot \colon X_\bot \to \C(Y)$ such that $f_\bot(x) = f(x)$ for $x\in X$ and $f_\bot(\bot) = \upcl\{\delta_\bot\}$. 
The $\C$ monad presented here has subtle differences to that of \citet{jacobs2008coalgebraic}---it is composed with an error monad to handle $\bot$, and the unit performs an up-closure---but it still upholds the monad laws, shown below, which we prove in \Aref{app:monad-laws}.
\[
  \eta^\dagger = \mathsf{id}
  \qquad\qquad
  f^\dagger \circ \eta = f
  \qquad\qquad
  (f^\dagger \circ g)^\dagger = f^\dagger \circ g^\dagger
\]
It was also shown by \citet{jifeng1997probabilistic} that $f^\dagger$ preserves up-closedness and convexity.

The last ingredient we need (for the semantics of loops) is an order on $\C(\Sigma)$:
\[
  S\sqsubseteq_{\C} T
    \quad\text{iff}\quad
  \forall \nu\in T.\ \exists \mu\in S.\ \mu \sqsubseteq_{\D} \nu
\]
This order, due to \citet{smyth1978power}, is not generally antisymmetric, but in this case it is antisymmetric because the sets in $\C$ are up-closed. In fact, due to up-closure, the Smyth order is equivalent to reverse subset inclusion $S\sqsubseteq_{\C} T$ iff $S \supseteq T$.
The bottom element of $\C(\Sigma)$ in this order is $\eta(\bot) = \mathord\uparrow\{\delta_\bot\} = \D(\Sigma_\bot)$, the set of all distributions. Operationally, this means that nontermination is identified with total uncertainty about the program outcome. As we unroll loops to obtain tighter and tighter approximations of their semantics, we rule out more and more possible behaviors.

In addition, we note that $\langle\C(\Sigma), \sqsubseteq_\C\rangle$ is a directed complete partial order (dcpo), meaning that all increasing chains of elements $S_1 \sqsubseteq_\C S_2 \sqsubseteq_\C \cdots$  have a supremum. Since $S\sqsubseteq_\C T$ is equivalent to $S\supseteq T$, then suprema are given by standard set intersection. So, to show that $\langle\C(\Sigma), \sqsubseteq_\C\rangle$ is a dcpo we need to show that any chain $S_1 \supseteq S_2 \supseteq \cdots$ has a supremum (\ie intersection) in $\C(\Sigma)$.
\citet[Lemma B.4.4]{mciver2005abstraction}, showed that $\D(\Sigma)$ is compact using Tychonoff's Theorem, and therefore it is well known that such a chain has a nonempty intersection. The remaining properties (convexity, closure, up-closure) are well known to be preserved by intersections too.

\begin{figure}[t]
\begin{align*}
\de{\skp}(\sigma) &\triangleq \eta(\sigma)
\\
\de{x \coloneqq e}(\sigma) &\triangleq \eta(\sigma[x\coloneqq \de{e}(\sigma)])
\\
\de{C_1 \fatsemi C_2}(\sigma) &\triangleq \dem{C_2}{{\de{C_1}(\sigma)}}
\\
\de{C_1 \nd C_2}(\sigma) &\triangleq \de{C_1}(\sigma) \nd \de{C_2}(\sigma)
\\
\de{C_1 \oplus_e C_2}(\sigma) &\triangleq \de{C_1}(\sigma) \oplus_{\de{e}(\sigma)} \de{C_2}(\sigma)
\\
\de{\iftf e{C_1}{C_2}}(\sigma) &\triangleq \left\{
  \begin{array}{ll}
    \de{C_1}(\sigma) & \text{if} ~ \de{e}(\sigma) = \tru \\
    \de{C_2}(\sigma) & \text{if} ~ \de{e}(\sigma) = \fls
  \end{array}
\right.
\\
\de{\whl eC}(\sigma) & \triangleq \mathsf{lfp}\left(\Phi_{\langle C, e\rangle}\right)(\sigma)
\\\\
& \hspace{-9em}\text{where}~ \Phi_{\langle C, e\rangle}(f)(\sigma) \triangleq \left\{
  \begin{array}{ll}
    f_\bot^\dagger(\de{C}(\sigma)) & \text{if} ~ \de{e}(\sigma) = \tru \\
    \eta(\sigma) & \text{if} ~ \de{e}(\sigma) = \fls
  \end{array}
  \right.
\end{align*}
\caption{Denotational Semantics for programs $\de{C} \colon \Sigma \to \C(\Sigma)$, where $\de{e} \colon \Sigma\to\mathsf{Val}$ is the interpretation of expressions, defined in the obvious way.}
\label{fig:semantics}
\end{figure}

\subsection{Semantics of Sequential Commands}
\label{sec:sem-seq}

We are now ready to define the semantics, shown in \Cref{fig:semantics}. We interpret commands denotationally as maps from states to convex sets of distributions, \ie $\de{C} \colon \Sigma \to \C(\Sigma)$. No-ops and variable assignment are defined as point-mass distributions. Sequential composition is a Klesili composition.
The probabilistic and nondeterministic choice operations are defined in terms of new operators:
\[
S \oplus_p T \triangleq \{ p\cdot \mu + (1-p)\cdot \nu \mid \mu \in S, \nu\in T \}
\qquad\qquad
S \nd T \triangleq \textstyle{\bigcup_{p\in [0,1]}} S\oplus_p T
\]
As expected, probabilistic branching chooses an element of $\de{C_1}(\sigma)$ with probability $p = \de{e}(\sigma)$, and chooses an element of $\de{C_2}(\sigma)$ with probability $1-p$. Nondeterministic choices are equivalent to a union of all the possible probabilistic choices between $C_1$ and $C_2$. If we think of nondeterminism being resolved by a scheduler, this operationally corresponds to the scheduler picking a bias $p$ (which could be 0 or 1, corresponding to certainty), then flipping a coin with bias $p$ to decide which command to execute \cite{segala1994probabilistic,varacca2002powerdomain}.

\citet{jifeng1997probabilistic} showed that all of these operations preserve up-closedness and convexity and \citet{wpe} showed that they preserve topological closure  (\citeauthor{wpe} refer to this as Cauchy Closure) and non-emptiness.

Conditional statements are defined in the standard way. A branch is taken deterministically depending on whether the guard $e$ evaluates to true or false.
As syntactic sugar,
we define special syntax for biased coin flips and 
nondeterministic choice from a nonempty finite set $S = \{ v_1, \ldots, v_n\}$:
\[
  x \coloneqq \mathsf{flip}(e)
  \;\triangleq\;
  (x \coloneqq \tru) \oplus_e (x\coloneqq \fls)
\qquad\qquad
  x \gets S
  \;\triangleq\;
  (x \coloneqq v_1) \nd \cdots \nd (x \coloneqq v_n)  
\]

\subsection{Semantics of Loops}
\label{sec:fixpoint}

Loops are interpreted as the least fixed point of $\Phi_{\langle C, e\rangle}$ (see \Cref{fig:semantics}), which essentially means that:
\[
\de{\whl eC} = \de{\ift e (C \fatsemi \whl eC)}
\]
%
We will use the Kleene fixed point theorem to prove that a least fixed point exists. To do so, we first define an ordering on functions $\mathord{\sqsubseteq_\C^\bullet} \subseteq (\Sigma \to \C(\Sigma)) \times (\Sigma \to \C(\Sigma))$, which is the pointwise extension of the order $\sqsubseteq_\C$ from \Cref{sec:prelim} and is defined as follows:
\[
  f \sqsubseteq_\C^\bullet g
  \qquad\text{iff}\qquad
  \forall \sigma\in \Sigma.\ f(\sigma) \sqsubseteq_\C g(\sigma)
  \qquad\text{iff}\qquad
  \forall \sigma\in \Sigma.\ f(\sigma) \supseteq g(\sigma)
\]
Clearly, the function $\bot_\C^\bullet(\sigma) \triangleq \eta(\bot)$ is the bottom of this order, since $\eta(\bot)$ is the bottom of $\sqsubseteq_\C$. The characteristic function $\Phi_{\langle C, e\rangle}$ is also Scott continuous in this order, meaning that it preserves suprema of directed sets \Acref{lem:scott-continuity}:
\[
  \sup_{f \in D} \Phi_{\langle C, e\rangle}(f) = \Phi_{\langle C, e\rangle}(\sup D)
\]
So, by the Kleene fixed point theorem, we conclude that the least fixed point exists, and is characterized as the supremum of the iterates of $\Phi_{\langle C, e\rangle}$ over all the natural numbers.
\[
  \mathsf{lfp}\left(\Phi_{\langle C, e\rangle}\right)(\sigma)
  =
  \left(\sup_{n\in \mathbb N}\Phi_{\langle C, e\rangle}^n\left(\bot_\C^\bullet\right)\right)(\sigma)
  = \bigcap_{n\in \mathbb N}\Phi_{\langle C, e\rangle}^n\left(\bot_\C^\bullet\right)(\sigma)
\]
These iterates are defined as $f^0 \triangleq \mathsf{id}$ and $f^{n+1} \triangleq f\circ f^n$, where $\circ$ is function composition.

\section{Demonic Outcome Logic}
\label{sec:logic}

We now present Demonic Outcome Logic, a new logic for reasoning about programs that are both randomized and nondeterministic. This logic has constructs for reasoning about probabilistic branching, inspired by Outcome Logic (OL) \cite{outcome} and probabilistic Hoare logics \cite{ellora,hartog}. In addition, nondeterminism is treated {\em demonically}: the postcondition must hold regardless of how the nondeterminism is resolved.

\subsection{Outcome Assertions}
\label{sec:assertions}

Outcome assertions $\varphi,\psi\in \prop$  are used in pre- and postconditions of triples in Demonic Outcome Logic. The syntax is shown below, where $p\in[0,1]$ is a probability, and $P,Q \in\mathsf{Atom}$ are atoms.
\begin{align*}
\mathsf{Prop} \ni \varphi &\Coloneqq
 \top \mid \bot 
 \mid \varphi\land \psi
\mid \varphi \oplus_p \psi 
\mid \varphi\nd\psi
\mid \sure P
&
p\in[0,1]
\\
\mathsf{Atom} \ni P &\Coloneqq
 \tru \mid \fls \mid P \land Q \mid P\lor Q \mid \lnot P \mid e_1 = e_2 \mid e_1 \le e_2 \mid \cdots
\end{align*}
Atomic assertions $P,Q\in\mathsf{Atom}$ describe states, and are interpreted using $\sem- \colon \mathsf{Atom} \to \bb{2}^\Sigma$, giving the set of states satisfying $P$, defined as usual. The satisfaction relation $\mathord\vDash \subseteq \D(\Sigma_\bot) \times \prop$ defined in \Cref{fig:assertions} relates each assertion to probability distributions $\mu \in \mathcal \D(\Sigma_\bot)$ (not \emph{sets} distributions). As explained in \Cref{sec:triples}, this corresponds to the logic's demonic treatment of nondeterminism.
\begin{figure}
\[
\begin{array}{lcl}
  \mu\vDash\top & \multicolumn{2}{l}{\text{always}}
  \\
  \mu\vDash\bot & \multicolumn{2}{l}{\text{never}}
  \\
  \mu\vDash \varphi\land\psi
  & \text{iff} &
  \mu\vDash\varphi \quad\text{and}\quad
  \mu\vDash\psi
  \\
  \mu \vDash \varphi\oplus_p \psi & \text{iff} & \exists \mu_1, \mu_2.\quad \mu = p\cdot\mu_1 + (1-p)\cdot \mu_2
    \quad\text{and}\quad  \mu_1 \vDash \varphi
    \quad\text{and}\quad  \mu_2\vDash \psi
  \\
  \mu\vDash\varphi\nd\psi
    &\text{iff}&
     \mu \vDash \varphi \oplus_p \psi\quad\text{for some}~ p\in[0,1]
  \\
  \mu\vDash \sure P & \text{iff} &
     \supp(\mu) \subseteq \sem P
\end{array}
\]
\caption{Definition of the satisfaction relation $\mathord\vDash \subseteq \D(\Sigma_\bot) \times \prop$ for Outcome Assertions.}
\label{fig:assertions}
\end{figure}

As expected, $\top$ is satisfied by any distribution, whereas $\bot$ is satisfied by nothing. The logical conjunction $\varphi\land\psi$ is true iff both conjuncts are true.
If a distribution $\mu$ satisfies the probabilistic outcome conjunction $\varphi \oplus_p\psi$, then $\mu$ must be a convex combination with parameter $p$ of a distribution satisfying $\varphi$ and one satisfying $\psi$. Similarly, $\mu\vDash\varphi\nd\psi$ means that $\mu$ is some convex combination (where the parameter is existentially quantified) of distributions satisfying $\varphi$ and $\psi$. Finally, a distribution satisfies $\sure P$ if its support is contained in $\sem P$. Note that $P$ can only describe states (and not $\bot$), so $\mu\vDash \sure P$ implies that $\mu(\bot) = 0$, \ie that the program that generated $\mu$ almost surely terminated. This is a crucial difference between $\tru$ and $\top$; whereas $\tru$ guarantees almost sure termination, $\top$ is satisfied by any distribution.

As an example, $\sure{x=1} \oplus_{\frac13} \sure{x=2}$ means that the event $x=1$ occurs with probability $\frac13$ and the event $x=2$ occurs with probability $\frac23$. On the other hand, $\sure{x=1}\nd \sure{x=2}$ means that $x=1$ occurs with some probability $p$ and $x=2$ occurs with probability $1-p$.
Given that we represent nondeterminism as convex union, $\varphi\nd \psi$ characterizes nondeterministic choice.
In addition, we can \emph{forget} about the probabilities of outcomes by weakening $\varphi\oplus_p\psi \Rightarrow \varphi\nd\psi$.
As a shorthand, we will often write ${\textstyle\bignd_{k=1}^n} \varphi_k$ instead of $\varphi_1 \nd \cdots \nd \varphi_n$ for finite $\nd$ conjunctions of assertions.
Unlike in standard Outcome Logic \cite{outcome}, $\varphi\nd\psi$ does not imply that both $\varphi$ and $\psi$ are realizable via an actual trace; for example if $\mu\vDash \sure{x=1} \nd \sure{x=2}$, it is possible that the event $x=1$ occurs with probability 0 according to $\mu$. This is an intentional choice, as it allows us to retain desirable propositional properties such as idempotence of $\nd$, as explained in \Cref{sec:equations}.

Echoing the equational laws from \Cref{sec:equations}, outcome assertions can be manipulated using the following implications, where $\varphi \Rightarrow \psi$ means that if $\mu\vDash \varphi$ then $\mu\vDash\psi$ for all $\mu\in\D(\Sigma_\bot)$. These implications are not included in the syntax of $\prop$, since they are not allowed to be used in the pre- and postconditions of triples.
\begin{align*}
  \varphi\nd\varphi & \Leftrightarrow \varphi &
  \varphi\oplus_p \varphi & \Leftrightarrow \varphi &&
  \textsf{(Idempotence)}
  \\
  \varphi \nd \psi & \Leftrightarrow \psi \nd \varphi &
    \varphi \oplus_p \psi &\Leftrightarrow \psi \oplus_{1-p} \varphi &&
  \textsf{(Commutativity)}
  \\
  \varphi \nd (\psi \nd\vartheta) & \Leftrightarrow (\varphi \nd\psi) \nd \vartheta &
  \quad(\varphi \oplus_p \psi) \oplus_q \vartheta &\Leftrightarrow \varphi \oplus_{pq} (\psi \oplus_{\frac{(1-p)q}{1-pq}} \vartheta) &&
  \textsf{(Associativity)}
  \\
  \span\varphi \oplus_p ( \psi \nd \vartheta) \Leftrightarrow (\varphi \oplus_p \psi) \nd (\varphi \oplus_p \vartheta) \span\span
  &&
  \textsf{(Distributivity)}
\end{align*}
We remark that these laws depend on what is---and, crucially, what is \emph{not}---included in the assertion language. As we saw in \Cref{sec:equations}, idempotence is delicate due to the fact that assertions are only \emph{approximations} of the distributions that they model. Despite this, idempotence turns out to be crucial to the usability of the logic, as the soundness of several inference rules depends on it (\eg \ruleref{Nondet} and \ruleref{Bounded Variant}). 
The idempotence laws would be invalidated if the syntax included disjunctions or existential quantification. For example, in the following implication, $x$ having value 1 or 2 each with probability $\frac12$ does not imply that $x$ is always 1 or always 2.
\[
  \sure{x = 1} \oplus_{\frac12} \sure{x = 2}
  \quad\implies\quad
  (\sure{x=1} \vee \sure{x=2}) \oplus_{\frac12} (\sure{x=1} \vee \sure{x=2})
  \quad\centernot\implies\quad
  \sure{x=1} \vee \sure{x=2}
\]
Note that the first implication is valid since $\sure{x=1} \Rightarrow \sure{x=1} \vee \sure{x=2}$ and weakening can be applied inside of $\oplus_p$ as follows: $\sure P\Rightarrow \sure{P'} \vdash \sure P\oplus_p \sure Q \Rightarrow \sure{P'} \oplus_p \sure{Q}$.

We do not believe that the exclusion of disjunctions and existential quantification poses a severe restriction in practice. Existentials are often used to quantify over the values of certain program variables; in Demonic Outcome Logic, we quantify over values in a different way. In a typical logic, pre- and postconditions are predicates over individual states, so $\exists v:T. x=v$ asserts that the value of $x$ takes on some value from the set $T$. In our case, we use predicates over distributions, so it is more appropriate to say $\bignd_{v\in T} \sure{x=v}$, which asserts that the value of $x$ is in $T$ for every state in the support of the distribution. We use this technique in \Cref{sec:von-neumann}.

\subsection{Semantics of Triples}
\label{sec:triples}

Similar to Hoare Logic and Outcome Logic \cite{outcome}, specifications in Demonic Outcome Logic are triples of the form $\triple\varphi{C}\psi$. Intuitively, the semantics of these triples is that if states are initially distributed according to a distribution $\mu \in \D(\Sigma_\bot)$ that satisfies $\varphi$, then after the program $C$ is run, the resulting states will be distributed according to some distribution $\nu \in \D(\Sigma_\bot)$ that satisfies $\psi$, 
regardless of how nondeterministic choices in $C$ are resolved. We formalize the semantics of triples, denoted by $\vDash\triple\varphi{C}\psi$, as follows.
\begin{definition}[Semantics of Demonic Outcome Triples]
\[
\vDash\triple{\varphi}C\psi
\qquad\text{iff}\qquad
\forall \mu \in \D(\Sigma_\bot).\quad
\mu \vDash\varphi
\quad\implies\quad
\forall \nu\in \dem{C}{\upset{\mu}}.\quad \nu\vDash\psi
\]
\end{definition}
We note that when limited to basic assertions $P,Q\in\mathsf{Atom}$, $\vDash\triple{\sure P}C{\sure Q}$ is semantically equivalent to a total correctness Hoare triple \cite{manna1974axiomatic} (albeit, in a language with randomization). That is, for any start state $\sigma\in\sem P$, the program will terminate in a state $\tau \in\sem Q$.

\subsection{Inference Rules}
\label{sec:rules}

The inference rules for reasoning about non-looping commands are shown in \Cref{fig:rules} (we will revisit loops in \Cref{sec:loops}). We write $\vdash\triple\varphi C\psi$ to mean that a triple is derivable using these rules. This relates to the semantics of triples via the following soundness theorem, which is proved by induction on the derivation \Acite{app:soundness}.
\begin{restatable}[Soundness]{theorem}{thmsoundness}
\label{thm:soundness}
\(
  \vdash\triple{\varphi}C{\psi}
  \qquad\implies\qquad
  \vDash\triple{\varphi}C{\psi}
\)
\end{restatable}
Now, we will describe the rules in more depth.

\subsubsection*{Sequential and Probabilistic Commands}
Many of the rules for analyzing commands are as expected. The \ruleref{Skip} rule simply preserves the precondition, as no-ops do not affect the distribution of outcomes. The \ruleref{Assign} rule uses standard backward substitution, where $\varphi[e/x]$ is the assertion obtained by syntactically substituting $e$ for all occurrences of $x$.
The \ruleref{Seq} rule allows us to analyze sequences of commands compositionally, and relies on the fact that $\dem CS = \bigcup_{\mu\in S}\dem C{\upset{\mu}}$ for soundness.

In order to analyze a probabilistic choice $C_1 \oplus_e C_2$ using \ruleref{Prob}, the precondition $\varphi$ must contain enough information to ascertain that the expression $e$ evaluates to a precise probability $p\in [0,1]$. If $e$ is a literal $p$, then this restriction is trivial since $\varphi\Rightarrow \sure{p=p}$ for any $\varphi$. The postcondition then joins the outcomes of the two branches using $\oplus_p$.
Similarly, the rules for analyzing if statements require that the precondition selects one of the two branches deterministically. \ruleref{If1} applies when the precondition forces the true branch to be taken and \ruleref{If2} applies when it forces the false branch. We will soon see derived rules that allow both branches to be analyzed in a single rule.

\subsubsection*{Structural Rules}
The bottom of \Cref{fig:rules} also contains structural rules, which do not depend on the program command. The \ruleref{Prob Split} and \ruleref{ND Split} rules allow us to deconstruct pre- and postconditions in order to build derivations compositionally. As we will see shortly, these rules are necessary to analyze nondeterministic choices, since the \ruleref{Nondet} rule requires the precondition to be a basic assertion $\sure P$. They are also useful for analyzing if statements, since the \ruleref{If1} and \ruleref{If2} rules require the precondition to imply the truth and falsity of the guard, respectively. The soundness of these rules relies on the following equality \Acref{lem:oplus_dist}.
\[
  \dem{C}{S_1 \oplus_p S_2}
  \quad=\quad
  \dem{C}{S_1} \oplus_p \dem{C}{S_2}
\]
Next, we have the usual rule of \ruleref{Consequence}, which allows the precondition to be strengthened and the postcondition to be weakened. These implications are semantic ones; we do not provide proof rules to dispatch them beyond the laws at the end of \Cref{sec:assertions}.

Finally, the rule of \ruleref{Constancy} allows us to conjoin additional information $P$ about the program state, so long as it does not involve any of the modified program variables. We let $\mathsf{mod}(C)$ denote the set of variables modified by the program $C$, defined inductively on its structure. One subtlety is that possibly nonterminating programs must be considered to modify all the program variables, meaning that \ruleref{Constancy} only applies to terminating programs. However, this restriction does not matter much in practice, since all the loop rules we present in \Cref{sec:loops} guarantee almost sure termination. In addition $\mathsf{fv}(P)$ is the set of variables occurring free in $P$.
 Just like the frame rule from Outcome Separation Logic, $P$ is a basic assertion rather than an outcome assertion, since this extra information concerns only the local state and not the branching behavior of the program \cite{zilberstein2024outcome}. The soundness of \ruleref{Constancy} is considerably simpler than that of the frame rule, since it does not deal with dynamically allocated pointers and aliasing.
 
 \begin{figure}
\begin{flushleft}{\textit{\sffamily Commands}}\end{flushleft}
\[
\ruledef{Skip}{\;}{\triple{\varphi}{\skp}{\varphi}}
\qquad\quad
\ruledef{Assign}{\;}{
  \triple{\varphi[e/x]}{x \coloneqq e}{\varphi}
}
\qquad\quad
\ruledef{Seq}{
  \triple{\varphi}{C_1}{\vartheta}
  \\
  \triple{\vartheta}{C_2}{\psi}
}{
  \triple{\varphi}{C_1\fatsemi C_2}{\psi}
}
\]
\medskip
\[
\ruledef{Prob}{
  \varphi \Rightarrow \sure{e = p}
  \quad
  \triple{\varphi}{C_1}{\psi_1}
  \quad
  \triple{\varphi}{C_2}{\psi_2}
}{
  \triple{\varphi}{C_1 \oplus_e C_2}{\psi_1 \oplus_p \psi_2}
}
\qquad
\ruledef{Nondet}{
  \triple{\sure P}{C_1}{\psi_1}
  \quad
  \triple{\sure P}{C_2}{\psi_2}
}{
  \triple{\sure P}{C_1 \nd C_2}{\psi_1 \nd \psi_2}
}
\]
\medskip
\[
\ruledef{If1}{
  \varphi \Rightarrow \sure e
  \\
  \triple{\varphi}{C_1}{\psi}
}{
  \triple{\varphi}{\iftf e{C_1}{C_2}}{\psi}
}
\qquad\qquad\qquad
\ruledef{If2}{
  \varphi \Rightarrow \sure{\lnot e}
  \\
  \triple{\varphi}{C_2}{\psi}
}{
  \triple{\varphi}{\iftf e{C_1}{C_2}}{\psi}
}
\]
\smallskip
\begin{flushleft}{\textit{\sffamily Structural Rules}}\end{flushleft}
\medskip
\[
\inferrule{
  \triple{\varphi_1}C{\psi_1}
  \\
  \triple{\varphi_2}C{\psi_2}
}{
  \triple{\varphi_1 \oplus_p \varphi_2}C{\psi_1\oplus_p\psi_2}
}{\rulename{Prob Split}}
\qquad\qquad
\ruledef{ND Split}{
  \triple{\varphi_1}C{\psi_1}
  \\
  \triple{\varphi_2}C{\psi_2}
}{
  \triple{\varphi_1 \nd \varphi_2}C{\psi_1\nd\psi_2}
}
\]
\medskip
\[
\ruledef{Consequence}{
  \varphi' \Rightarrow \varphi
  \quad
  \triple{\varphi}C\psi
  \quad
  \psi \Rightarrow \psi'
}{
  \triple{\varphi'}C{\psi'}
}
\qquad
\ruledef{Constancy}{
  \triple{\varphi}C\psi
  \quad
  \mathsf{mod}(C) \cap \mathsf{fv}(P) = \emptyset
}{
  \triple{\varphi\land \sure P}C{\psi\land \sure P}
}
\]
\caption{Inference rules for non-looping commands in Demonic Outcome Logic.}
\label{fig:rules}
\end{figure}
 
\subsubsection*{Nondeterministic Branching}

The \ruleref{Nondet} rule can only be applied if the precondition is a basic assertion $P$. If the precondition contained information about probabilistic outcomes, then this rule would be unsound. To demonstrate why this is the case, let us revisit the coin flip game:
\[
  x \coloneqq\mathsf{flip}\left(\tfrac12\right)
  \fatsemi
  y \gets \Bool
\]
If we imagine that nondeterminism is controlled by an adversary, then it is always possible for the adversary to \emph{guess} the coin flip, that is, to force $x$ and $y$ to be equal. However, if we allowed the precondition in the \ruleref{Nondet} rule to contain probabilistic outcomes, then we could derive the triple below, stating that $x=y$ always occurs with probability $\frac12$, which is untrue.
\[\footnotesize
\hspace{-6em}
\inferrule*[right=\ruleref{Consequence}]{
  \inferrule*[Right=\ruleref{Nondet}\textnormal{ (incorrect usage)}]{
      \inferrule*{\vdots}{
        \triple{\sure{x = \tru} \oplus_\frac12 \sure{x=\fls}}{y \coloneqq \tru}{\sure{x = y} \oplus_\frac12 \sure{x\neq y}}
      }
      \quad
      \inferrule*[vdots=2.5em,leftskip=5em,rightskip=15em]{\vdots}{
        \triple{\sure{x = \tru} \oplus_\frac12 \sure{x=\fls}}{y \coloneqq \fls}{\sure{x \neq y} \oplus_\frac12 \sure{x= y}}
      }
  }{
    \triple{\sure{x = \tru} \oplus_\frac12 \sure{x=\fls}}{y\gets\Bool}{(\sure{x=y} \oplus_\frac12 \sure{x\neq y}) \nd (\sure{x\neq y} \oplus_\frac12 \sure{x= y})}
  }
}{
  \triple{\sure{x = \tru} \oplus_\frac12 \sure{x=\fls}}{y\gets\Bool}{\sure{x=y} \oplus_\frac12 \sure{x\neq y}}
}
\]
Instead---as shown below---we must de-structure the precondition using \ruleref{Prob Split}, and then apply \ruleref{Nondet} twice, using each of the basic assertions ($\sure{x=\tru}$ and $\sure{x=\fls}$) as preconditions. This has the effect of expanding each basic outcome inside of the $\oplus_\frac12$.
After applying idempotence in the postcondition, we see that $x=y$ occurs with unknown probability, which does not preclude that the adversary could force it to occur. In fact, the postcondition is equivalent to $\tru$.
\[\footnotesize
\hspace{2em}
\inferrule*[right=\ruleref{Consequence}]{
  \inferrule*[Right=\ruleref{Prob Split},leftskip=3em,rightskip=3em]{
      \inferrule*[right=\ruleref{Nondet}]{\vdots
      }{
        \triple{\sure{x = \tru}}{y \gets \Bool}{\sure{x = y} \nd \sure{x\neq y}}
      }
      \qquad
      \inferrule*[Right=\ruleref{Nondet}]{\vdots
      }{
        \triple{\sure{x = \fls}}{y \gets \Bool}{\sure{x \neq y} \nd \sure{x= y}}
      }
  }{
    \triple{\sure{x = \tru} \oplus_\frac12 \sure{x=\fls}}{y\gets\Bool}{(\sure{x=y} \nd \sure{x\neq y}) \oplus_\frac12 (\sure{x\neq y} \nd \sure{x= y})}
  }
}{
  \triple{\sure{x = \tru} \oplus_\frac12 \sure{x=\fls}}{y\gets\Bool}{\sure{x=y} \nd \sure{x\neq y}}
}
\]
Requiring basic assertions as the precondition may seem restrictive, but the \ruleref{Nondet} rule can still be applied in all scenarios by deconstructing the precondition, as we saw in \Cref{sec:composition} \Acite{app:coin}.
The soundness of \ruleref{Nondet} fundamentally depends on idempotence. In the proof, we show that if $\mu\vDash P$, then $\nu \vDash\psi_1 \nd \psi_2$ for each $\nu \in \de{C_1\nd C_2}(\sigma)$ and $\sigma\in \supp(\mu)$.
Any distribution in $\dem{C_1\nd C_2}{\upset\mu}$ is therefore a convex combination of distributions, all of which satisfy $\psi_1\nd\psi_2$. We collapse that convex combination using, \eg $(\psi_1 \nd \psi_2) \oplus_{p} (\psi_1 \nd \psi_2)\Rightarrow \psi_1 \nd \psi_2$ (see \Cref{lem:convex_sum_idem} for the more general property).
Soundness must be established in this way, since $\dem{C_1\nd C_2}S \neq \dem{C_1}S \nd \dem{C_2}S$ and so:
\[
  \dem{C_1}S\vDash \psi_1
  \quad\text{and}\quad
  \dem{C_2}S\vDash\psi_2
  \qquad\centernot\Rightarrow\qquad
  \dem{C_1\nd C_2}S \vDash \psi_1 \nd \psi_2
\]
On the other hand, $\de{C_1\nd C_2}(\sigma) = \de{C_1}(\sigma) \nd\de{C_2}(\sigma)$, so we can analyze nondeterministic branching \emph{compositionally} only when starting from a single state. 

\subsubsection*{Derived Rules}

In addition to the rules in \Cref{fig:rules}, we also provide some derived rules for convenience in common scenarios. All the derivations are shown in \Aref{app:derived}.
The first rules pertain to conditional statements. These rules use \ruleref{Prob Split} and \ruleref{ND Split} to deconstruct the pre- and postconditions in order to analyze both branches of the conditional statement. These are similar to the rules found in Ellora \cite{ellora} and Outcome Logic \cite{zilberstein2024relatively}.
\[
\inferrule{
  \substack{
    \varphi_1 \Rightarrow \sure e
    \\
    \varphi_2 \Rightarrow \sure{\lnot e}
  }
  \quad
  \triple{\varphi_1}{C_1}{\psi_1}
  \quad
  \triple{\varphi_2}{C_2}{\psi_2}
}{
  \triple{\varphi_1 \oplus_p \varphi_2}{\iftf e{C_1}{C_2}}{\psi_1 \oplus_p \psi_2}
}
\qquad
\inferrule{
  \substack{
    \varphi_1 \Rightarrow \sure e
    \\
    \varphi_2 \Rightarrow \sure{\lnot e}
  }
  \quad
  \triple{\varphi_1}{C_1}{\psi_1}
  \quad
  \triple{\varphi_2}{C_2}{\psi_2}
}{
  \triple{\varphi_1 \nd \varphi_2}{\iftf e{C_1}{C_2}}{\psi_1 \nd \psi_2}
}
\]
In addition, if the precondition of the conditional statement is a basic assertion $P$, then we can use the typical conditional rule from Hoare logic. This relies on the Hahn decomposition theorem: $\sure{P} \Rightarrow \sure{P\land e} \nd \sure{P\land\lnot e}$, that is, if $\mu\vDash P$, then $\mu$ can be separated into two portions where $e$ is true and false, respectively. Due to idempotence, we can simplify the postconditions using the consequence $\psi\nd\psi \Rightarrow \psi$.
\[
\inferrule{
  \triple{\sure{P\land e}}{C_1}{\psi}
  \\
  \triple{\sure{P\land\lnot e}}{C_2}{\psi}
}{
  \triple{\sure P}{\iftf e{C_1}{C_2}}{\psi}
}
\]
Finally, we provide rules for analyzing the coin flip and nondeterministic selection syntactic sugar introduced in \Cref{sec:sem-seq}. The flip rule is derived using a straightforward application of \ruleref{Prob}. The rule for nondeterministic selection is proven by induction on the size of $S$ (recall that $S$ is finite).
\[
  \inferrule{
    \varphi\Rightarrow \sure{e=p}
    \\
    x \notin \mathsf{fv}(\varphi)
  }{
    \triple{\varphi}{x \coloneqq \mathsf{flip}(e)}{\varphi \land (\sure{x = \tru} \oplus_p \sure{x=\fls})}
  }
  \qquad\quad
  \inferrule{\;}{\triple{\sure\tru}{x\gets S}{\textstyle\bignd_{v\in S} \sure{x=v}}}
\]
Although the precondition of the nondeterministic selection rule is $\tru$, it can be used in conjunction with the rule of \ruleref{Constancy} so that any basic assertion $P$ can be the precondition. Beyond that, to extend to any precondition $\varphi$, de-structuring rules must be applied just like with the \ruleref{Nondet} rule.

\begin{remark}[Completeness]
We have not explored completeness of Demonic Outcome Logic, even for the loop-free fragment. One reason for this is that the derivations witnessed by the completeness proofs for similar probabilistic logics do not mimic the sort of derivations that one would produce by hand. For example, in Ellora \cite{ellora}, the completeness proof involves quantifying over the (infinitely many) distributions that could satisfy the precondition, then showing that a derivation of the strongest postcondition is possible given any fixed one of those distributions. The complexity comes not only from loops, but also purely sequential constructs like probabilistic branching and if statements. To see examples of where this complexity arises, see \citet[Definition 3.5]{zilberstein2024relatively}, or \citet[Example 1]{hyperhoare}.
\end{remark}

\section{Analyzing Loops}
\label{sec:loops}

In this section, we discuss proof rules for analyzing loops. We are inspired by work on weakest pre-expectations \cite{mciver2005abstraction,kaminski,mciver2018new}, where probabilistic loop analysis has been studied extensively, but will argue in this section that our program logic approach has two advantages.

\subsubsection*{Fewer Conditions to Check} The weakest pre-expectation proof rules involve multiple checks, which include both \emph{sub}-invariants and \emph{super}-invariants and computing expected values of ranking functions. In contrast, in Demonic OL all the proof rules revolve around just one construct---outcome triples---and the premises of the rules can accordingly be consolidated.

\subsubsection*{Multiple Outcomes} It is often useful to specify programs in terms of their distinct outcomes, which we achieve using the assertions from \Cref{sec:assertions}. Pre-expectation calculi can only represent multiple outcomes by carrying out several distinct derivations, whereas Demonic OL can do so in one shot.

\subsection{Almost Sure Termination}

As we mentioned in \Cref{sec:fixpoint}, our semantics based on the Smyth powerdomain is suitable for total correctness---specifications implying that the program terminates.
Since we are in a probabilistic setting, it makes sense to talk about a finer notion of termination---\emph{almost sure termination}---meaning that the program terminates with probability 1. In terms of our program semantics, a program almost surely terminates if $\bot$ does not appear in the support of any of its resultant distributions.

\begin{definition}[Almost Sure Termination]
A program $C$ \emph{almost surely terminates} on input $\sigma$ iff
\[
  \forall \mu \in \de{C}(\sigma).~\bot \notin \supp(\mu)
\]
In addition, $C$ \emph{universally almost surely terminates} if it almost surely terminates on all $\sigma\in \Sigma$.
\end{definition}

Going further, we show how almost sure termination is established in Demonic Outcome Logic.

\begin{theorem}\label{thm:ast}
A program $C$ almost surely terminates starting from any state satisfying $P$ if:
\[
  \vDash\triple{\sure P}C{\sure\tru}
\]
As a corollary, $C$ universally almost surely terminates if $\vDash\triple{\sure\tru}C{\sure\tru}$.
\end{theorem}
\begin{proof}
The triple $\vDash\triple{\sure P}C{\sure\tru}$ means that if $\sigma\in\sem P$, that is, if $\delta_\sigma\vDash\sure P$, then $\supp(\mu) \subseteq \sem\tru = \Sigma$ for all $\mu\in\de{C}(\sigma)$. Since $\bot\notin \Sigma$, $\supp(\mu) \subseteq \Sigma$ iff $\bot\notin\supp(\mu)$.
\end{proof}
Following from \Cref{thm:ast}, if $\triple{\sure P}C{\psi}$ holds, then $C$ almost surely terminates as long as $\psi \Rightarrow \sure\tru$, which is simple to check in many cases. For example, if $\psi$ is formed as a collection of atoms $Q_1, \ldots, Q_n$ joined by $\nd$ and $\oplus_p$ connectives, then the program almost surely terminates since $\sure{Q_i} \Rightarrow \sure\tru$ holds trivially, and connectives can be collapsed using idempotence of $\nd$ and $\oplus_p$.

\subsection{The Zero-One Law}

\citet{mciver2005abstraction} showed that under certain conditions, probabilistic programs must terminate with probability either 0 or 1. In this circumstance, almost sure termination can be established simply by showing that the program terminates with nonzero probability.

The original rule of \citet[\S2.6]{mciver2005abstraction} used a propositional invariant $P$ to describe all reachable states after each iteration of the loop. We generalize their rule by using an outcome assertion $\varphi$ as the invariant, so that in addition to describing which states are reachable, we can also describe how those reachable states are distributed. Our version of the rule is stated below.
\[
\ruledef{Zero-One}{
  \varphi \Rightarrow \sure{e}
  \quad
  \psi\Rightarrow \sure{\lnot e}
  \quad
  \triple{\varphi}C{\varphi \nd \psi}
  \quad
  \triple{\varphi}{\whl eC}{\sure{\lnot e} \oplus_p \top}
  \quad
  p >0
}{
  \triple{\varphi}{\whl eC}{\psi}
}
\]
Given some loop $\whl eC$, the first step of the \ruleref{Zero-One} law is to come up with an invariant pair $\varphi$ and $\psi$, where $\varphi$ represents the distribution of states where the guard $e$ remains true and $\psi$ represents the distribution of states in which the loop has terminated. More precisely, $\varphi\Rightarrow\sure{e}$ and $\psi\Rightarrow\sure{\lnot e}$.
Next, we must prove that this is an invariant pair, by proving the following triple.
\begin{equation}\label{eq:invariant}
  \triple{\varphi}C{\varphi\nd\psi}
\end{equation}
That is, if the initial states are distributed according to $\varphi$, then loop body $C$ will reestablish $\varphi$ with some probability $q$ and will terminate (in $\psi$) with probability $1-q$. In the limit, $\psi$ will describe the entire distribution of \emph{terminating} outcomes \Acref{lem:invariant}, although we do not yet know if the loop almost surely terminates.
We can establish almost sure termination if $\varphi$ guarantees some nonzero probability of termination, as represented by the following triple.
\begin{equation}\label{eq:term}
  \triple\varphi{\whl eC}{\sure{\lnot e}\oplus_p \top}
  \qquad\text{where}\qquad
  p>0
\end{equation}
That is, for every $\mu\vDash\varphi$, the loop will always terminate (with $\lnot e$ holding) with probability at least $p>0$. From (\ref{eq:invariant}), we know that on the $i^\text{th}$ iteration of the loop, there is some probability $q_i$ of reestablishing the invariant $\varphi$, in which case the loop will continue to execute. This means that the total probability of nontermination is the product of all the $q_i$.
In general, it is possible for such an infinite product to converge to a nonzero probability, however, from (\ref{eq:term}) we know that every tail of that product must be at most $1-p$, which is strictly less than 1.
\[
  \PP[\text{nonterm}]
  \quad=\quad
  q_1 \times q_2 \times q_3 \times \cdots \times \underbrace{ q_n \times q_{n+1}\times \cdots }_{\le 1-p}
  \quad=\quad
  0
\]
This implies that the product of the $q_i$ must go to $0$, so that the loop almost surely terminates.
The full soundness proof is available in \Aref{app:zero-one}.

\subsection{Proving Termination with Variants and Ranking Functions}
\label{sec:rank}

Although the \ruleref{Zero-One} law provides a means for proving almost sure termination, it can be difficult to use directly, because one still must establish a {\em minimum probability of termination}. In this section, we provide some inference rules derived from \ruleref{Zero-One} that are easier to apply.
The first rule uses a bounded family of \emph{variants} $(\varphi_n)_{n=0}^N$, where $\varphi_0 \Rightarrow \lnot e$ and $\varphi_n\Rightarrow e$ for $1 \le n\le N$. The index $n$ can be thought of as a \emph{rank}, so that we get closer to termination as $n$ descends towards zero.
The premise of the rule is that the rank must decrease by at least 1 with probability $p > 0$ on each iteration. We represent this formally using the following triple.
\[
  \triple{\varphi_n}C{\textstyle(\bignd_{k=0}^{n-1} \varphi_k) \oplus_p (\bignd_{k=0}^N \varphi_k)}
\]
The assertion $\bignd_{k=0}^{n-1} \varphi_k$ is an aggregation of all the variants with rank strictly lower than $n$, essentially meaning that the states must be distributed according to those variants, but without specifying their relative probabilities. Note that, \eg $\varphi_{n-1}\Rightarrow \bignd_{k=0}^{n-1} \varphi_k$, so it is possible to establish this assertion if the rank always decreases by exactly 1. So, the postcondition states that with probability $p$ the rank decreases by at least 1. That means that $\bignd_{k=0}^N \varphi_k$ must hold with probability $1-p$, meaning that the rank can \emph{increase} too, as long as that increase is not too likely.

In order to establish a minimum termination probability, we note that starting at $\varphi_n$, it takes at most $N$ steps to reach rank 0, therefore the loop terminates with probability at least $p^N$, which is greater than 0 since $p>0$ and $N$ is finite. Putting this all together, we get the following rule.
\[
\ruledef{Bounded Variant}{
  \varphi_0 \Rightarrow \sure{\lnot e}
  \quad
  \forall n \in \{1,\ldots, N\}.
  \;\;
  \varphi_n \Rightarrow \sure{e}
  \quad
  \triple{\varphi_n}C{\textstyle(\bignd_{k=0}^{n-1} \varphi_k) \oplus_p (\bignd_{k=0}^N \varphi_k)}
}{
  \triple{\textstyle\bignd_{k=0}^N \varphi_k}{\whl eC}{\varphi_0}
}
\]
As a special case of the \ruleref{Bounded Variant} rule, we can derive the variant rule of \citet[Lemma 7.5.1]{mciver2005abstraction}. Instead of recording the rank with a family of outcome assertions, we will instead use an integer-valued expression $e_{\mathsf{rank}}$. This can be thought of as a ranking function, since $\de{e_\mathsf{rank}} \colon \Sigma\to\mathbb{Z}$ gives us a rank for each state $\sigma\in\Sigma$. In addition, the propositional invariant $P$ describes the reachable states after each iteration of the loop.
Finally, as long as the invariant holds and the loop guard is true, $e_\mathsf{rank}$ must be bounded between $\ell$ and $h$, in other words $\sure{P\land e}\Rightarrow \sure{\ell\le e_\mathsf{rank}\le h}$. The premise of the rule is that each iteration of the loop must strictly decrease the rank with probability at least $p > 0$. That is, for any $n$:
\[
  \triple{\sure{P \land e \land e_\mathsf{rank} = n}}C{\sure{P \land e_\mathsf{rank} < n} \oplus_p \sure{P}}
\]
Given that the rank is integer-valued and strictly decreasing, it must fall below the lower bound $\ell$ within at most $h-\ell + 1$ steps, at which point $e$ becomes false since $\sure{P\land e} \Rightarrow \sure{\ell \le e_\mathsf{rank}}$. So the loop terminates with probability at least $p^{h-\ell + 1}$, and so by the \ruleref{Zero-One} law, it almost surely terminates.
In \Aref{app:variants}, we show how this rule is derived from \ruleref{Bounded Variant} by letting $\varphi_0 \triangleq \sure{P\land\lnot e}$ and $\varphi_n \triangleq \sure{P\land e\land e_\mathsf{rank} = \ell + n - 1}$ for $1 \le n\le h - \ell + 1$.
The full rule is shown below:
\[
\ruledef{Bounded Rank}{
  \sure{P\land e}\Rightarrow \sure{\ell\le e_\mathsf{rank}\le h}
  \;\;\;
  \forall n.
  \;
  \triple{\sure{P \land e \land e_\mathsf{rank} = n}}C{\sure{P \land e_\mathsf{rank} < n} \oplus_p \sure{P}}
}{
  \triple{\sure P}{\whl eC}{\sure{P\land\lnot e}}
}
\]
As an example application of the rule, recall the following resetting random walk, where the agent moves left with probability $\frac12$, otherwise it resets to a position chosen by an adversary.
\[
\begin{array}{l}
  \code{while}~ x > 0 ~\code{do} \\
  \quad (x \coloneqq x-1) \oplus_{\frac12} (x\gets \{1, \ldots, 5\})
\end{array}
\]
When starting in a state where $0 \le x \le 5$, this program almost surely terminates. Although there are uncountably many nonterminating traces, the probability of nontermination is zero. Even in the worst case in which the adversary always chooses 5, the agent eventually moves left in five consecutive iterations with probability 1. In fact, the program terminates in finite expected time, as it is a Bernoulli process. Using the invariant $P \triangleq 0\le x \le 5$, the ranking function $e_\mathsf{rank} = x$, and the probability $p = \frac12$, the premise of \ruleref{Bounded Rank} is simply the following triple, which is easy to prove using the \ruleref{Prob} and \ruleref{Assign} rules.
\[
  \triple{\sure{0 < x \le 5 \land x = n}}{(x \coloneqq x-1) \oplus_{\frac12} (x\gets \{1, \ldots, 5\})}{\sure{0 \le x < n} \oplus_\frac12 \sure{0 \le x \le 5}}
\]
\citet[Lemma 7.6.1]{mciver2005abstraction} showed that \ruleref{Bounded Rank} is complete for proving almost sure termination if the state space is finite. While we did not assume a finite state space for our language, this result nonetheless shows that the rule is broadly applicable. In addition, our new \ruleref{Bounded Variant} rule is more expressive, as it allows the invariants to have multiple outcomes. We will see an example of how this is useful in \Cref{sec:von-neumann}.

More sophisticated rules are possible in which the rank need not be bounded. One such rule is shown below, based on that of \citet{mciver2018new}. Instead of bounding the rank, we now require that the \emph{expected} rank decreases each iteration, which is guaranteed in our rule by bounding the amount that it can increase in the case that an increase occurs.
\[
\ruledef{Progressing Rank}{
  \triple{\sure{P\land e \land e_{\mathsf{rank}}=k}}C{
     \sure{P\land e_{\mathsf{rank}} \le k-d} \oplus_{p} \sure{P\land e_{\mathsf{rank}} \le k + {\textstyle\frac{p}{1-p}}d}
   }
}{
  \triple{\sure P}{\whl eC}{\sure{P\land\lnot e}}
}
\]
We can use this rule to prove almost sure termination of the following demonically fair random walk, in which the agent steps towards the origin with probability $\frac12$, otherwise an adversary can choose whether or not the adversary steps away from the origin.
\[
\begin{array}{l}
\code{while}~ x>0 ~\code{do} \\
\quad x\coloneqq x-1 \oplus_{\frac12} (x \coloneqq x+1 \nd \skp)
\end{array}
\]
We instantiate \ruleref{Progressing Rank} with $P\triangleq x \ge 0$, $e_\mathsf{rank} \triangleq x$, $p \triangleq \frac12$, and $d \triangleq 1$ to get the following premise, which is easily proven using \ruleref{Prob}, \ruleref{Nondet}, \ruleref{Assign}, and basic propositional reasoning.
\[
  \triple{\sure{x = k > 0}}{x\coloneqq x-1 \oplus_{\frac12} (x \coloneqq x+1 \nd \skp)}{\sure{0 \le x \le k - 1} \oplus_\frac12 \sure{0 \le x \le k + 1}}
\]
The full derivation for the demonic random walk and soundness proof for a more general version of the \ruleref{Progressing Rank} rule---where $p$ and $d$ do not have to be constants---appear in \Aref{app:new-rule}.
Compared to the original version of this rule due to \citet{mciver2018new}, our rule consolidates three premises into just one.

\section{Case Studies}
\label{sec:examples}

In this section, we present three case studies in how our logic can be used to analyze programs that contain both probabilistic and nondeterministic operations.

\subsection{The Monty Hall Problem}
\label{sec:monty}

The Monty Hall problem is a classic paradox in probability theory in which a game show contestant tries to win a car by guessing which door it is behind. The player has three initial choices; the car is behind one door and the other two contain goats. After choosing a door, the host reveals one of the goats among the unopened doors and the player chooses to \emph{stick} with the original door or \emph{switch}---\emph{which strategy is better?}

We model this problem with the $\progname{Game}$ program on the left side of \Cref{fig:monty-hall}. First, the car is randomly placed behind a door. Next, the player chooses a door. Without loss of generality, we say that the player always chooses door 1. We could have instead universally quantified the choice of the player to indicate that the claim holds for any deterministic strategy. In that case, the proof would be largely the same, although with added cases so that the host does not open the player's door; we instead fix the player's choice to be door 1 for simplicity.
Finally, the host {\em nondeterministically} chooses a door to open, which is neither the player's pick, nor the car. 
\begin{figure}
\[
\begin{array}{ll}
\progname{Game}\triangleq\left\{
\begin{array}{l}
\mathsf{car} \coloneqq 1 \oplus_{\frac13} ( \mathsf{car} \coloneqq 2 \oplus_{\frac12} \mathsf{car} \coloneqq 3) \fse \\
\mathsf{pick} \coloneqq 1 \fse \\
\code{if} ~ \mathsf{car} = 1 ~ \code{then} \\
\quad \mathsf{open} \gets \{ 2, 3\} \\
\code{else if} ~ \mathsf{car} = 2 ~ \code{then} \\
\quad \mathsf{open} \coloneqq 3 \\
\code{else} \\
\quad \mathsf{open} \coloneqq 2
\end{array}
\right.
&
\quad
\progname{Switch} \triangleq\left\{
\begin{array}{l}
\code{if} ~ \mathsf{open} = 2 ~ \code{then} \\
\quad \mathsf{pick} \coloneqq 3 \\
\code{else} \\
\quad \mathsf{pick} \coloneqq 2
\end{array}
\right.
\end{array}
\]
\caption{Left: the Monty Hall program. Right: additional program to switch doors}
\label{fig:monty-hall}
\end{figure}
We now use Demonic OL to determine the probability of winning (that is, $\mathsf{pick} = \mathsf{car}$) in both the stick strategy ($\progname{Game}$) and the switch strategy ($\progname{Game}\fatsemi\progname{Switch}$, where  $\progname{Switch}$ is the program representing the player switching doors, presented on the right of \Cref{fig:monty-hall}). 
We derive the following triple for the $\progname{Game}$ program:
\[
\begin{array}{l}
\ob{\sure\tru}
\\
\quad\progname{Game}
\\
\ob{
  (\sure{\var{car} = 1} \land (\sure{\var{open} = 2} \nd\sure{\var{open} = 3})) \oplus_{\frac13} (\sure{\var{car} = 2 \land \var{open} = 3} \oplus_{\frac12} \sure{\var{car} = 3 \land \var{open} = 2})
}
\end{array}
\]
The derivation, using the rules from \Cref{fig:rules}, is shown in \Cref{fig:game-prog}. Note that to analyze the if statement, we first use the rule of \ruleref{Constancy} with $\var{pick}=1$ and then de-structure the remaining assertion with two applications of \ruleref{Prob Split}.
Below, we show manipulation of the postcondition of \progname{Game} to give us the probability of winning using the \emph{stick} strategy.
Irrelevant information about the opened door is first removed. Next, since $\var{pick} = 1$ in all cases, we weaken $\var{car}=1$ to $\var{pick} = \mathsf{car}$, and use $\var{pick} \neq \var{car}$ in the other outcomes. Finally, we use idempotence of $\oplus_\frac12$ in the last step.
\begin{align*}
& \sure{\var{pick} = 1} \land ((\sure{\var{car} = 1} \land (\sure{\var{open} = 2} \nd\sure{\var{open} = 3}))
\\ & \qquad \phantom{x}\oplus_{\frac13} (\sure{\var{car} = 2 \land \var{open} = 3} \oplus_{\frac12} \sure{\var{car} = 3 \land \var{open} = 2}))
\\
& \quad\implies \sure{\var{pick}= 1} \land (\sure{\var{car}=1} \oplus_{\frac13} (\sure{\var{car}=2} \oplus_{\frac12} \sure{\var{car}=3}))
\\
& \quad\implies \sure{\var{pick} = \var{car}} \oplus_{\frac13} (\sure{\var{pick} \neq\var{car}} \oplus_{\frac12} \sure{\var{pick} \neq\var{car}})
\\
& \quad\implies \sure{\var{pick} = \var{car}} \oplus_{\frac13} \sure{\var{pick} \neq\var{car}}
\end{align*}
So, the player wins with probability $\frac13$ in the stick strategy.
Now, for the \emph{switch} strategy, we can compositionally reason by appending the $\progname{Switch}$ program to the end of the previous derivation and then continue using the derivation rules. We again must de-structure to analyze the if statement, this time using both \ruleref{Prob Split} and also \ruleref{ND Split}.
\[
\begin{array}{l}
\ob{\sure\tru} \\
\;\;\progname{Game} \fse\\
\ob{(\sure{\var{car} = 1} \land (\sure{\var{open} = 2} \nd \sure{\var{open} = 3})) \oplus_{\frac13} (\sure{\var{car} = 2 \land \var{open} = 3} \oplus_{\frac12} \sure{\var{car} = 3 \land \var{open} = 2)}}
\\
\;\;\iftf{\mathsf{open} = 2}{\mathsf{pick} \coloneqq 3}{\mathsf{pick} \coloneqq 2}
\\
\ob{(\sure{\var{car} = 1} \land (\sure{\var{pick} = 3} \nd \sure{\var{pick} = 2})) \oplus_{\frac13} (\sure{\var{car} = 2 \land \var{pick} = 2} \oplus_{\frac12} \sure{\var{car} = 3 \land \var{pick} = 3})}
\\
\ob{\sure{\var{pick} \neq \var{car}} \oplus_{\frac13} (\sure{\var{pick} = \var{car}}\oplus_{\frac12} \sure{\var{pick} = \var{car}})}
\\
\ob{\sure{\var{pick} \neq \var{car}} \oplus_{\frac13} \sure{\var{pick} = \var{car}}}
\end{array}
\]
Note that the last two lines of the derivation are obtained by weakening the postcondition with the rule of \ruleref{Consequence}.
Just like in the previous case, we weaken the postcondition to only assert whether $\var{pick} = \var{car}$ or $\var{pick}\neq\var{car}$, and then collapse two outcomes using idempotence of $\oplus_\frac12$. This time the player wins with probability $\frac23$, meaning that switching doors is the better strategy.

\begin{figure}
\[
\begin{array}{l}
\ob{\sure\tru} \\
\;\;\mathsf{car} \coloneqq 1 \oplus_{\frac13} ( \mathsf{car} \coloneqq 2 \oplus_{\frac12} \mathsf{car} \coloneqq 3) \fse \\
\ob{\sure{\var{car} = 1} \oplus_{\frac13} (\sure{\var{car} = 2} \oplus_{\frac12} \sure{\var{car} = 3})} \\
\;\;\mathsf{pick} \coloneqq 1 \fse \\
\ob{\sure{\var{pick} = 1} \land (\sure{\var{car} = 1} \oplus_{\frac13} (\sure{\var{car} = 2} \oplus_{\frac12} \sure{\var{car} = 3}))} \\
\;\;\code{if} ~ \mathsf{car} = 1 ~ \code{then} \\
\quad\ob{\sure{\var{car} = 1}} \\
\quad\;\; \mathsf{open} \gets \{ 2, 3\} \\
\quad\ob{\sure{\var{car} = 1} \land (\sure{\var{open} = 2} \nd \sure{\var{open} = 3})} \\
\;\;\code{else if} ~ \mathsf{car} = 2 ~ \code{then} \\
\quad\ob{\sure{\var{car} = 2}} \\
\quad\;\; \mathsf{open} \coloneqq 3 \\
\quad\ob{\sure{\var{car} = 2 \land \var{open} = 3}} \\
\;\;\code{else} \\
\quad\ob{\sure{\var{car} = 3}} \\
\quad\;\; \mathsf{open} \coloneqq 2 \\
\quad\ob{\sure{\var{car} = 3 \land \var{open} = 2}} \\
\lrob{\sure{\var{pick} = 1} \land \left(
  \begin{array}{l}
    (\sure{\var{car} = 1} \land (\sure{\var{open} = 2} \nd \sure{\var{open} = 3}))
    \\
    \quad \oplus_{\frac13} (\sure{\var{car} = 2 \land \var{open} = 3} \oplus_{\frac12} \sure{\var{car} = 3 \land \var{open} = 2})
  \end{array}
  \right)}
\end{array}
\]
\caption{Derivation for the $\progname{Game}$ program from \Cref{fig:monty-hall}.}
\label{fig:game-prog}
\end{figure}

\subsection{The Adversarial von Neumann Trick}
\label{sec:von-neumann}

The \citet{neumann1951various} trick is a protocol for simulating a fair coin using a coin of unknown bias $p$. To do so, the coin is flipped twice. If the outcome is heads, tails---occurring with probability $p(1-p)$---then we consider the result to be heads. If the outcome is tails, heads---which also occurs with probability $p(1-p)$---then we consider to result to be tails. Otherwise, we try again.

In this case study, we work with an {\em adversarial} version of the von Neumann trick in which an adversary can alter the bias of the coin on each round, as long as the bias is between $\varepsilon$ and $1-\varepsilon$ for some fixed $0< \varepsilon \le\frac12$. We will show that just like in the original von Neumann trick, and  somewhat surprisingly, the simulated coin is fair in the presence of an adversarial bias. To model this protocol, we let the set $[\varepsilon, 1-\varepsilon]_N$ be a finite subset of the interval of rational numbers $[\varepsilon, 1-\varepsilon] \subseteq \mathbb{Q}$, formally defined as $[\varepsilon, 1-\varepsilon]_N \triangleq \{ \varepsilon + \frac{k(1 - 2\varepsilon)}{N} \mid k = 0\ldots N \}$. The program is shown below.
\[
\progname{AdvVonNeumann} \triangleq \left\{\begin{array}{l}
x \coloneqq \fls \fatsemi y \coloneqq \fls \fse \\
\code{while} ~x = y~ \code{do} \\
\quad p \gets[\varepsilon, 1-\varepsilon]_N \fse \\
\quad x \coloneqq \mathsf{flip}(p) \fse \\
\quad y \coloneqq \mathsf{flip}(p)
\end{array}\right.
\]
So, the program will terminate once $x \neq y$, meaning that one heads and one tails were flipped. We wish to prove that this program almost surely terminates, and that $x=\tru$ and $x=\fls$ occur with equal probability, meaning that we have successfully modeled a fair coin. More formally, we will prove that $\ob{\sure\tru} \ \progname{AdvVonNeumann} \ 
\ob{\sure{x=\tru} \oplus_{\frac12} \sure{x=\fls}}
$. 
We will use the \ruleref{Bounded Variant} rule to analyze the main loop, with the following variants.
\[
  \varphi_0 \triangleq (\sure{x = \tru} \oplus_\frac12 \sure{x = \fls})\land \sure{x \neq y}
  \qquad\qquad
  \varphi_1 \triangleq \sure{x = y}
\]
The variant $\varphi_1$ with the higher rank simply states that $x=y$, meaning that the loop will continue to execute. The lower-ranked variant $\varphi_0$ states both that $x\neq y$---the loop will terminate---and that $x=\tru$ and $x = \fls$ both occur with probability $\frac12$. This is an example of a variant with multiple outcomes that is not supported in pre-expectation reasoning, as mentioned in \Cref{sec:rank}.

\begin{figure}
\[
\begin{array}{l|l}
\begin{array}{l}
\ob{\sure\tru} \\
\;\;x \coloneqq \fls \\
\ob{\sure{\lnot x}} \\
\;\;y \coloneqq \fls \\
\ob{\sure{\lnot x \land \lnot y}} \implies \\
\ob{\sure{x = y}} \\
\;\;\whl{x=y}{} \\
\quad\;\; p \gets [\varepsilon, 1-\varepsilon]_N \fse \\
\quad\;\; x \coloneqq \mathsf{flip}(p) \fse \\
\quad\;\; y \coloneqq \mathsf{flip}(p) \\
\ob{\sure{x=\tru} \oplus_{\frac12} \sure{x=\fls}}
\end{array}
&
\def\arraystretch{1.1}
\begin{array}{l}
\ob{\sure{x = y}} \\
\;\; p \gets [\varepsilon, 1-\varepsilon]_N \fse \\
\ob{\bignd_{q\in [\varepsilon, 1-\varepsilon]_N}  \sure{p=q}} \\
\;\; x \coloneqq \mathsf{flip}(p) \fse \\
\ob{\bignd_{q\in [\varepsilon, 1-\varepsilon]_N}  \sure{p=q} \land (\sure{x=\tru} \oplus_q \sure{x=\fls})} \\
\;\; y \coloneqq \mathsf{flip}(p) \\
\lrob{
  \bignd_{q\in [\varepsilon, 1-\varepsilon]_N}
  \begin{array}{l}
    (\sure{x=\tru} \land (\sure{x= y} \oplus_q \sure{x\neq y})) \oplus_q \phantom{x}\\
     (\sure{x=\fls} \land (\sure{x\neq y} \oplus_q \sure{x= y}))
    \end{array}
} \\
\ob{\bignd_{q\in[\varepsilon, 1-\varepsilon]_N}
    \varphi_0 \oplus_{2q(1-q)} \varphi_1
}
\\
\ob{\varphi_0 \oplus_{2\varepsilon(1-\varepsilon)} (\varphi_0\nd\varphi_1)}
\end{array}
\end{array}
\]
\caption{Derivation of the von Neumann trick program.}
\label{fig:von-neumann}
\end{figure}

Each execution of the loop body will reduce the rank of the variant from 1 to 0 with probability $2p(1-p)$, where $p$ is chosen by the adversary. The \emph{worst case} is that the adversary chooses either $p = \varepsilon$ or $p = 1-\varepsilon$, in which case the probability of terminating the loop is $2\varepsilon(1-\varepsilon)$.
Given that there are only two variants, the \ruleref{Bounded Variant} rule simplifies to:
\[
\inferrule{
  \triple{\varphi_1}{C}{\varphi_0 \oplus_{2\varepsilon(1-\varepsilon)} (\varphi_0 \nd \varphi_1)}
}{
  \triple{\varphi_1}{\whl eC}{\varphi_0}
}
\]
The main derivation is shown on the left of \Cref{fig:von-neumann}, and the premise of the \ruleref{Bounded Variant} rule is on the right. After the two flips, all four probabilistic outcomes are enumerated. This is simplified using the associativity and commutativity rules from \Cref{sec:assertions} to conclude that $\varphi_0 \oplus_{2q(1-q)} \varphi_1$ for each $q$. As mentioned before, since we know that $2q(1-q) \ge 2\varepsilon(1-\varepsilon)$, we can weaken this to be $\varphi_0 \oplus_{2\varepsilon(1-\varepsilon)} (\varphi_0\nd\varphi_1)$. Now, since the assertion no longer depends on $q$, we use idempotence to remove the outer $\nd$. In the end, we get the postcondition $\sure{x=\tru}\oplus_\frac12 \sure{x=\fls}$, as desired.

\subsection{Probabilistic SAT Solving by Partial Rejection Sampling}
\label{sec:sat-solving}

Rejection sampling is a standard technique for generating random samples from certain distributions.
A basic version of rejection sampling can be used when a program has a way to generate random samples uniformly from a set $X$, and needs to generate uniform random samples from a set $S$, where $S \subseteq X$.
To do so, a simple rejection sampling procedure will draw a sample $x$ from $X$ and then check whether $x \in S$.
If $x \in S$, the rejection sampler is said to \emph{accept} $x$, and returns it.
However, if $x \notin S$, the sampler is said to \emph{reject} $x$, and repeats the process with a fresh sample from $X$.

In some situations, the set $X$ is a product of sets $X_1 \times \cdots \times X_n$, and a sample $x = (x_1, \dots, x_n)$ from $X$ is generated by independently drawing samples $x_1, \dots, x_n$, where each $x_i \in X_i$.
In this case, when $x$ is rejected, rather than redrawing \emph{all} of the $x_i$ to form a new sample from $X$, one might consider instead trying to \emph{partially} resample the components of $x$.
In particular if $x$ is \emph{close} to being in $S$, then one might try to only redraw some subset of components $x_{j_1}, \dots, x_{j_k}$, and re-use the other components of $x$ to form a new sample $x'$ to test for membership in $S$.

In general, partial resampling can result in drawing samples that are not \emph{uniformly} distributed over the set $S$.
However, \citet{guo2019uniform} observed that under certain conditions on the set $S$ and $X$, a partial rejection sampling procedure \emph{does} generate uniform samples from $S$.
In particular, when the $x_i$ are boolean variables, and the test for $(x_1, \dots, x_n) \in S$ can be encoded as a boolean formula $\phi$ over these variables, then it suffices for $\phi$ to be a so-called \emph{extremal} formula.
\citet{guo2019uniform} showed that many algorithms for sampling combinatorial structures can be formulated in terms of sampling a satisfying assignment to an extremal formula.
In this example, we consider a partial rejection sampler for generating a random satisfying assignment for a formula in 3-CNF form. We will prove that the sampler almost surely terminates if the formula has a satisfying assignment\footnote{Since this termination property holds even if the formula does not satisfy the extremal property, we will not formally define the extremal property or assume it as a precondition.}.

\begin{figure}
\[
\begin{array}{l|l|l}
\begin{array}{l}
\prs\triangleq \\
\;\; b \coloneqq \eval \fse \\
\;\; \code{while} ~\lnot{b}~ \code{do} \\
\;\; \quad \selclause \fse \\
\;\; \quad \resampleclause\fse \\
\;\; \quad b\coloneqq \eval \\
\\
\resampleclause\triangleq \\
\;\; \vars[\clvars[s][1]] \coloneq \flip{\tfrac12} \fse  \\
\;\; \vars[\clvars[s][2]] \coloneq \flip{\tfrac12} \fse  \\
\;\; \vars[\clvars[s][3]] \coloneq \flip{\tfrac12}
\end{array}
&
\begin{array}{l}
\selclause\triangleq \\
\;\;  s \coloneq -1 \fse \\
\;\;  i \coloneq 1  \fse \\ 
\;\;  \code{while} ~i \le M~ \code{do} \\
\;\;  \quad \code{if} ~\lnot \evalclause(i)~ \code{then} \\
\;\; \quad \quad \code{if} ~s = -1~ \code{then} \\
\;\; \quad \quad \quad s \coloneq i  \\
\;\; \quad \quad \code{else}  \\
\;\; \quad \quad \quad \skp \nd s\coloneqq i\fse \\
\;\; \quad  i \coloneq i + 1 \fse
\\\\
\end{array}
&
\begin{array}{l}
\evalclause(i)\triangleq \\
\;\;\xnor{\clsigns[i][1]}{\vars[\clvars[i][1]]} \orb\phantom{x} \\
 \;\;\xnor{\clsigns[i][2]}{\vars[\clvars[i][2]]} \orb\phantom{x} \\
 \;\;\xnor{\clsigns[i][3]}{\vars[\clvars[i][3]]}  \\
 \\\\
\eval \triangleq \\
\;\; \evalclause(1) \land\phantom{x} \\
\;\; \evalclause(2) \land\phantom{x} \\
\;\; \qquad\quad\vphantom{\int^0}\smash[t]{\vdots} \\
\;\; \evalclause(M)
\\
\end{array}
\end{array}
\]
\caption{SAT solving via rejection sampling, split into subroutines.}
\label{fig:solver-program}
\end{figure}

\Cref{fig:solver-program} shows the solver program, $\prs$, broken up into subroutines\footnote{Note that our language does not include subroutines, but these routines are interpreted as macros and are inlined into the main program. We only separate them for readability.}. 
The clauses are encoded using two 2-dimensional lists, $\clvars$ and $\clsigns$, each of size $M \times 3$, where $M$ is the number of clauses. See \Aref{app:list-sem} for the semantics of list operations. The entry $\clvars[i][j]$ gives the variable of the $j^\text{th}$ variable in clause $i$, and $\clsigns[i][j]$ is $0$ if this variable occurs in negated form, and is $1$ otherwise. The $\odot$ operation is \emph{xnor}, so $1\odot0 = 0\odot1 = 0$ and $0\odot0 = 1\odot1=1$. The program stores its current truth-value assignment for each variable in the list $\vars$.

Each iteration of the loop in $\prs$ starts by nondeterministically selecting an unsatisfied clause $s$ to resample via the $\selclause$ subroutine. To do so, it iterates over the clauses, checking if each one is satisfied using $\evalclause$. When an unsatisfied clause is found, $s$ is nondeterministically either updated to $i$ or left as is (unless $s=-1$, in which case $s$ is updated to $i$, to ensure that some unsatisfied clause is picked).
Nondeterminism allows us to under-specify \emph{how} the sampler selects a clause to resample, which in practice might be based on various heuristics. By proving almost-sure termination for this non-deterministic version, we establish almost-sure termination no matter which heuristics are used, including randomized ones.
Presuming that the formula is not yet satisfied ($\eval=\fls$), the $\selclause$ routine selects an $s$ such that $1\le s \le M$ and $\evalclause(s) = \fls$, which is captured by the following specification and proven in \Aref{app:sat-solving}.
\[
  \triple{\sure{\eval = \fls}}{\selclause}{\sure{1 \le s \le M \land \evalclause(s) = \fls}}
\]
Next, the three variables in the selected clause are resampled. In order to prove that the program almost surely terminates, we need to show that the resampling operation brings the process closer to termination with nonzero probability. To do this, we measure how close the candidate solution is to some satisfying assignment $\xgood$ (recall we assumed that at least one such satisfying assignment exists). Closeness is measured via the Hamming distance, computed as follows, where the Iverson brackets $[e]$ evaluates to 1 if $e$ is true and 0 is $e$ is false, and $N$ is the number of variables.
\[
  \mathsf{dist}(x, y) \triangleq \textstyle\sum_{i=1}^{N} \big[x[i] \neq y[i]\big]
\]
Now, we can give a specification for $\resampleclause$ in terms of the Hamming distance. That is, if $\dist(\vars, \xgood)$ is initially $k$ and clause $s$ is not satisfied, then resampling $s$ will strictly reduce the Hamming distance with probability at least $\frac18$.
The reason for this is that before resampling, $\vars$ and $\xgood$ must disagree on at least one of the variables in clause $s$, since clause $s$ is not satisfied by $\vars$.
After resampling, there is at least a $\frac18$ probability that all 3 resampled variables agree with $\xgood$, in which case the Hamming distance is reduced by at least 1. The full proof is in \Aref{app:sat-solving}.
\[
\triple{\sure{\dist(\vars, \xgood) = k \land \evalclause(s) = \fls}}{\resampleclause}{\sure{\dist(\vars, \xgood) < k} \oplus_{\frac18} \sure\tru}
\]
Using these specifications, we now prove that $\prs$ almost surely terminates. The derivation is shown in \Cref{fig:solver}. We instantiate \ruleref{Bounded Rank} to analyze the loop with the following parameters:
\[
  P \;\;\triangleq\;\; b = \eval
  \qquad\qquad
  e_\mathsf{rank} \;\;\triangleq\;\; [\lnot\eval] \cdot \dist(\vars,\xgood)
  \qquad\qquad
  p\;\;\triangleq\;\; \frac18
\]
The invariant $P$ simply states that $b$ indicates whether the current assignment of variables satisfies the formula. The ranking function is equal to the Hamming distance between $\vars$ and the sample solution $\xgood$ if the the formula is not yet satisfied, otherwise it is zero, which accounts for the fact that the program may find a solution other than $\xgood$. We also remark that $e_\mathsf{rank}$ is bounded between $1$ and $N$ (where $N$ is the total number of variables) as long as the formula is not yet satisfied. As we saw in the specification for $\resampleclause$, the probability of reducing the rank is $\frac18$.

\begin{figure}
\[
\begin{array}{l}
\ob{\sure\tru}\\
\;\; b \coloneqq \eval \fse \\
\ob{\sure{b = \eval}}\\
\;\; \code{while} ~\lnot{b}~ \code{do} \\
\quad \ob{\sure{b = \eval \land \lnot b\land [\lnot\eval]\cdot\dist(\vars,\xgood) = k}} \\
\quad \ob{\sure{\eval = \fls\land\dist(\vars,\xgood) = k}} \\
\;\; \quad \selclause \fse \\
\quad \ob{\sure{0 \le s < M \land \evalclause(s) = \fls \land \dist(\vars,\xgood) = k}} \\
\;\; \quad \resampleclause \fse \\
\quad \ob{\sure{\dist(\vars,\xgood) < k} \oplus_\frac18 \sure\tru} \\
\;\; \quad b \coloneqq \eval \\
\quad \ob{\sure{b = \eval \land \dist(\vars,\xgood) < k} \oplus_\frac18 \sure{b=\eval}} \\
\quad \ob{\sure{b = \eval \land [\lnot \eval]\cdot \dist(\vars,\xgood) < k} \oplus_\frac18 \sure{b=\eval}} \\
\ob{\sure{\eval = \tru}}
\end{array}
\]
\caption{Derivation of the $\prs$ program, where $\xgood$ is a known satisfying assignment.}
\label{fig:solver}
\end{figure}

Entering the loop, we see that $\eval$ must be false, so the rank is just $\dist(\vars,\xgood)$. Applying the specifications for the two subroutines, we prove that the Hamming distance strictly decreases with probability at least $p$. The assignment to $b$ then reestablishes the invariant. When the Hamming distance has decreased, we also have that $e_\mathsf{rank}$ decreased, as multiplying by $[\lnot\eval]$ can only make the term smaller. Upon exiting the loop, we have that $b = \tru$ and hence the final postcondition $\eval = \tru$, meaning that the formula is satisfied and the program almost surely terminates.

\citet{error_credits} prove termination of a similar randomized SAT solving technique using a separation logic for reasoning about upper bounds on probabilities of non-termination. Because the language they consider does not have non-determinism, they fix a particular strategy for selecting clauses to resample, whereas the use of nondeterministic choice in the proof above implies termination for any strategy that selects an unsatisfied clause.
Their proof essentially shows that for any $\varepsilon > 0$, after some number of iterations, the Hamming distance will decrease with probability at least $1 - \varepsilon$.
The \ruleref{Bounded Rank} rule effectively encapsulates this kind of reasoning in our proof.

\section{Related Work}
\label{sec:related}

\subsubsection*{Program Logics}

Demonic Outcome Logic takes inspiration from program logics for reasoning about purely probabilistic programs, such as Probabilistic Hoare Logic \cite{hartog,corin2006probabilistic,den_hartog1999verifying}, VPHL \cite{rand2015vphl}, Ellora \cite{ellora}, and Outcome Logic \cite{outcome,zilberstein2024outcome,zilberstein2024relatively}.
Those logics provide means to prove properties about the distributions of outcomes in probabilistic programs, to which we added the ability to also reason about demonic nondeterminism. 

Although this paper introduces the first logic for reasoning about the \emph{outcomes} in demonic probabilistic programs, there is some prior work on other styles of analysis.
Building on the work of \citet{varacca2002powerdomain,varacca2003probability}, Polaris is a relational separation logic for reasoning about concurrent probabilistic programs \cite{tassarotti2018verifying,polaris}. Specifications take the form of refinements, where a complex program is shown to behave equivalently to an idealized version. Probabilistic analysis can then be done on the idealized program to determine its expected behavior, but it is external to the program logic. Polaris also does not support unbounded looping, and therefore it cannot be used to analyze our last two case studies.

\subsubsection*{Weakest Pre-Expectations}
Weakest pre-expectation ($\mathsf{wp}$) transformers are calculi for reasoning about probabilistic programs in terms of expected values \cite{wpe}. They were inspired by propositional weakest precondition calculi \cite{gcl,Dijkstra76}, Probabilistic Propositional Dynamic Logic \cite{ppdl}, and probabilistic predicate transformers \cite{jones}. Refer to \citet{kaminski} for a thorough overview of this technique.

From the start, $\wpre$ supported nondeterminism; in fact, $\wpre$ emerged from a line of work on semantics for randomized nondeterministic programs \cite{jifeng1997probabilistic,mciver2001partial,morgan1996refinement}. Nondeterminism is handled by lower-bounding expectations, corresponding to larger expected values being \emph{better}. An \emph{angelic} variant can alternatively be used for upper bounds.

Work on $\wpre$ has intersected with termination analysis for probabilistic programs. Some of this work uses martingales \cite{chakarov2013probabilistic} to show that programs terminate with finite expected running time. More sophisticated techniques exist for almost sure termination too \cite{kaminski,mciver2005abstraction,mciver2018new}.

As noted by \citet[\S 2.3.3]{kaminski}, the choice of either upper or lower bounding the expected values is \emph{``extremal''}---it forces a view where expectations must be either maximized or minimized, as opposed to our approach where multiple outcomes can be represented in one specification. However, reasoning about outcomes and expectations are not mutually exclusive; \citet[Theorem 1]{ellora} showed how to embed a $\wpre$ calculus in a probabilistic program logic. A similar construction is possible in Demonic Outcome Logic.

\subsubsection*{Powerdomains for Probabilistic Nondeterminism}

Powerdomains are a well-studied domain-theoretic tool for reasoning about looping nondeterministic programs, providing a means for defining a continuous domain in which loops can be interpreted as fixed points. This revolves around defining appropriate orders over sets of states to show that iterated actions eventually converge.
Given a partially ordered domain of program states $\langle \Sigma, \le\rangle$, there are three typical choices for orders over sets of states, known as the Hoare, \citet{smyth1978power}, and Egli-Milner orders, defined below:
\[
\arraycolsep=.25em
\begin{array}{lclcl}
S &\sqsubseteq_{\mathsf{H}}& T & \qquad\text{iff}\qquad &
  \forall \sigma\in S.\quad
   \exists \tau\in T.\quad
    \sigma \le \tau
\\
S &\sqsubseteq_{\mathsf{S}} &T & \text{iff} & \forall \tau\in T.\quad
   \exists \sigma\in S.\quad
    \sigma \le \tau
\\
S &\sqsubseteq_{\mathsf{EM}} &T & \text{iff} & S \sqsubseteq_{\mathsf{H}} T \quad\text{and}\quad S \sqsubseteq_{\mathsf{S}} T
\end{array}
\]
In general, none of these relations are antisymmetric, making them \emph{preorders}, whereas domain theoretic tools for finding fixed points operate on  \emph{partial orders}. So the sets representing the program semantics must be closed in order to obtain a proper domain. This closure operation loses precision of the semantics, incorporating additional possibilities which are not always intuitive.

The Hoare order requires a down-closure, essentially meaning that nontermination may always be an option. This makes it a good choice for partial correctness, where we only wish to determine what happens \emph{if} the program terminates, as in Hoare Logic.
The \citet{smyth1978power} order, which we use in this paper, requires an upwards closure, so that nontermination becomes erratic behavior. This makes it a good choice for total correctness \cite{manna1974axiomatic}, which is concerned only with terminating programs where erratic behavior does not arise.

In the Egli-Milner case---and the associated \citeauthor{plotkin1976powerdomain} Powerdomain [\citeyear{plotkin1976powerdomain}]---the more precise, but also less intuitive Egli-Milner closure is used.
\citet{mciver2001partial} created a denotational model where fixed points are taken with respect to the Egli-Milner order rather than the \citet{smyth1978power} one.
As such, they require the domain of computation to be Egli-Milner closed, which means that $S = \mathord\uparrow S \cap \mathord\downarrow S$. Unlike up-closedness (required for the Smyth approach), which is preserved by all the operations in \Cref{fig:semantics}, the semantics of \citet{mciver2001partial} must take the Egli-Milner closure after nondeterministic and probabilistic choice and after sequential composition, making the model more complex and adding outcomes that do not have an obvious operational meaning. Refer to \citet{tix2009semantic,keimel2017mixed} for a more complete exploration of that approach.

Let us examine the semantics of a coin flip in order to demonstrate why the Smyth order is preferable to Hoare. The variable $x$ is assigned the values $\tru$ or $\fls$ each with probability $\frac12$. So, the result of running the program is a singleton set containing the aforementioned distribution.
\[
\de{x \coloneqq \mathsf{flip}\left({\textstyle\frac12}\right)}(\sigma) = \left\{\;\;
  \arraycolsep=0em
  \def\arraystretch{1.25}
  \scriptsize
  \begin{array}{llcl}
    \sigma[ x \coloneqq \tru &] & \quad\mapsto\quad & \frac12
    \\
    \sigma[ x \coloneqq \fls &] & \mapsto & \frac12
  \end{array}\;\;
\right\}
\]
If we were to use the Hoare powerdomain, then we would need to down-close this set, adding all smaller distributions too. This not only means that nontermination is possible, but we would not even be able to determine that $x=\tru$ and $x=\fls$ occur with equal probability.
\[
\de{x \coloneqq \mathsf{flip}\left({\textstyle\frac12}\right)}(\sigma) = \left\{
  \arraycolsep=0em
  \;\;
  \scriptsize
  \begin{array}{llcl}
    \sigma[ x \coloneqq \tru&] & \quad\mapsto\quad & p
    \\
    \sigma[ x \coloneqq \fls&] & \mapsto & q
    \\
    \bot && \mapsto & 1-p-q
  \end{array}
  \;\;
  \middle|
  \quad
  p \le \frac12, ~q \le \frac12
  \quad
\right\}
\]
This was the approach taken by \citet{varacca2002powerdomain}, and the loss of precision is reflected in the adequacy theorems of that work. In particular, \citeauthor{varacca2002powerdomain}'s Proposition 6.10 shows that the denotational model includes outcomes that may not be possible according to the associated operational model.
By contrast, the up-closure---required by the Smyth order---adds nothing for this program; the semantics is already a full distribution and therefore there are no distributions larger than it. We can therefore conclude that the two outcomes occur with probability exactly $\frac12$, as desired.

This example demonstrates that the notion of partial correctness (as embodied by the Hoare order) does not make much sense in probabilistic settings, since it translates to uncertainty about the minimum probability of an event. Total correctness, on the other hand, does make sense, and corresponds to the notion of \emph{almost sure termination}, which is a property of great interest in probabilistic program analysis \cite{mciver2018new,chakarov2013probabilistic}.

The problem with the Smyth order is that a semantics based on it is not \citet{scott1972continuous} continuous in the presence of unbounded nondeterminism \cite{apt1986countable,s_ondergaard1992non}. This is the reason why \citet{jifeng1997probabilistic} instead use the Knaster-Tarski theorem to guarantee the fixed point existence via transfinite iteration, which only requires monotonicity and not Scott continuity. The main shortcoming of \citeauthor{jifeng1997probabilistic}'s approach is that it did not guarantee non-emptiness of the set of result distributions, meaning that some programs may have vacuous semantics.

To address this, \citet{wpe} added the additional requirement that domain only include topologically closed sets (\citeauthor{wpe} called this property \emph{Cauchy closure}). 
As we mentioned in \Cref{sec:fixpoint} and proved in \Aref{app:fixpoint}, closure ensures that no programs are modeled as empty sets, but it also prevents commands from exhibiting unbounded nondeterminism. For example, it is not possible to represent a program $x \coloneqq\bigstar$, which nondeterministically selects a value for $x$ from the natural numbers---the set $\mathbb N$ is not closed since it does not contain a limit point.

\citet{mciver2005abstraction} suggested that topological closure opens up the possibility of Scott continuity.
In addition, there has been work to combine classical powerdomains for nondeterminism \cite{plotkin1976powerdomain,smyth1978power} with the probabilistic powerdomain of \citet{jones,jones1989probabilistic}. This was first pursued by \citet{tix1999continuous}, and was later refined in \citet{tix2000convex,tix2009semantic,keimel2017mixed}. They obtain a Scott continuous composition operation (which they call $\hat f$) via a universal property, as opposed to the direct construction of \citet{jacobs2008coalgebraic} that we use.

\subsubsection*{Monads for Probabilistic Nondeterminism}
\citet{varacca2002powerdomain,varacca2003probability} introduced powersets of \emph{indexed valuations}. An indexed valuation behaves similarly to a distribution, but the idempotence property is removed, so that $X \oplus_p X \neq X$.
As shown by \citet{varacca_winskel_2006}, a powerset of indexed valuations has a \citet{beck1969distributive} distributive law, and is therefore a monad. However, indexed valuations are difficult to work with since equivalence is taken modulo renaming of the indices.

\citet[Theorem 6.5]{varacca2002powerdomain} proved that denotational models based on indexed valuations are equivalent to operational models in which a \emph{deterministic} scheduler resolves the nondeterminism. Given our goals of modeling \emph{adversarial} nondeterminism, we opted to use convex sets, which model a more powerful probabilistic scheduler, giving us robust guarantees in a stronger threat model.

An alternative approach is to flip the order of composition and work instead with distributions of nondeterministic outcomes. While the distribution monad does not compose with powerset, it does compose with multiset, as shown by \citet{jacobs2021multisets} and further explored by \citet{kozen2024multisets}.
The barrier to this approach is that the multisets must be finite, but it is easy to construct programs that reach infinitely many nondeterministic outcomes via while loops. So this model cannot be used to represent arbitrary programs from the language in \Cref{sec:semantics}.
The use of multiset instead of powerset is again an instance of removing an idempotence law, this time for nondeterminism: $X \nd X \neq X$. Indeed, idempotence is the key reason why no distributive law exists in both cases \cite{parlant2020monad,zwart2019no,zwart2020non}.

\subsubsection*{Other Semantic Approaches}
\citet{segala1995modeling} created a model in which a tree of alternating probabilistic and nondeterministic choices is collapsed into a set of distributions collected from all combinations of nondeterministic choices. However, this model does not lead to a compositional semantics.
Additional operational models of probabilistic nondeterminism have been studied through the lens of process algebras \cite{hartog,hartog1999mixing,den_hartog1998comparative,mislove2004axioms}.
In addition, coalgebraic methods have been used to define trace semantics and establish bisimilarity of randomized nondeterministic automata \cite{jacobs2008coalgebraic,bonchi2019theory,bonchi2021presenting,bonchi2022theory,bonchi2021distribution}.

\citet{DBLP:journals/pacmpl/AguirreB23} note the difficulties of building denotational models that combine probabilistic and nondeterministic choice with other challenging semantic features. Instead, they start with an operational semantics for probabilistic and nondeterministic choice and then construct a step-indexed logical relations model for a typed, higher-order language with polymorphism and recursive types. Using this logical relations model, they derive an equational theory for contextual equivalence and show that it validates many of the equations found in denotational models.

\section{Conclusion}
\label{sec:conclusion}

This paper introduced \emph{Demonic Outcome Logic}, a logic for outcome based reasoning about programs that are both randomized and nondeterministic, a combination that presents many challenges for program semantics and analysis. The logic includes several novel features, such as equational laws for manipulating pre- and postconditions and rules for loops that both establish termination and quantify the distribution of final outcomes from a single premise. We build on a large body of work on semantics for probabilistic nondeterminism \cite{jifeng1997probabilistic,morgan1996refinement,wpe,jacobs2008coalgebraic,tix2009semantic,varacca2002powerdomain}, and also draw inspiration from Outcome Logic \cite{outcome,zilberstein2024outcome,zilberstein2024relatively} and weakest pre-expectation calculi \cite{wpe,kaminski,zhang2024quantitative}. The resulting logic contains  rules that enable effective reasoning about distributions of outcomes in randomized nondeterministic programs, as illustrated through the three presented case studies. The simplicity of the rules is enabled by a carefully chosen denotational semantics that allows us to hide the complex algebraic properties of the domain in the proof of their soundness. Compared to weakest pre-expectation reasoning, the propositional approach afforded by Demonic Outcome Logic enables reasoning about multiple outcomes in tandem, leading to more expressive specifications, and the loop rules rely on fewer, simpler premises. 

Moving forward, we want to go beyond standard nondeterminism and extend the logic for reasoning about probabilistic fine-grain concurrency with shared memory. This will require fundamental changes to the denotational semantics and inference rules, although prior work on Concurrent Separation Logic \cite{csl}, Outcome Separation Logic \cite{zilberstein2024outcome}, and Concurrent Kleene Algebra~\cite{hoare2011concurrent} will provide a good source of inspiration.
Using the resulting logic, we will verify concurrent algorithms such as distributed cryptographic protocols, for which state of the art techniques use limited models of concurrency and operate by establishing observational equivalence and then separately proving properties of an idealized program \cite{gancher2023core}. By contrast, we plan to develop a logic based on a fine-grain concurrency model, which can prove direct specifications involving probabilistic outcomes. We also plan to explore a mechanized implementation of the logic, building on existing frameworks for (concurrent) separation logic such as Iris \cite{iris1}.

\section*{Acknowledgments}

This work was partially supported by ERC grant Autoprobe (no. 101002697), ARIA's Safeguarded AI programme, and NSF grant CCF-2008083. Some work on this paper was completed during a workshop at the Bellairs Research Institute of McGill University; we thank Prakash Panangaden for the invitation and the institute and their staff for providing a wonderful research environment. 

\bibliographystyle{ACM-Reference-Format}
 \bibliography{refs}
 
 \ifx\extended\undefined\else
\allowdisplaybreaks
\appendix
\clearpage

{\noindent \huge\bfseries\sffamily Appendix}

\section{Examples and Counterexamples}

\subsection{Non-Idempotence of Logical Operators}
\label{app:idempotence}

In this section, we show why the idempotence rule $\varphi\nd\varphi\Rightarrow\varphi$ does not apply when assertions model sets of distributions of states. To do so, we first need to define a new semantics for the logical operators:
\[
\begin{array}{lcllllll}
S \vDash \varphi \nd \psi
  &\text{iff}&
  S = S_1 \nd S_2
  &\text{and}&
  S_1 \vDash\varphi
  &\text{and}&
  S_2 \vDash\psi
  &\text{for some}~S_1, S_2 \in \C(\Sigma)
\\
S \vDash \varphi \oplus_p \psi
  &\text{iff}&
  S = S_1 \oplus_p S_2
  &\text{and}&
  S_1 \vDash\varphi
  &\text{and}&
  S_2 \vDash\psi
  &\text{for some}~S_1, S_2 \in \C(\Sigma)
\\
S\vDash \sure{P}
  &\text{iff}&
  \multicolumn{4}{l}{
    \bigcup_{\mu\in S}\supp(\mu) \subseteq \sem P
  }
\end{array}
\]
This is similar to how the \emph{outcome conjunction} of Outcome Logic is defined, where the collection of outcomes is split according to the same nondeterminism operation as is used in the program semantics, to satisfy the two assertions, $\varphi$ and $\psi$, individually \cite{outcome}.
Now, let us revisit the coin flip game. The semantics of the program is shown below.
\[
  \de{y \gets\Bool \fatsemi x\coloneqq\mathsf{flip}\left(\tfrac12\right)}(\sigma) = \left\{\;
  \arraycolsep=0em
    \def\arraystretch{1.25}
    \begin{array}{lll}
      \sigma[x \coloneqq \tru, &y\coloneqq\tru&] \mapsto \frac12\cdot p
      \\
      \sigma[x \coloneqq \fls,~ &y\coloneqq\tru &] \mapsto \frac12\cdot p
      \\
      \sigma[x \coloneqq \tru, &y\coloneqq\fls &] \mapsto \frac12\cdot(1-p)
      \\
      \sigma[x \coloneqq \fls, &y\coloneqq\fls &] \mapsto \frac12\cdot(1-p)
    \end{array}
    \quad\middle|\quad
    p\in[0,1]~
  \right\}
\]
Let us call this set $S$.
It is not hard to see that $S$ is the convex union of two sets that satisfy $\sure{x =y} \oplus_\frac12 \sure{x\neq y}$, so we have that:
\[
  S \vDash (\sure{x =y} \oplus_\frac12 \sure{x\neq y}) \nd (\sure{x =y} \oplus_\frac12 \sure{x\neq y})
\]
It is tempting to say that $S\vDash \sure{x=y} \oplus_\frac12 \sure{x\neq y}$, but this is not the case. To prove this, we will show that $S \neq S_1 \nd S_2$ for any $S_1$ and $S_2$ such that $S_1 \vDash \sure{x=y}$ and $S_2 \vDash \sure{x\neq y}$. If $S_1\vDash \sure{x=y}$ and $S_2\vDash \sure{x\neq y}$, then they must be of the following forms, where $T_1,T_2 \subseteq [0,1]$.
\[\footnotesize
S_1 \triangleq
\left\{\;
\arraycolsep=0em
  \begin{array}{lll}
    \sigma[x \coloneqq \tru, &y\coloneqq \tru &] \mapsto p
    \\
    \sigma[x \coloneqq \fls, &y\coloneqq \fls &] \mapsto 1-p
  \end{array}
  \;\middle|\;
  p\in T_1
\;\right\}
\qquad
S_2 \triangleq
\left\{\;
\arraycolsep=0em
  \begin{array}{lll}
    \sigma[x \coloneqq \tru, &y\coloneqq \fls &] \mapsto q
    \\
    \sigma[x \coloneqq \fls, &y\coloneqq \tru &] \mapsto 1-q
  \end{array}
  \;\middle|\;
  q\in T_2
\;\right\}
\]
Clearly, $S_1 \vDash \sure{x=y}$ and $S_2\vDash \sure{x\neq y}$, but $S \neq S_1 \nd S_2$. To see this, let us consider $S_1 \nd S_2$: 
\begin{align*}
  S_1 \nd S_2
  &= \left\{\;
\arraycolsep=0em
  \begin{array}{lll}
    \sigma[x \coloneqq \tru, &y\coloneqq \tru &] \mapsto r\cdot p
    \\
    \sigma[x \coloneqq \fls, ~&y\coloneqq \fls &] \mapsto r\cdot (1-p)
    \\
    \sigma[x \coloneqq \tru, & y\coloneqq \fls&] \mapsto (1-r) \cdot q
    \\
    \sigma[x \coloneqq \fls, & y\coloneqq \tru&] \mapsto (1-r) \cdot (1-q)
  \end{array}
  \;\middle|\;\;
  r\in[0,1], ~p \in T_1, ~q \in T_2\;\;
\right\}
\end{align*}
Whereas in $S$, the probability that $x$ is true or false given a fixed value of $y$ is always exactly $\frac12$, that is not the case in $S_1\nd S_2$. For instance, when $r = 1$ and $p$ is any arbitrary element of $T_1$, we get:
\[
  \left(
  \arraycolsep=0em
  \begin{array}{lll}
    \sigma[x \coloneqq \tru, &y\coloneqq \tru &] \mapsto p
    \\
    \sigma[x \coloneqq \fls, ~&y\coloneqq \fls &] \mapsto 1-p
  \end{array}
  \right)
  \in S_1 \nd S_2
\]
But clearly that distribution is not in $S$, since the two outcomes where $x=y$ must occur with aggregate probability $\frac12$.
So, we just saw that for any $S_1 \vDash \sure{x=y}$ and $S_2 \vDash \sure{x\neq y}$, there exists a $\mu \in S_1 \nd S_2$ such that $\mu \notin S$. This suggests we could modifiy the semantics above to say that $S \subseteq S_1 \nd S_2$ instead of $S=S_1 \nd S_2$:
\[
S \vDash \varphi \nd \psi
  \quad\text{iff}\quad
  {\color{purple}S \subseteq S_1 \nd S_2}
  \quad\text{and}\quad
  S_1 \vDash\varphi
  \quad\text{and}\quad
  S_2 \vDash\psi
  \quad\text{for some}~S_1, S_2 \in \C(\Sigma_\bot)
\]
However, this removes all reachability claims since \eg $S_1 \subseteq S_1 \nd S_2$, and therefore $\varphi\Rightarrow \varphi\nd \psi$ for any satisfiable $\psi$. In fact, we get that $S \vDash \varphi \nd \psi$ iff for every $\mu \in S$, there exists a probability $p$ such that $\mu \vDash \varphi \oplus_p \psi$, which is exactly the semantics we have defined in \Cref{sec:assertions}. This suggests that a demonic interpretation is inevitable if idempotence is a desired propositional property.

\subsection{Derivation of Coin Flip Programs}
\label{app:coin}

In this section, we give the full derivation for the coin flip game that was introduced in \Cref{sec:overview}. Below, we show the derivation for the variant in which the adversary picks first. The variant in which the coin flip happens first is completely analogous.
\begin{equation}\label{eq:flip-true}
      \inferrule*[right=\ruleref{Prob}]{
        \inferrule*[right=\ruleref{Assign}]{\;}{
          \triple{\sure{\tru=y}}{x \coloneqq\tru}{\sure{x=y}}
        }
        \\
        \inferrule*[Right=\ruleref{Consequence}]{
          \inferrule*[Right=\ruleref{Assign}]{\;}{
            \triple{\sure{\fls \neq y}}{x\coloneqq\fls}{\sure{x\neq y}}
          }
        }{
          \triple{\sure{\tru=y}}{x\coloneqq\fls}{\sure{x\neq y}}
        }
      }{
        \triple{\sure{\tru=y}}{x\coloneqq\mathsf{flip}\left(\tfrac12\right)}{\sure{x=y} \oplus_\frac12 \sure{x\neq y}}
      }
\end{equation}

\begin{equation}\label{eq:flip-false}
      \inferrule*[right=\ruleref{Prob}]{
        \inferrule*[right=\ruleref{Consequence}]{
          \inferrule*[Right=\ruleref{Assign}]{\;}{
            \triple{\sure{\tru \neq y}}{x\coloneqq\tru}{\sure{x\neq y}}
          }
        }{
          \triple{\sure{\fls=y}}{x \coloneqq\tru}{\sure{x\neq y}}
        }
        \\
        \inferrule*[Right=\ruleref{Assign}]{\;}{
          \triple{\sure{\fls = y}}{x\coloneqq\fls}{\sure{x= y}}
        }
      }{
        \triple{\sure{\fls = y}}{x\coloneqq\mathsf{flip}\left(\tfrac12\right)}{\sure{x\neq y} \oplus_\frac12 \sure{x= y}}
      }
\end{equation}

\begin{equation}\label{eq:flip-deriv}
  \inferrule*[right=\ruleref{Consequence}]{
    \inferrule*[Right=\ruleref{ND Split}]{
      (\ref{eq:flip-true})
      \\
      (\ref{eq:flip-false})
    }{
      \triple{\sure{\tru = y} \nd \sure{\fls = y}}{x \coloneqq\mathsf{flip}\left(\tfrac12\right)}{(\sure{x = y} \oplus_\frac12 \sure{x\neq y}) \nd (\sure{x \neq y} \oplus_\frac12 \sure{x= y})}
    }
  }{
    \triple{\sure{y = \tru} \nd \sure{y=\fls}}{x \coloneqq\mathsf{flip}\left(\tfrac12\right)}{\sure{x = y} \oplus_\frac12 \sure{x\neq y}}
  }
\end{equation}

\[
\inferrule*[right=\ruleref{Seq}]{
  \inferrule*[right={\Cref{lem:nd-rule}}]{\;}{
    \triple{\sure{\tru}}{y \gets\Bool}{\sure{y = \tru} \nd \sure{y=\fls}}
  }
  \\
  (\ref{eq:flip-deriv})
}{
  \triple{\sure{\tru}}{y \gets\Bool \fatsemi x \coloneqq\mathsf{flip}\left(\tfrac12\right)}{ \sure{x = y} \oplus_\frac12 \sure{x\neq y}}
}
\]

\section{Semantics}
\label{app:mono}

In this section, we provide some of the details about the well-definedness of the program semantics that are omitted from \Cref{sec:semantics}.

\subsection{Properties of the Kleisli Extension}

\begin{lemma}\label{lem:countably-convex}
If $S\subs \D(X)$ is convex and closed, then $S$ is countably convex; that is, every countable convex combination of elements of $S$ is in $S$.
\end{lemma}
\begin{proof}
The basic open sets are of the form $(\prod_{x\in F}U_x)\times[0,1]^{X\setminus F}$, where $F\subs X$ is finite and $U_x$ are basic open intervals. Let $\sum_{i\in I} a_i\mu_i$ be a countable convex combination of elements of $S$, $a_0>0$. Let $c_n = \sum_{i\le n} a_i$. By convexity, $\sum_{i\le n}(a_i/c_n)\mu_i$ are in $S$, since these are finite convex combinations. Since $c_n\to 1$, for sufficiently large $n$, $c_n > a_i/(\eps + a_i)$, so $(a_i/c_n) - a_i < \eps$. Thus every basic open set $(\prod_{x\in F}U_x)\times[0,1]^{X\setminus F}$ containing $\sum_i a_i\mu_i$ contains all but finitely many $\sum_{i\le n}(a_i/c_n)\mu_i$. Since $S$ is closed, $\sum_i a_i\mu_i\in S$.
\end{proof}
In \Cref{sec:prelim}, we introduced the Kleisi extension $f^\dagger\colon\C(X) \to \C(Y)$, but to verify this typing we must show that $f^\dagger(S)\in \C(Y)$ for any $f\colon X\to\C(Y)$ and $S\in \C(X)$, that is, $f^\dagger(S)$ is convex, closed, and up-closed. It was asserted by \citet{jifeng1997probabilistic} that $f^\dagger(S)$ is convex and up-closed, but they did not present the proof, so we provide it here.

\begin{lemma}
$f^\dagger(S)$ is convex.
\end{lemma}
\begin{proof}
Take any $\mu,\nu\in f^\dagger(S)$, so that means that there must be $\mu',\nu' \in S$, $(\mu_x)_{x\in\supp(\mu')}$, and $(\nu_x)_{x\in\supp(\nu')}$ such that $\mu = \sum_{x\in \supp(\mu')} \mu'(x)\mu_x$ and $\nu = \sum_{x\in\supp(\nu')} \nu'(x)\nu_x$ and $\mu_x,\nu_x \in f(x)$ for all $x$. Now, take any $p$:
\begin{align*}
  p\mu + (1-p)\nu
  &= p\smashoperator{\sum_{x\in \supp(\mu')}} \mu'(x)\mu_x + (1-p)\smashoperator{\sum_{x\in\supp(\nu')}} \nu'(x)\nu_x
  \\
  & = \sum_{x \in X} p\mu'(x)\mu_x + (1-p)\nu'(x)\nu_x
  \\
  \intertext{Let $\xi = p\mu' + (1-p)\nu'$. Clearly $\xi\in S$ since it is a convex combination of $\mu',\nu'\in S$.}
  & = \smashoperator{\sum_{x \in \supp(\xi)}} \xi(x)\left( p\frac{\mu'(x)}{\xi(x)}\mu_x + (1-p)\frac{\nu'(x)}{\xi(x)} \nu_x \right) 
\end{align*}
Now, we also clearly see that $p\frac{\mu'(x)}{\xi(x)}\mu_x + (1-p)\frac{\nu'(x)}{\xi(x)} \nu_x \in f(x)$ since it is a convex combination of $\mu_x$ and $\nu_x$. Therefore, $p\mu + (1-p)\nu \in f^\dagger(S)$, and therefore $f^\dagger(S)$ is convex.

\end{proof}

\begin{lemma}
$f^\dagger(S)$ is closed in the product topology.
\end{lemma}
\begin{proof}
Let
\begin{align*}
& h \colon \D(X)\times\prod_{x\in X}\D(Y_\bot)\to \D(Y_\bot) &&
h(\mu,(\nu_x)_{x\in X}) = \sum_{x\in X} \mu(x)\nu_x
\end{align*}
We first show that $h$ is continuous in the product topology. We must to show that the preimage of any subbasic open set $(a,b)\times[0,1]^{Y_\bot\setminus y}$ of $[0,1]^{Y_\bot}$ is open in $\D(X)\times\prod_{x\in X}\D(Y_\bot)$. The preimage is
\begin{align*}
h^{-1}((a,b)\times[0,1]^{Y_\bot\setminus y})
&= \set{(\mu,(\nu_x)_{x\in X})}{h(\mu,(\nu_x)_{x\in X})\in(a,b)\times[0,1]^{Y_\bot\setminus y}}\\
&= \set{(\mu,(\nu_x)_{x\in X})}{\sum_{x\in X} \mu(x)\nu_x\in(a,b)\times[0,1]^{Y_\bot\setminus y}}\\
&= \set{(\mu,(\nu_x)_{x\in X})}{\sum_{x\in X} \mu(x)\nu_x(y)\in(a,b)}.
\end{align*}
If this set is nonempty, say $\sum_{x\in X} c(x)\mu_x(y)\in(a,b)$, let $0 < \eps < (b-a)/2$ and let $F\subs X$ be a large enough finite set that $\sum_{x\in F} c(x)\mu_x(y)\in(a+\eps,b-\eps)$. Let $\delta>0$ be such that $2\delta+\delta^2 < \eps/\len F$ and define
\begin{align*}
U &= \set{\mu\in \D(X)}{\forall x\in F\ \mu(x)\in(c(x)-\delta,c(x)+\delta)}\\
V_x &= \begin{cases}
\set{\nu\in f(x)}{\nu_x(y)\in(\mu_x(y)-\delta,\mu_x(y)+\delta)}, & \text{if $x\in F$,}\\
f(x), & \text{if $x\in X\setminus F$.}
\end{cases}
\end{align*}
Then $U\times \prod_{x\in F} V_x$ is an open subset of $\D(X)\times\prod_{x\in X} f(x)$ in the relative topology. If $(d,(\rho_x)_{x\in X}) \in U\times\prod_{x\in X} V_x$, then for all $x\in F$,
\begin{align*}
& c_x-\delta < d(x) < c(x)+\delta && \mu_x(y)-\delta < \rho_x(y) < \mu_x(y)+\delta.
\end{align*}
It follows that
\begin{align*}
& \sum_{x\in F}(c(x)-\delta)(\mu_x(y)-\delta) < \sum_{x\in F}d(x)\rho_x(y) < \sum_{x\in F}(c(x)+\delta)(\mu_x(y)+\delta)\\
&\Imp\ \sum_{x\in F}c(x)\mu_x(y) - (2\delta - \delta^2)\len F < \sum_{x\in F}d(x)\rho_x(y) < \sum_{x\in F}c(x)\mu_x(y) + (2\delta + \delta^2)\len F\\
&\Imp\ \sum_{x\in F}c(x)\mu_x(y) - \eps < \sum_{x\in F}d(x)\rho_x(y) < \sum_{x\in F}c(x)\mu_x(y) + \eps\\
&\Imp\ \sum_{x\in F}d(x)\rho_x(y) \in (a,b),
\end{align*}
therefore
\begin{align*}
U\times\prod_{x\in X} V_x &\subs \set{(\mu,\nu_x \mid x\in X)}{\sum_{x\in X} \mu(x)\nu_x(y)\in(a,b)}.
\end{align*}
We have shown that $h$ is continuous. The set $S\times\prod_{x\in X}f(x)$ is a closed subset of $\D(X)\times\prod_{x\in X}\D(Y_\bot)$ and its image under $h$ is $f^\dagger(S)$. The space $\D(X)\times\prod_{x\in X}\D(Y_\bot)$, being a closed subset of the compact space $[0,1]^X\times\prod_{x\in X}[0,1]^{Y_\bot}$, is itself compact in the relative topology. The space $[0,1]^{Y_\bot}$ is clearly Hausdorff, and it is well known that the image of a closed set under a continuous map from a compact space to a Hausdorff space is closed.
\end{proof}

\begin{lemma}
$f^\dagger(S)$ is up-closed.
\end{lemma}
\begin{proof}
Take any $\sum_{x\in\supp(\mu)} \mu(x)\nu_x \in S$ and let $\sum_{x\in\supp(\mu)} \mu(x)\nu_x \sqle_\D \xi$, so for all $y\in Y$:
\[
  \smashoperator{\sum_{x\in\supp(\mu)}} \mu(x)\cdotp\nu_x(y) \le \xi(y)
\]
Let $\set{a_x}{x\in X}$ be a collection of probability distributions on $Y_\bot$ such that the $a_x$ are $\sqle_\D$-maximal subject to
\begin{itemize}
\item
$\forall x\in X\ \ a_x(Y) \le \nu_x(\bot)$
\item
$\forall s\in Y\ \ \sum_x b_xa_x(s) \le (\mu - \sum_x b_x\nu_x)(s)$.
\end{itemize}
The distributions $a_x= \delta_\bot$ satisfy these constraints, so there exists a maximal one.
By maximality, one of the two constraints must be universally tight, that is, either
\begin{itemize}
\item
$\forall x\in X\ \ a_x(Y) = \nu_x(\bot)$, or
\item
$\forall s\in Y\ \ \sum_x b_xa_x(s) = (\mu - \sum_x b_x\nu_x)(s)$.
\end{itemize}

Let $\nu_x'(s) = \nu_x(s)+a_x(s)$ for $s\in Y$ and $\nu_x'(\bot) = \nu_x(\bot)-a_x(Y)$. For all $s\in Y$,
\begin{align*}
\sum_x b_xa_x(s) \le (\mu - \sum_x b_x\nu_x)(s)\ &\Imp\ \sum_x b_x\nu_x(s) + \sum_x b_xa_x(s) \le \mu(s)\\
&\Imp\ \sum_x b_x\nu_x'(s)\le\mu(s),
\end{align*}
so $\sum_x b_x\nu_x'\sqle\mu$.
If the first constraint is universally tight, then all $\nu_x'(\bot) = 0$, thus $\mu(\bot) \le \sum_x b_x\nu_x'(\bot) = 0$. Since $\sum_x b_x\nu_x'\sqle\mu$ and $\mu(\bot) = \sum_x b_x\nu_x'(\bot) = 0$, we have $\sum_x b_x\nu_x' = \mu$, so the second constraint is universally tight as well. Thus in either case, $\sum_x b_x\nu_x' = \mu$.

As $f(x)$ is up-closed and $\nu_x\sqle\nu_x'$, we have $\nu_x'\in f(x)$, thus $\mu = \sum_x b_x\nu_x'\in f^\dagger(S)$ by definition of $f^\dagger$. Since $\mu\sqsupseteq\sum_x b_x\nu_x$ was arbitrary, $f^\dagger(S)$ is up-closed.
\end{proof}

\subsection{Monad Laws}
\label{app:monad-laws}

\begin{theorem}[Left Identity]
$\eta^\dagger = \mathsf{id}$
\end{theorem}
\begin{proof}
For any $S\in\C(X)$:
\begin{align*}
  \eta^\dagger(S)
  &= \left\{
      \sum_{x \in\supp(\mu)} \mu(x)\cdotp \nu_x
       \;\;\middle|\;\;
       \mu \in S,
       \forall x\in\supp(\mu). \nu_x \in \eta_\bot(x)
   \right\}
   \intertext{Note that $\eta_\bot(x) = \eta(x) = \upset{\delta_x}$.}
  &= \left\{
      \sum_{x \in\supp(\mu)} \mu(x)\cdotp \nu_x
       \;\;\middle|\;\;
       \mu \in S,
       \forall x\in\supp(\mu). \nu_x \in \upset{\delta_x}
   \right\}
   \intertext{The set above contains every $\mu\in S$, as well as every $\mu'$ such that $\mu \sqsubseteq_\D \mu'$ for some $\mu \in S$. In other words, the set above is $\upcl S$, however, $S$ is already up-closed, so the set above is simply $S$.}
   &= \upcl S = S
\end{align*}
\end{proof}

\begin{theorem}[Right Identity]
$f^\dagger \circ \eta = f$
\end{theorem}
\begin{proof}
For any $f\colon X\to\C(Y)$ and $x\in X$:
\begin{align*}
  f^\dagger(\eta(x))
  &= \left\{
      \sum_{x \in\supp(\mu)} \mu(x)\cdotp \nu_x
       \;\;\middle|\;\;
       \mu \in \eta(x),
       \forall x\in\supp(\mu). \nu_x \in f_\bot(x)
   \right\}
   \intertext{Since $x\neq\bot$, then $\eta(x) = \{ \delta_x \}$ and $f_\bot(x) = f(x)$.}
  &= \{
      \nu_x
       \mid
       \nu_x \in f_\bot(x)
   \}
   = f(x)
\end{align*}
\end{proof}

\begin{lemma}\label{lem:convex-assoc}
For any $\mu \in \D(X)$, collection $(\nu_x)_{x\in \supp(\mu)}$ such that each $\nu_x \in \D(Y)$, and collection $(S_y)_{y \in Y}$ such that each $S_y \in \C(Z)$:
\[
  \left\{
    \smashoperator[r]{\sum_{x \in \supp(\mu)}} \mu(x)\cdot
    \smashoperator{\sum_{y\in\supp(\nu_x)}} \nu_x(y) \cdotp \nu'_{x, y}
    ~\Big|~
    \forall x, y.\ \nu'_{x, y} \in S_y
  \right\}
  =
  \left\{
    {\sum_{y \in Y}}
      \left( \smashoperator[r]{\sum_{x \in \supp(\mu)}} \mu(x) \cdotp \nu_x \right)(y)
      \cdotp \nu'_y
    ~\Big|~ 
    \forall x,y.\ \nu'_y \in S_y
  \right\}
\]
\end{lemma}
\begin{proof}
We prove the equality by establishing that each set is a subset of the other. We first prove the $\supseteq$ direction. Any element $\xi$ from the set on the left has the form $\sum_x \mu(x) \cdotp \sum_y \nu_x(y) \cdotp \nu'_{x, y}$ where $\nu'_{x, y} \in S_y$ for each $x$ and $y$. Now, for each $y\in Y$, we construct $\xi_y$ as follows:
\[
  \xi_y \triangleq \smashoperator[l]{\sum_{x\in\supp(\mu)}} \frac{\mu(x) \cdotp \nu_x(y)}{\sum_{x' \in\supp(\mu)} \mu(x')\cdotp \nu_{x'}(y)}\cdotp \nu'_{x,y}
\]
Since each $\nu'_{x,y} \in S_y$, then clearly $\xi_y$ is a convex combination of elements of $S_y$, and since $S_y$ is convex, then $\xi_y\in S_y$ for each $y$. Now, we have:
\begin{align*}
  \xi &= \smashoperator{\sum_{x\in\supp(\mu)}} \mu(x) \cdot \smashoperator{\sum_{y \in \supp(\nu_x)}} \nu_x(y) \cdot \nu'_{x, y}
  \\
  &= \sum_{y\in Y}
    \cancel{\smashoperator[r]{\sum_{x\in\supp(\mu)}} \mu(x) \cdot\nu_x(y)}
    \cdot
    \smashoperator[l]{\sum_{x\in\supp(\mu)}} \frac{\mu(x) \cdotp \nu_x(y)}{\cancel{\sum_{x' \in\supp(\mu)} \mu(x')\cdotp \nu_{x'}(y)}}\cdotp \nu'_{x,y}
  \\
  &= \sum_{y\in Y}\left(
    \smashoperator[r]{\sum_{x\in\supp(\mu)}} \mu(x) \cdot\nu_x \right)(y)
    \cdot \xi_y
\end{align*}
Therefore $\xi$ is contained in the second set too. Now, we prove the $\supseteq$ direction. Any element $\xi$ of the second set has the form $\xi = \sum_y (\sum_x\mu(x)\cdotp\nu_x)(y)\cdotp\nu'_y = \sum_x \mu(x)\cdotp \sum_y \nu_x(y)\cdotp\nu'_y$ where $\nu'_y \in S_y$ for each $y$. Now, for each $x$ and $y$, let $\xi_{x,y} = \nu'_y$, which is clearly in $S_y$. So, $\xi = \sum_x \mu(x)\cdotp \sum_y \nu_x(y)\cdotp\xi_{x,y}$, and therefore $\xi$ is in the first set.
\end{proof}

\begin{theorem}[Associativity]
$(f^\dagger \circ g)^\dagger = f^\dagger \circ g^\dagger$
\end{theorem}
\begin{proof}
For any $f\colon Y\to\C(Z)$, $g\colon X\to\C(Y)$, and $S \in \C(X)$:
\begin{align*}
  (f^\dagger \circ g)^\dagger(S)
  &= \left\{
      \sum_{x \in \supp(\mu)} \mu(x)\cdotp \mu'_x
      \;\;\Big|\;\;
      \mu \in S,
      \forall x.\ \mu'_x \in (f^\dagger \circ g)_\bot(x)
  \right\}
  \intertext{Clearly $(f^\dagger \circ g)_\bot(x) = f^\dagger(g(x))$ since $(f^\dagger \circ g)_\bot(x) = (f^\dagger \circ g)(x) = f^\dagger(g(x)) = f^\dagger(g_\bot(x))$ if $x\in X$, and $(f^\dagger \circ g)_\bot(\bot) = \upset{\delta_\bot} = f^\dagger(g_\bot(\bot))$. Expanding the second Kleisli extension, we get:}
  &= \left\{
      \smashoperator[r]{\sum_{x \in \supp(\mu)}} \mu(x)\cdotp \mu'_x
      \;\;\Big|\;\;
      \mu \in S,
      \forall x.\ \mu'_x \in 
      \left\{
          \smashoperator[r]{\sum_{y\in \supp(\nu)}} \nu(y)\cdotp \nu'_y
          \;\;\Big|\;\;
          \nu \in g_\bot(x),
          \forall y.\ \nu'_y \in f_\bot(y)
        \right\}
  \right\}
\\
  &= \left\{
      \smashoperator[r]{\sum_{x \in \supp(\mu)}} \mu(x)\cdot
        \smashoperator{\sum_{y \in \supp(\nu_x)}} \nu_x(y) \cdotp \nu'_{x,y}
      \;\;\Big|\;\;
      \mu \in S,
      \forall x.\
        \nu_x \in g_\bot(x),
        \forall y.\ \nu'_{x,y} \in f_\bot(y)
  \right\}
\intertext{By \Cref{lem:convex-assoc}.}
  &= \left\{
      {\sum_{y \in \bigcup_{x}\supp(\nu_x)}} \left(
        \smashoperator[r]{\sum_{x \in \supp(\mu)}} \mu(x)\cdotp \nu_{x}
        \right) \cdotp \nu'_y
        \;\;\Big|\;\;
        \mu \in S,
        \forall x.\ \nu_x \in g_\bot(x),
      \forall y.\ \nu'_y \in f_\bot(y)
  \right\}
  \\
  &= \left\{
      \smashoperator[r]{\sum_{y\in \supp(\nu)}} \nu(y)\cdotp \nu'_y
      \;\;\Big|\;\;
      \nu \in \left\{
        \smashoperator[r]{\sum_{x \in \supp(\mu)}} \mu(x)\cdotp \mu'_x
        \;\;\Big|\;\;
        \mu \in S,
        \forall x.\ \mu'_x \in g_\bot(x)
      \right\},
      \forall y.\ \nu'_y \in f_\bot(y)
  \right\}
  \\
  &= f^\dagger\left(\left\{
      \smashoperator[r]{\sum_{x \in \supp(\mu)}} \mu(x)\cdotp \mu'_x
      \;\;\Big|\;\;
      \mu \in S,
      \forall x.\ \mu'_x \in g_\bot(x)
  \right\}\right)
  \\
  &= f^\dagger(g^\dagger(S))
\end{align*}
\end{proof}

\subsection{Fixed Point Existence}
\label{app:fixpoint}

The function $f^\dagger$ is monotone with respect to $\sqle_\C^\bullet$: if $S\sqle_\C T$, then $f^\dagger(S)\sqle f^\dagger(T)$. This is true because $f^\dagger$ is monotone with respect to set inclusion: If $T\subs S$, then $f^\dagger(T)\subs f^\dagger(S)$, as is obvious from the definition.

Any $\sqle_\C$-directed set $D$ of elements of $\C(X)$ has a supremum $\sup D$, namely $\bigcap D$. This set is nonempty and closed, since it is the intersection of a family of closed sets with the finite intersection property in a compact space; convex, since the intersection of convex sets is convex; and up-closed, since the intersection of up-closed sets is up-closed. Thus $\C(X)$ forms a DCPO under $\sqle_\C$.

In the following we use $\Pi_{x\in X}$ to denote a cartesian product indexed by $x\in X$. For any $y \in \Pi_{x\in X} S_x$, we let $y_x$ be the $x^\text{th}$ projection of $y$, so that $y_x \in S_x$.

\begin{lemma}\label{lem:prod-inter}
\[
  \prod_{x\in X} \bigcap_{f\in D} f(x) = \bigcap_{f\in D} \prod_{x\in X} f(x)
\]
\end{lemma}
\begin{proof}
\begin{align*}
y \in \prod_{x\in X} \bigcap_{f\in D} f(x)
& \Leftrightarrow \forall x \in X. ~y_x \in \bigcap_{f\in D} f(x)
\\
& \Leftrightarrow \forall x \in X. ~\forall f \in D. ~ y_x \in f(x)
\\
& \Leftrightarrow \forall f \in D.~ \forall x \in X.~ y_x \in f(x)
\\
& \Leftrightarrow \forall f \in D. ~y \in \prod_{x\in X}  f(x)
\\
& \Leftrightarrow y \in \bigcap_{f\in D} \prod_{x\in X}  f(x)
\tag*\qedhere
\end{align*}
\end{proof}

\begin{lemma}\label{lem:bot-continuous}
$(-)_\bot \colon (X \to \C(Y)) \to X_\bot \to \C(Y)$ is Scott-continuous with respect to $\sqle^\bullet_\C$.
\end{lemma}
\begin{proof}
First, we show that $(-)_\bot$ is monotone. Suppose that $f \sqle_\C^\bullet g$. Now take any $x\in X_\bot$. If $x\in X$, then we have:
\[
f_\bot(x)
= f(x)
\sqle_\C g(x) = g_\bot(x)
\]
If instead $x=\bot$, then we have:
\[
  f_\bot(\bot) = \unit(\bot) = g_\bot(\bot)
\]
So in both cases $f_\bot \sqle_\C^\bullet g_\bot$. Now we show that $(-)_\bot$ preserves suprema. Let $D \subseteq X \to \C(Y_\bot)$ be a directed set, and take any $x\in X_\bot$. If $x\in X$, then we have:
\[
  \sup_{f\in D} f_\bot(x) = \sup_{f\in D} f(x) = (\sup D)(x) = (\sup D)_\bot(x)
\]
If, alternatively, $x = \bot$, then we have:
\[
  \sup_{f\in D} f_\bot(\bot)
  = \sup_{f\in D} \unit(\bot)
  = \unit(\bot)
  = (\sup D)_\bot(\bot)
\]
So, in either case $\sup_{f \in D} f_\bot(x) = (\sup D)_\bot(x)$.
\end{proof}

\begin{lemma}\label{lem:kleisli-continuous}
$(-)^\dagger \colon (X \to \C(Y)) \to \C(X) \to \C(Y)$ is Scott-continuous with respect to $\sqle^\bullet_\C$ in its first argument.
\end{lemma}
\begin{proof}
\begin{align*}
(\sup D)^\dagger(S)
&= \set{\smashoperator[r]{\sum_{x\in \supp(\mu)}}\mu(x)\cdotp\nu_x}{\mu\in S,\ \nu\in\prod_{x\in X} (\sup D)_\bot(x)}\\
\intertext{By \Cref{lem:bot-continuous}.}
&= \set{\smashoperator[r]{\sum_{x\in  \supp(\mu)}}\mu(x)\cdotp\nu_x}{\mu\in S,\ \nu\in\prod_{x\in X} \sup_{f\in D} f_\bot(x)}\\
&= \set{\smashoperator[r]{\sum_{x\in  \supp(\mu)}}\mu(x)\cdotp\nu_x}{\mu\in S,\ \nu\in\prod_{x\in X} \bigcap_{f\in D} f_\bot(x)}\\
\intertext{By \Cref{lem:prod-inter}.}
&= \set{\smashoperator[r]{\sum_{x\in  \supp(\mu)}}\mu(x)\cdotp\nu_x}{\mu\in S,\ \nu\in\bigcap_{f\in D}\prod_{x\in X} f_\bot(x)}\\
&= \set{\smashoperator[r]{\sum_{x\in  \supp(\mu)}}\mu(x)\cdotp\nu_x}{\mu\in S,\ \forall f\in D. ~ \nu\in\prod_{x\in X} f_\bot(x)}\\
&= \set{\smashoperator[r]{\sum_{x\in  \supp(\mu)}}\mu(x)\cdotp\nu_x}{\forall f\in D. ~ (\mu\in S,\ \nu\in\prod_{x\in X} f_\bot(x))}\\
&= \bigcap_{f\in D} \set{\smashoperator[r]{\sum_{x\in  \supp(\mu)}}\mu(x)\cdotp\nu_x}{\mu\in S,\ \nu\in\prod_{x\in X} f_\bot(x)}\\
&= \bigcap_{f\in D} f^\dagger(S)
= \sup_{f\in D} f^\dagger(S)
\end{align*}
The third-to-last inference is justified by the logical rule $\phi\wedge\forall x\ \psi \equiv \forall x\ (\phi\wedge\psi)$ provided $x$ is not free in $\phi$.
\end{proof}

\begin{lemma}[Scott Continuity]\label{lem:scott-continuity}
The function $\Phi_{\langle C, e\rangle} \colon \Sigma\to\C(\Sigma)$ is Scott continuous with respect to $\sqsubseteq_{\C}^\bullet$ (the pointwise order).
\end{lemma}
\begin{proof}
Suppose that $D \subseteq \C(\Sigma)$ is a directed set, then for any $\sigma\in\Sigma$ we have:
\begin{align*}
\sup_{f\in D} \Phi_{\langle C, e\rangle}(f)(\sigma)
&= \sup_{f\in D}\left\{
  \begin{array}{ll}
    f^\dagger(\de{C}(\sigma)) & \text{if} ~ \de{e}(\sigma) = \tru \\
    \eta(\sigma) & \text{if}~ \de{e}(\sigma) = \fls
  \end{array}
\right.
\\
\intertext{Since $f$ is not free in $e$ or $\sigma$.}
&= \left\{
  \begin{array}{ll}
    \sup_{f\in D} f^\dagger(\de{C}(\sigma)) & \text{if} ~ \de{e}(\sigma) = \tru \\
    \eta(\sigma) & \text{if}~ \de{e}(\sigma) = \fls
  \end{array}
\right.
\\
\intertext{By \Cref{lem:kleisli-continuous}.}
&= \left\{
  \begin{array}{ll}
    (\sup D)^\dagger(\de{C}(\sigma)) & \text{if} ~ \de{e}(\sigma) = \tru \\
    \eta(\sigma) & \text{if}~ \de{e}(\sigma) = \fls
  \end{array}
\right.
\\
&= \Phi_{\langle C, e\rangle}(\sup D)(\sigma)
\end{align*}
\end{proof}

\subsection{Countability of State Space}

\begin{lemma}
For any program $C$ and initial state $\sigma\in\Sigma$, $\de{C}(\sigma)$ can only reach countably many states. That is, $\tau\in \Sigma_\bot$ for any distribution $\mu \in \de{C}(\sigma)$ and any $\tau\in\supp(\mu)$, and $\Sigma_\bot$ is a countable set.
\end{lemma}
\begin{proof}
We first recall that:
\[
  \Sigma_\bot = \Sigma \cup \{\bot\} = (\mathsf{Var} \to \mathsf{Val}) \cup \{\bot\}
\]
Recall from \Cref{sec:prelim} that $\mathsf{Val}$ contains Booleans and rationals, so it is a countable set. We later add finite length lists in \Cref{app:list-sem}, which preserves countability. The set $\mathsf{Var}$ is finite, making $\Sigma_\bot$ as a whole countable.

The proof is by induction of the structure of $C$. In the base case where $C=\skp$, the claim follows trivially, since the only reachable state is $\sigma$. If $C = x\coloneqq e$, then the program reaches the state $\sigma[ x \coloneqq \de{e}(\sigma) ]$. We presume that $e$ can only perform Boolean logic, rational arithmetic, and operations on finite length lists (see \Cref{app:list-sem}), meaning that $\de{e}(\sigma) \in \mathsf{Val}$, and so $\sigma[ x \coloneqq \de{e}(\sigma) ] \in \Sigma$.

For sequential composition, we have:
\[
  \de{C_1 \fatsemi C_2}(\sigma) =
  \left\{
    \sum_{\tau\in\supp(\mu)} \mu(\tau) \cdot \nu_\tau
    \;\; \Big| \;\;
    \mu\in\de{C_1}(\sigma),
    \forall\tau\in\supp(\mu).~
    \nu_\tau \in \de{C_2}(\tau)
  \right\}
\]
Bu the induction hypothesis, we know that $\supp(\mu) \subseteq \Sigma_\bot$ for each $\mu \in \de{C_1}(\sigma)$. By the induction hypothesis again, we get that $\supp(\nu_\tau) \subseteq \Sigma_\bot$ for each $\tau\in\supp(\mu)$. Since any distribution in $\de{C_1 \fatsemi C_2}(\sigma)$ is obtained from combinations of those $\nu_\tau$, all the states are again contained in $\Sigma_\bot$.

In the other two composition operators $\de{C_1\oplus_p C_2}(\sigma)$ and $\de{C_1\nd C_2}(\sigma)$ and if statements, we use the induction hypothesis to conclude that the claim holds for $\de{C_1}(\sigma)$ and $\de{C_2}(\sigma)$. The final distributions are then convex combinations of those smaller distributions, so the claim again holds.

Finally, for loops we have:
\[
  \de{\whl eC}(\sigma)
  = \bigcap_{n\in\mathbb{N}} \Phi^n_{\langle C,e\rangle}(\bot_\C^\bullet)(\sigma)
  = \bot_\C^\bullet(\sigma) \cap \Phi_{\langle C,e\rangle}(\bot_\C^\bullet)(\sigma) \cap \cdots
  \subseteq \D(\Sigma_\bot)
\]
More precisely, since the semantics involves taking an intersection with $\bot_\C^\bullet(\sigma) = \D(\Sigma_\bot)$, we know that the resulting set cannot contain terms outside of $\D(\Sigma_\bot)$, meaning that the support of every distribution in this sequence must by contained in $\Sigma_\bot$.
\end{proof}

\section{Logical Rules}

\subsection{Nondeterministic and Probabilistic Choice}

\begin{lemma}\label{lem:dagger-upset}
For any function $f\colon \Sigma\to \C(\Sigma)$ and set $S\in\C(\Sigma)$:
\[
  f^\dagger(S) = \bigcup_{\mu\in S} f^\dagger(\upset\mu)
\]
\end{lemma}
\begin{proof}
\begin{align*}
  f^\dagger(S)
  &= \left\{ 
    \smashoperator[r]{\sum_{\sigma\in\supp(\mu)}} \mu(\sigma)\cdot \nu_\sigma
    \;\;\Big|\;\;
    \mu \in S,
    \forall \sigma\in\supp(\mu).~
    \nu_\sigma \in f_\bot(\sigma)
  \right\}
\\
  &= \bigcup_{\mu\in S}\left\{
    \smashoperator[r]{\sum_{\sigma\in\supp(\mu)}} \mu(\sigma)\cdot \nu_\sigma
    \;\;\Big|\;\;
    \forall \sigma\in\supp(\mu).~
    \nu_\sigma \in f_\bot(\sigma)
  \right\}
\\
\intertext{Now, since $S$ is up-closed, then $\upset\mu \subseteq S$ for each $\mu\in S$. Since set union is idempotent, we can expand the union as follows:}
  &= \bigcup_{\mu\in S} \bigcup_{\mu' \in \upset\mu}\left\{
    \smashoperator[r]{\sum_{\sigma\in\supp(\mu')}} \mu'(\sigma)\cdot \nu_\sigma
    \;\;\Big|\;\;
    \forall \sigma\in\supp(\mu').~
    \nu_\sigma \in f_\bot(\sigma)
  \right\}
\\
  &= \bigcup_{\mu\in S} \left\{
    \smashoperator[r]{\sum_{\sigma\in\supp(\mu)}} \mu(\sigma)\cdot \nu_\sigma
    \;\;\Big|\;\;
    \mu' \in \upset\mu,
    \forall \sigma\in\supp(\mu').~
    \nu_\sigma \in f_\bot(\sigma)
  \right\}
\\
  &= \bigcup_{\mu\in S} f^\dagger(\upset\mu)
\end{align*}
\end{proof}

\begin{lemma}
\label{lem:oplus_dist}
For any function $f \colon X\to\C(Y)$, probability $p\in[0,1]$, and sets $S,T\in \C(X)$.
\[
  f^\dagger(S \oplus_p T) = f^\dagger(S) \oplus_p f^\dagger(T)
\]
\end{lemma}
\begin{proof}
Let $\xi \in \D(\{1, 2\})$ be a distribution where $\xi(1)= p$ and $\xi(2)= 1-p$. Let $g \colon \{1,2\} \to\C(X)$ be defined as $g(1)= S$ and $g(2) = T$. Now, we have:
\begin{align*}
  f^\dagger(S \oplus_p T)
  &= f^\dagger( \{ p\cdot\mu + (1-p)\cdot \nu \mid \mu \in S, \nu\in T\})
  \\
  &= f^\dagger(\{ \xi(1)\cdot \mu_1 + \xi(2) \cdot \mu_2 \mid \forall i\in\{1,2\}.\ \mu_i \in g(i) \})
  \\
  &= f^\dagger(g^\dagger(\{\xi\}))
  \intertext{By the monad laws:}
  &= (f^\dagger \circ g)^\dagger(\{\xi\})
  \\
  &= \{ \xi(1)\cdot \mu_1 + \xi(2)\cdot\mu_2 \mid\forall i\in\{1,2\}.\ \mu_i \in f^\dagger(g(i)) \}
  \\
  &= \{ p \cdot \mu_1 + (1-p)\cdot \mu_2 \mid \mu_1 \in f^\dagger(S), \mu_2 \in f^\dagger(T) \}
  = f^\dagger(S) \oplus_p f^\dagger(T)
\end{align*}

\end{proof}

%

\begin{lemma}
\label{lem:bigger-prob}
Let $\mu = p\cdot\mu_1 + (1-p) \cdot \mu_2$.
If $\mu \sqsubseteq_\D \nu$, then there exists $\nu_1$ and $\nu_2$ such that $\nu = p\cdot \nu_1 + (1-p)\cdot \nu_2$ and $\mu_i \sqsubseteq_\D \nu_i$ for $i\in\{1,2\}$.
\end{lemma}
\begin{proof}
If $\mu(\bot) = 0$, the claim is trivial since $\mu = \nu$ and thus we can simply let $\nu_i = \mu_i$ for $i\in\{1, 2\}$. If instead $\mu(\bot) > 0$, then for every $\sigma\in\Sigma$ and $i\in\{1,2\}$, let:
\[
  \nu_i(\sigma) \triangleq \mu_i(\sigma) + \frac{\mu_i(\bot)}{\mu(\bot)}\cdot(\nu(\sigma) - \mu(\sigma))
\]
First, we must establish that $\nu_i$ is a probability distribution. Note that $\nu_i(\sigma)$ is nonnegative for all $\sigma\in \Sigma$, since $\nu(\sigma) \ge \mu(\sigma)$, and thus the two terms being summed are both nonnegative. In addition, we have:
\begin{align*}
  \nu_i(\bot)
  = \mu_i(\bot) + \frac{\mu_i(\bot)}{\mu(\bot)}\cdot(\nu(\bot) - \mu(\bot))
  = \mu_i(\bot)\cdot\left(\cancel1 + \frac{\nu(\bot)}{\mu(\bot)} - \cancel1\right)
  = \mu_i(\bot)\cdot\frac{\nu(\bot)}{\mu(\bot)}
\end{align*}
Which is clearly nonnegative too. Now, we also have:
\begin{align*}
  |\nu_i|
  &= \smashoperator{\sum_{\sigma\in \supp(\nu_i)}} \nu_i(\sigma)
  \\
  \intertext{We can expand the bounds of the sum, since $\nu_i(\tau) = 0$ for any $\tau\notin\supp(\nu_i)$, and therefore including the extra terms will not affect the value.}
  &= \smashoperator{\sum_{\sigma\in \Sigma_\bot}} \nu_i(\sigma)
  \\
  &= \smashoperator{\sum_{\sigma\in \Sigma_\bot}} \mu_i(\sigma) + \frac{\mu_i(\bot)}{\mu(\bot)}\cdot(\nu(\sigma) - \mu(\sigma))
  \\
  &= \left({\sum_{\sigma\in \Sigma_\bot}} \mu_i(\sigma)\right) + \frac{\mu_i(\bot)}{\mu(\bot)}\cdot
    \left(\left({\sum_{\sigma\in \Sigma_\bot}}\nu(\sigma)\right) -
    \left({\sum_{\sigma\in \Sigma_\bot}}\mu(\sigma)\right)\right)
  \\
  &= 1 + \frac{\mu_i(\bot)}{\mu(\bot)}\cdot(1 - 1)
  = 1 + \frac{\mu_i(\bot)}{\mu(\bot)}\cdot 0
  = 1
\end{align*}
So, since all the points are nonnegative and the total mass of $\nu_i$ is 1, then all the points are in-bounds and $\nu_i$ is a valid distribution. In addition, since $\nu(\sigma) \ge\mu(\sigma)$ for all $\sigma\in\Sigma$, then $\nu(\sigma) -\mu(\sigma) \ge 0$ and so:
\[
  \nu_i(\sigma)
  = \mu_i(\sigma) + \frac{\mu_i(\bot)}{\mu(\bot)} \cdot (\nu(\sigma) - \mu(\sigma))
  \ge \mu_i(\sigma) + \frac{\mu_i(\bot)}{\mu(\bot)} \cdot 0
  = \mu_i(\sigma)
\]
So we have that $\mu_i\sqsubseteq_\D \nu_i$. And finally, for any $\sigma\in \Sigma_\bot$:
\begin{align*}
  &p\cdot \nu_1(\sigma) + (1-p)\cdot\nu_2(\sigma)
  \\
  &= p\cdot \left(\mu_1(\sigma) + \frac{\mu_1(\bot)}{\mu(\bot)}\cdot(\nu(\sigma) - \mu(\sigma))\right)
  + (1-p)\cdot \left(\mu_2(\sigma) + \frac{\mu_2(\bot)}{\mu(\bot)}\cdot(\nu(\sigma) - \mu(\sigma))\right)
  \\
  &= \left(p\cdot\mu_1(\sigma) + (1-p)\cdot\mu_2(\sigma)\right)
      + (p\cdot\mu_1(\bot)+(1-p)\cdot\mu_2(\bot))\cdot \frac{\nu(\sigma) - \mu(\sigma)}{\mu(\bot)}
  \\
  &= \mu(\sigma) + \cancel{\mu(\bot)} \cdot \frac{\nu(\sigma) - \mu(\sigma)}{\cancel{\mu(\bot)}}
  = \mu(\sigma) + \nu(\sigma) - \mu(\sigma)
  = \nu(\sigma)
\end{align*}
\end{proof}

\begin{lemma}
\label{lem:upcl-prob}
\[
  \upset{p\cdot \mu_1 + (1-p)\cdot\mu_2} = \upset{\mu_1} \oplus_p \upset{\mu_2}
\]
\end{lemma}
\begin{proof}
First, we will show that $\upset{p\cdot \mu_1 + (1-p)\cdot\mu_2} \subseteq \upset{\mu_1} \oplus_p \upset{\mu_2}$. Take some $\nu \in \upset{p\cdot \mu_1 + (1-p)\cdot\mu_2}$. By \Cref{lem:bigger-prob} we know there must be $\nu_1$ and $\nu_2$ such that $\mu_i \sqsubseteq \nu_i$ and $\nu = p\cdot\nu_1 +(1-p)\cdot\nu_2$. So, $\nu_i \in \upset{\mu_i}$, and therefore also $\nu \in \upset{\mu_1} \oplus_p \upset{\mu_2}$.

Now, we show that $\upset{\mu_1} \oplus_p \upset{\mu_2}\subseteq \upset{p\cdot \mu_1 + (1-p)\cdot\mu_2}$. Take some $\nu \in \upset{\mu_1} \oplus_p \upset{\mu_2}$. So, there is a $\nu_1 \in \upset{\mu_1}$ and $\nu_2\in\upset{\mu_2}$ such that $\nu = p\cdot \nu_1 + (1-p)\cdot\nu_2$. Take any $\sigma\in\Sigma$, we have that:
\begin{align*}
  \nu(\sigma) &= p\cdot\nu_1(\sigma) + (1-p)\cdot\nu_2(\sigma)
  \\
  & \ge p\cdot\mu_1(\sigma) + (1-p)\cdot\mu_2(\sigma)
  \\
  & = (p\cdot \mu_1 + (1-p)\cdot\mu_2)(\sigma)
\end{align*}
Therefore $p\cdot \mu_1 + (1-p)\cdot\mu_2 \sqsubseteq_\D \nu$ and so $\nu \in \upset{p\cdot \mu_1 + (1-p)\cdot\mu_2}$.
\end{proof}

\begin{lemma}
\label{lem:convex_sum_idem}
For any sequence of distributions $(\mu_i)_{i\in I}$, any $\xi \in \D(I)$, and any assertion $\varphi$
, if $\mu_i\vDash \varphi$ for each $i\in I$, then $\sum_{i\in I} \xi(i)\cdot \mu_i \vDash \varphi$.
\end{lemma}
\begin{proof}
By induction on the structure of $\varphi$
\begin{itemize}
\item $\varphi=\top$. Trivial since $\sum_{i\in I}\xi(i)\cdot\mu_i \vDash\top$ always.

\item $\varphi = \bot$. Vacuous since $\mu_i\vDash\bot$ is impossible.

\item $\varphi = \varphi_1 \land\varphi_2$. We know that $\mu_i \vDash \varphi_1$ and $\mu_i\vDash\varphi_2$ for each $i\in I$. So, by the induction hypothesis, $\sum_{i\in I}\xi(i)\cdot\mu_i \vDash\varphi_1$ and $\sum_{i\in I}\xi(i)\cdot\mu_i \vDash\varphi_2$, therefore $\sum_{i\in I}\xi(i)\cdot\mu_i \vDash\varphi_1\land\varphi_2$.


\item $\varphi = \varphi_1\nd\varphi_2$. So for each $i\in I$, there exists $(\nu_{i, j})_{j\in \{1,2\}}$ and $p_i \in [0,1]$ such that $\mu_i = p_i \cdot \nu_{i,1} + (1-p_i)\cdot \nu_{i,2}$ and $\nu_{i,j}\vDash\varphi_j$ for each $j\in \{1,2\}$. 
Now, let $\xi_1,\xi_2 \in \D(I)$ be defined as follows:
\[
  \xi_1(i) \triangleq \frac{p_i\cdot \xi(i)}{\sum_{k\in I} p_k \cdot \xi(k)}
  \qquad
  \xi_2(i) \triangleq \frac{(1-p_i)\cdot \xi(i)}{\sum_{k\in I} (1-p_k) \cdot \xi(k)}
\]
And by the induction hypothesis, we know that $\sum_{i\in I} \xi_j(i)\cdot \nu_{i, j} \vDash \varphi_j$.
Now, observe that:
\begin{align*}
\sum_{i\in I}\xi(i) \cdot \mu_i
&= \sum_{i\in I}\xi(i) \cdot (p_i\cdot \nu_{i,1} + (1-p_i)\cdot \nu_{i,2})
\\
&= \sum_{i\in I} p_i\cdot\xi(i) \cdot \nu_{i,1} +
    \sum_{i\in I} (1-p_i)\cdot \xi(i)\cdot \nu_{i,2}
\\
&= \left(\sum_{k\in I} p_k\cdot \xi(k)\right)\cdot \sum_{i\in I} \xi_1(i) \cdot \nu_{i,1} +
    \left(\sum_{k\in I} (1-p_k)\cdot \xi(k)\right)\cdot \sum_{i\in I} \xi_2(i)\cdot \nu_{i,2}
\end{align*}
Therefore $\sum_{i\in I}\xi(i) \cdot \mu_i \vDash \varphi_1 \nd \varphi_2$.


\item $\varphi = \varphi_1\oplus_p\varphi_2$. We know that $\mu_i = p\cdot\mu_{i,1} +(1-p)\cdot\mu_{i,2}$ such that $\mu_{i,j}\vDash\varphi_j$ for $j\in\{1,2\}$ and $i\in I$. By the induction hypothesis, we therefore get that $\sum_{i\in I} \xi(i)\cdot \mu_{i,j}\vDash \varphi_j$ for each $j\in \{1,2\}$. Now, we have:
\begin{align*}
\sum_{i\in I} \xi(i)\cdot \mu_i
&= \sum_{i\in I} \xi(i) \cdot (p\cdot \mu_{i,1} + (1-p)\cdot \mu_{i,2}
\\
&= p\cdot\sum_{i\in I} \xi(i) \cdot \mu_{i,1} + (1-p)\cdot \sum_{i\in I} \xi(i) \cdot\mu_{i,2})
\end{align*}
Therefore $\sum_{i\in I} \xi(i)\cdot \mu_i \vDash \varphi_1 \oplus_p \varphi_2$.

\item $\varphi=P$. We know that $\supp(\mu_i) \subseteq \sem P$ for each $i\in I$. Therefore, we have:
\[
\supp(\sum_{i\in I}\xi(i)\cdot \mu_i)
= \bigcup_{i\in I} \supp(\mu_i)
\subseteq \bigcup_{i\in I} \sem P
= \sem P
\]
Therefore $\sum_{i\in I}\xi(i)\cdot \mu_i \vDash P$.\qedhere
\end{itemize}
\end{proof}

\subsection{Soundness of the Sequential Proof Rules}
\label{app:soundness}

\begin{lemma}[Monotonicity] \label{lem:up-satisfy}
If $\mu\vDash\varphi$, then $\nu\vDash\varphi$ for any $\nu\in\D(\Sigma_\bot)$ such that $\mu\sqsubseteq_{\D} \nu$.
\end{lemma}
\begin{proof}
By induction on the structure of $\varphi$.
\begin{itemize}

\item $\varphi = \top$. Trivial since $\nu\vDash\top$ always.

\item $\varphi = \bot$. Vacuous, since $\mu\vDash\bot$ is impossible.

\item $\varphi = \varphi_1 \land \varphi_2$. We know that $\mu\vDash\varphi_1$ and $\mu\vDash\varphi_2$. Now take any $\nu$ such that $\mu\sqsubseteq_\D \nu$. By the induction hypothesis, we know that $\nu\vDash\varphi_1$ and $\nu\vDash\varphi_2$, therefore $\nu\vDash\varphi_1\land\varphi_2$.

\item $\varphi = \varphi_1 \oplus_p\varphi_2$. We know that $\mu = p\cdot\varphi_1 + (1-p)\cdot\varphi_2$ such that $\mu_i\vDash\varphi_i$ for $i\in\{1,2\}$. Now take any $\nu$ such that $\mu\sqsubseteq_\D\nu$. By \Cref{lem:bigger-prob} we know that there is $\nu_1$ and $\nu_2$ such that $\nu = p\cdot\nu_1 + (1-p)\cdot\nu_2$ and $\nu_i(\sigma) \ge \mu_i(\sigma)$ too. So, we can use the induction hypothesis to conclude that $\nu_i\vDash\varphi_i$, and so $\nu\vDash \varphi_1\oplus_p\varphi_2$.

\item $\varphi = \varphi_1\nd\varphi_2$. We know that $\mu \vDash\varphi_1\oplus_p\varphi_2$ for some $p\in[0,1]$. Now take any $\nu$ such that $\mu\sqsubseteq_\D\nu$. By the previous case, we know that $\nu\vDash\varphi_1\oplus_p\varphi_2$. Since $\varphi_1\oplus_p\varphi_2 \Rightarrow \varphi_1\nd\varphi_2$, then $\nu\vDash\varphi_1\nd\varphi_2$.

\item $\varphi = \sure{P}$. Since $\supp(\mu)\subseteq \sem{P}\subseteq \Sigma$, we know that $\bot\notin\supp(\mu)$. Therefore, if $\mu\sqsubseteq_\D\nu$, then $\mu = \nu$, since there is not distribution strictly larger than $\mu$, and so $\nu\vDash \sure P$.\qedhere

\end{itemize}
\end{proof}

\begin{lemma}
\label{lem:assign}
If $\mu\vDash \varphi[e/x]$ then $\nu \vDash \varphi$ for all $\nu\in\dem{x\coloneqq e}{\upset{\mu}}$.
\end{lemma}
\begin{proof}
By induction on the structure of $\varphi$.

\begin{itemize}
\item $\varphi = \top$. it is trivial that $\nu \vDash \top$ for any $\nu$.
\item $\varphi = \bot$. This case is vacuous, since $\bot[e/x] = \bot$ and therefore the premise that $\mu\vDash\bot[e/x]$ is impossible.

\item $\varphi = \varphi_1 \land\varphi_2$. We know that $(\varphi_1 \land \varphi_2)[e/x] = \varphi_1[e/x]\land \varphi_2[e/x]$, and so $\mu\vDash\varphi_i[e/x]$ for $i\in\{1,2\}$. Now, take any $\nu\in\dem{x\coloneqq e}{\upset\mu}$. By the induction hypothesis, we get that $\nu \vDash \varphi_i$ for $i\in\{1,2\}$, therefore $\nu\vDash\varphi_1\land\varphi_2$.

\item $\varphi = \varphi_1 \oplus_p\varphi_2$. We know that $\mu = p\cdot\mu_1 + (1-p)\cdot\mu_2$ such that $\mu_i\vDash\varphi_i[e/x]$ for $i\in\{1,2\}$. Now, by \Cref{lem:upcl-prob,lem:oplus_dist}. we get that:
\[
  \dem{x\coloneqq e}{\upset\mu}
  = \dem{x\coloneqq e}{\upset{\mu_1} \oplus_p \upset{\mu_2}}
  = \dem{x\coloneqq e}{\upset{\mu_1}} \oplus_p \dem{x\coloneqq e}{\upset{\mu_2}}
\]
And therefore any $\nu \in \dem{x\coloneqq e}{\upset{\mu}}$ must have the form $p\cdot\nu_1 + (1-p) \cdot\nu_2$ where $\nu_i \in \de{x\coloneqq e}^\dagger_\bot(\{\mu_i\})$ for $i\in\{1,2\}$. So, by the induction hypothesis, $\nu_i\vDash\varphi_i$, and therefore $\nu\vDash\varphi_1\oplus_p\varphi_2$.

\item $\varphi = \varphi_1 \nd\varphi_2$. We know that there is some $p$ such that $\mu \vDash (\varphi_1 \oplus_p \varphi_2)[e/x]$, so by the previous case, we get that for any $\nu \vDash \varphi_1 \oplus_p \varphi_2$ for any $\nu\in\dem{x\coloneqq e}{\upset{\mu}}$. Since $\varphi_1\oplus_p\varphi_2 \Rightarrow \varphi_1\nd\varphi_2$, then $\nu\vDash\varphi_1\nd\varphi_2$ as well.

\item $\varphi=\sure{P}$. We know that $\supp(\mu)\subseteq\sem{P[e/x]}$. Now, take any $\nu\in\dem{x\coloneqq e}{\upset\mu}$. Any $\tau\in \supp(\nu)$ must have the form $\sigma[x \coloneqq \de{e}(\sigma)]$ for some $\sigma\in\supp(\mu)$. Since $\sigma\in\sem{P[e/x]}$, it must be that $\sigma[x \coloneqq \de{e}(\sigma)] \in \sem{P}$.\qedhere

\end{itemize}
\end{proof}

\thmsoundness*
\begin{proof}
By induction on the derivation.

\begin{itemize}

\item\ruleref{Skip}. Suppose $\mu\vDash\varphi$. Note that $\de{\skp} = \eta$ and so $\de{\skp}^\dagger = \mathsf{id}$, therefore $\dem\skp{\upset\mu} = \upset\mu$.
Now, take any $\nu \in \upset\mu$, clearly $\mu\sqsubseteq_\D \nu$, so $\nu\vDash\varphi$ by \Cref{lem:up-satisfy}.

\item\ruleref{Assign}. Follows from \Cref{lem:assign}.

\item \ruleref{Seq}. By the induction hypothesis, we know that $\vDash\triple\varphi{C_1}\vartheta$ and $\vDash\triple\vartheta{C_2}\psi$ and we have to show that $\vDash\triple\varphi{C_1\fatsemi C_2}\psi$. Suppose that $\mu\vDash\varphi$, and observe that:
\begin{align*}
\dem{C_1\fatsemi C_2}{\upset\mu}
&= (\de{C_2}^\dagger_\bot \circ \de{C_1}_\bot)^\dagger(\upset\mu)
\\
&= \dem{C_2}{\dem{C_1}{\upset\mu}}\\
&= \smashoperator{\bigcup_{\nu\in \dem{C_1}{\upset\mu}}} \dem{C_2}{\upset\nu} \tag{By \Cref{lem:dagger-upset}.}
\end{align*}
So, for any $\xi \in \dem{C_1\fatsemi C_2}{\upset\mu}$, there must be a $\nu \in \dem{C_1}{\upset\mu}$ such that $\xi\in\dem{C_2}{\upset\nu}$. From $\vDash\triple\varphi{C_1}\vartheta$, we know that $\nu\vDash\vartheta$ and therefore from $\vDash\triple\vartheta{C_2}\psi$ we get that $\xi\vDash\psi$.

\item \ruleref{Prob}. Suppose $\mu\vDash\varphi$, therefore since $\varphi \Rightarrow (e = p)$, we know that $\de{e}(\sigma) = p$ for every $\sigma\in\supp(\mu)$.
Note that this also means that $\bot\notin\supp(\mu)$ and so $\upset\mu = \{\mu\}$.
Now, we have:
\begin{align*}
&\dem{C_1\oplus_p C_2}{\upset\mu}\\
&= \dem{C_1\oplus_p C_2}{\{\mu\}}
\\
&= \set{
    \smashoperator[r]{\sum_{\sigma\in\supp(\mu)}} \mu(\sigma)\cdot\nu_\sigma
    }{
    \forall \sigma\in \supp(\mu).~
    \nu_\sigma \in \de{C_1\oplus_e C_2}_\bot(\sigma)
 }
\\
&=  \set{
  \smashoperator[r]{\sum_{\sigma\in\supp(\mu)}} \mu(\sigma)\cdot
      (\de{e}(\sigma)\cdot \nu_\sigma + (1-\de{e}(\sigma))\cdot \nu'_\sigma)
  }{
    \forall \sigma.~
    \nu_\sigma \in \de{C_1}_\bot(\sigma),
    \nu'_\sigma\in\de{C_2}_\bot(\sigma)
}
\\
&= \set{
    p\cdot\smashoperator{\sum_{\sigma\in\supp(\mu)}} \mu(\sigma)\cdot \nu_\sigma +
    (1-p)\cdot\smashoperator{\sum_{\sigma\in\supp(\mu)}} \mu(\sigma) \cdot \nu'_\sigma 
   }{
    \forall \sigma.~
    \nu_\sigma \in \de{C_1}_\bot(\sigma),
    \nu'_\sigma\in\de{C_2}_\bot(\sigma)
}
\\
&= \set{
    p\cdot\nu + (1-p)\cdot\nu'
    }{
    \nu \in \dem{C_1}{\upset{\mu}},
    \nu'\in\dem{C_2}{\upset{\mu}}
}
\end{align*}
So for each $\xi \in \dem{C_1\oplus_p C_2}{\upset\mu}$ there is a $\nu \in \dem{C_1}{\upset{\mu}}$ and $\nu'\in\dem{C_2}{\upset{\mu}}$ such that $\xi = p\cdot\nu +(1-p)\cdot \nu'$. We know that $\mu\vDash\varphi$,so by the induction hypothesis $\nu \vDash\psi_1$ and $\nu'\vDash\psi_2$, therefore $\xi \vDash \psi_1 \oplus_p \psi_2$.

\item \ruleref{Nondet}. Suppose $\mu\vDash\sure{P}$, therefore $\supp(\mu) \subseteq \sem P$, and so $\{\delta_\sigma\} \vDash \sure{P}$ for each $\sigma\in \supp(\mu)$. By the induction hypotheses, we also get that $\de{C_i}(\sigma)\vDash\psi_i$ for each $\sigma\in\supp(\mu)$ and $i\in\{1,2\}$. This also means that $\nu\vDash \psi_1 \nd\psi_2$ for each $\nu\in \de{C_1\nd C_2}(\sigma)$.
Now, we have:
\begin{align*}
\dem{C_1\nd C_2}{\upset\mu}
&= \set{
    \smashoperator[r]{\sum_{\sigma\in\supp(\mu')}} \mu'(\sigma)\cdot\nu_\sigma
    }{
    \mu'\in\upset\mu,~
    \forall \sigma\in \supp(\mu').~
    \nu_\sigma \in \de{C_1\nd C_2}_\bot(\sigma)
}
\end{align*}
So, taking any $\xi\in\dem{C_1\nd C_2}{\upset\mu}$, we know that $\xi = {\sum_{\sigma\in\supp(\mu')}} \mu'(\sigma)\cdot\nu_\sigma$ for some $\mu'\in\upset\mu$ (note that $\upset\mu = \{\mu\}$ since $\mu(\bot) = 0$) where each $\nu_\sigma\in \de{C_1\nd C_2}(\sigma)$. That means that $\nu_\sigma\vDash \psi_1 \nd\psi_2$ and so by \Cref{lem:convex_sum_idem}, $\dem{C_1\nd C_2}{\upset\mu}\vDash\psi_1\nd\psi_2$.

\item \ruleref{If1} Suppose that $\mu\vDash\varphi$. From $\varphi\Rightarrow \sure{e}$, we know that $\de{e}(\sigma) = \tru$ for each $\sigma\in\supp(\mu)$. Also note that $\mu(\bot) = 0$, so $\upset\mu = \{\mu\}$. Therefore, we get:
\begin{align*}
  \dem{\iftf e{C_1}{C_2}}{\{\mu\}}
  &= \{
    \textstyle\sum_{\sigma\in\supp(\mu)} \mu(\sigma)\cdot\nu_\sigma
    \mid
    \forall\sigma.~ \nu_\sigma\in \de{\iftf e{C_1}{C_2}}_\bot(\sigma)
  \}
  \intertext{Since $\de{e}(\sigma) = \tru$ for all $\sigma\in\supp(\mu)$, then $\de{\iftf e{C_1}{C_2}}_\bot(\sigma) = \de{C_1}_\bot(\sigma)$.}
  &= \{
    \textstyle\sum_{\sigma\in\supp(\mu)} \mu(\sigma)\cdot\nu_\sigma
    \mid
    \forall\sigma.~ \nu_\sigma\in \de{C_1}_\bot(\sigma)
  \}
  \\
  &= \dem{C_1}{\{\mu\}}
\end{align*}
Now take any $\nu \in \dem{\iftf e{C_1}{C_2}}{\{\mu\}}$, we know that $\nu \in \dem{C_1}{\{\mu\}}$, therefore by the induction hypothesis $\nu\vDash\psi$.

\item\ruleref{If2}. Symmetric to the \ruleref{If1} case.


\item \ruleref{Prob Choice}. Suppose $\mu\vDash \varphi_1 \oplus_p \varphi_2$, so $\mu = p\cdot \mu_1 + (1-p)\cdot\mu_2$ such that $\mu_i\vDash\varphi_i$ for $i\in \{1,2\}$. Now, we have:
\begin{align*}
\dem{C}{\{\mu\}}
&= \dem{C}{\upset{p\cdot \mu_1 + (1-p)\cdot\mu_2}}
\\
&= \dem{C}{\upset{\mu_1} \oplus_p \upset{\mu_2}}\tag{By \Cref{lem:upcl-prob}.}
\\
&= \dem{C}{\upset{\mu_1}} \oplus_p \dem{C}{\upset{\mu_2}}\tag{By \Cref{lem:oplus_dist}.}
\end{align*}
So, for any $\nu \in \dem{C}{\upset{\mu}}$, it must be that $\nu = p\cdot\nu_1 + (1-p)\cdot \nu_2$ such that $\nu_i\in \dem{C}{\upset{\mu_i}}$ for $i\in \{1,2\}$. By the induction hypothesis, we know that each $\nu_i\vDash \psi_i$, therefore $\nu\vDash \psi_1 \oplus_p\psi_2$.

\item \ruleref{ND Choice}.
Suppose $\mu\vDash \varphi_1 \nd \varphi_2$, so we know there is some $p\in[0,1]$ such that $\mu\vDash \varphi_1\oplus_p\varphi_2$. Now, take any $\nu\in\dem C{\upset{\mu}}$. We know from the \ruleref{Prob Choice} case that $\nu\vDash \psi_1 \oplus_p\psi_2$. Since $\psi_1 \oplus_p\psi_2 \Rightarrow \psi_1\nd\psi_2$, then $\nu\vDash\psi_1\nd\psi_2$ as well.



\item \ruleref{Consequence}. Suppose $\mu\vDash \varphi'$, then since $\varphi'\Rightarrow \varphi$, $\mu\vDash\varphi$. By the induction hypothesis, $\nu\vDash\psi$ for any $\nu \in \dem{C}{\upset{\mu}}$. Finally, since $\psi\Rightarrow \psi'$, $\nu\vDash\psi'$.


\item\ruleref{Constancy}. Suppose that $\mu\vDash \varphi \land \sure{P}$, so $\mu\vDash\varphi$ and $\supp(\mu) \subseteq\sem P$. Now, take any $\nu\in\dem C{\upset{\mu}}$. By the induction hypothesis, we know that $\nu\vDash\psi$. In addition, since $\mathsf{mod}(C)\cap\mathsf{fv}(P) = \emptyset$, then the program cannot have changed the truth of $P$, and so $\supp(\nu)\subseteq \sem{P}$ as well. This means that $\nu\vDash \psi\land \sure{P}$.

\end{itemize}
\end{proof}

\subsection{Derived Rules}
\label{app:derived}

In this section we give derivations for the additional rules shown in \Cref{sec:rules}.

\begin{lemma}
The following inference is derivable:
\[
\inferrule{
  \varphi_1 \Rightarrow \sure{e}
  \\
  \triple{\varphi_1}{C_1}{\psi_1}
  \\
  \varphi_2 \Rightarrow \lnot \sure{e}
  \\
  \triple{\varphi_2}{C_2}{\psi_2}
}{
  \triple{\varphi_1 \oplus_p \varphi_2}{\iftf e{C_1}{C_2}}{\psi_1 \oplus_p \psi_2}
}
\]
\end{lemma}
\begin{proof}
\[
\inferrule*[right=\ruleref{Prob Split}]{
  \inferrule*[right=\ruleref{If1}]{
    \varphi_1 \Rightarrow \sure{e}
    \\
    \triple{\varphi_1}{C_1}{\psi_1}
  }{
    \triple{\varphi_1}{\iftf e{C_1}{C_2}}{\psi_1}
  }
  \\
  \inferrule*[Right=\ruleref{If2}]{
    \varphi_2 \Rightarrow \sure{\lnot e}
    \\
    \triple{\varphi_2}{C_2}{\psi_2}
  }{
    \triple{\varphi_2}{\iftf e{C_1}{C_2}}{\psi_2}
  }
}{
  \triple{\varphi_1 \oplus_p \varphi_2}{\iftf e{C_1}{C_2}}{\psi_1 \oplus_p \psi_2}
}
\]
\end{proof}

\begin{lemma}\label{lem:if-nd}
The following inference is derivable:
\[
\inferrule{
  \varphi_1 \Rightarrow \sure{e}
  \\
  \triple{\varphi_1}{C_1}{\psi_1}
  \\
  \varphi_2 \Rightarrow \sure{\lnot e}
  \\
  \triple{\varphi_2}{C_2}{\psi_2}
}{
  \triple{\varphi_1 \nd \varphi_2}{\iftf e{C_1}{C_2}}{\psi_1 \nd \psi_2}
}
\]
\end{lemma}
\begin{proof}
\[
\inferrule*[right=\ruleref{ND Split}]{
  \inferrule*[right=\ruleref{If1}]{
    \varphi_1 \Rightarrow \sure{e}
    \\
    \triple{\varphi_1}{C_1}{\psi_1}
  }{
    \triple{\varphi_1}{\iftf e{C_1}{C_2}}{\psi_1}
  }
  \\
  \inferrule*[Right=\ruleref{If2}]{
    \varphi_2 \Rightarrow \sure{\lnot e}
    \\
    \triple{\varphi_2}{C_2}{\psi_2}
  }{
    \triple{\varphi_2}{\iftf e{C_1}{C_2}}{\psi_2}
  }
}{
  \triple{\varphi_1 \nd \varphi_2}{\iftf e{C_1}{C_2}}{\psi_1 \nd \psi_2}
}
\]
\end{proof}

\begin{lemma}\label{lem:if-hoare}
The following inference is derivable:
\[
\inferrule{
  \triple{\sure{P\land e}}{C_1}{\psi}
  \\
  \triple{\sure{P\land\lnot e}}{C_2}{\psi}
}{
  \triple{\sure{P}}{\iftf e{C_1}{C_2}}{\psi}
}
\]
\end{lemma}
\begin{proof}
First, note that $\sure{P} \Rightarrow \sure{P\land e} \nd \sure{P\land \lnot e}$, since for any $\mu\vDash \sure{P}$, $e$ must be either true or false in every $\sigma\in \supp(\mu)$. We can therefore complete the derivation as follows:
\[
\inferrule*[right=\ruleref{Consequence}]{
  \inferrule*[Right=\Cref{lem:if-nd}]{
    \sure{P\land e} \Rightarrow \sure{e}
    \quad
    \triple{\sure{P\land e}}{C_1}{\psi}
    \quad
    \sure{P\land\lnot e} \Rightarrow \sure{\lnot e}
    \quad
    \triple{\sure{P\land\lnot e}}{C_2}{\psi}
  }{
    \triple{\sure{P\land e} \nd \sure{P\land \lnot e}}{\iftf e{C_1}{C_2}}{\psi\nd \psi}
  }
}{
  \triple{\sure{P}}{\iftf e{C_1}{C_2}}{\psi}
}
\]
\end{proof}

\begin{lemma}\label{lem:flip-rule}
The following inference is derivable:
\[
\inferrule{
    \varphi\Rightarrow \sure{e=p}
    \\
    x \notin \mathsf{fv}(\varphi)
  }{
    \triple{\varphi}{x \coloneqq \mathsf{flip}(e)}{\varphi \land (\sure{x = \tru} \oplus_p \sure{x=\fls})}
  }
\]
\end{lemma}
\begin{proof}
First, note that since $x \notin \mathsf{fv}(\varphi)$, then $\varphi[v/x] = \varphi$ for any $v \in\mathsf{Val}$. Therefore, we have:
\[
  \left(\varphi \land \sure{x = \tru}\right)[\tru/x]
  \quad=\quad
  \varphi[\tru/x] \land \sure{\tru = \tru}
  \quad=\quad
  \varphi
\]
And similarly for the $x=\fls$ case. We now complete the derivation as follows:
\[
\inferrule*[right=\ruleref{Consequence}]{
  \inferrule*[Right=\ruleref{Prob}]{
    \inferrule*[right=\ruleref{Assign}]{\;}{
      \triple{\varphi}{x \coloneqq \tru}{\varphi \land \sure{x=\tru}}
    }
    \quad
    \inferrule*[Right=\ruleref{Assign}]{\;}{
      \triple{\varphi}{x \coloneqq \fls}{\varphi \land \sure{x=\fls}}
    }
  }{
    \triple{\varphi}{(x \coloneqq \tru) \oplus_e (x\coloneqq\fls)}{(\varphi \land \sure{x = \tru}) \oplus_p (\varphi \land \sure{x=\fls})}
  }
}{
  \triple{\varphi}{x \coloneqq \mathsf{flip}(e)}{\varphi \land (\sure{x = \tru} \oplus_p \sure{x=\fls})}
}
\]
\end{proof}

\begin{lemma}\label{lem:nd-rule}
The following inference is derivable:
\[
  \inferrule{\;}{\triple{\sure{\tru}}{x\gets S}{\textstyle\bignd_{v\in S} \sure{x=v}}}
\]
\end{lemma}
\begin{proof}
  Recall that $S$ must be a nonempty finite set, so let $S = \{ v_1, \ldots, v_n \}$, and therefore the postcondition is equivalent to $\sure{x=v_1} \nd \cdots \nd \sure{x=v_n}$. We now proceed by induction on the size of the set $S$. In the base case, $|S| = 1$ and so $S = \{ v\}$ is a singleton set. We therefore the complete the derivation with a single application of the \ruleref{Assign} rule since $\sure{x =v}[v/x] = \sure{v=v} = \sure{\tru}$.
\[
\inferrule*[right=\ruleref{Assign}]{\;}{
  \triple{\sure{\tru}}{x \coloneqq v}{\sure{x = v}}
}
\]
Now, for the induction step, we suppose the claim holds for sets of size $n$, and that $S = \{ v_1, \ldots, n_{n+1} \}$. Let $S' = \{ v_1,\ldots, v_n\}$ and note that:
\[
  x\gets S
  \qquad=\qquad
  (x \gets S') \nd (x \coloneqq v_{n+1})
\]
We can now complete the derivation as follows:
\[
\inferrule*[right=\ruleref{Nondet}]{
  \inferrule*[right=\textnormal{Induction Hypothesis}]{\;}{
    \triple{\sure{\tru}}{x \gets S'}{\textstyle\bignd_{i=1}^n \sure{x=v_n}}
  }
  \quad
  \inferrule*[Right=\ruleref{Assign}]{\;}{
    \triple{\sure{\tru}}{x\coloneqq v_{n+1}}{\sure{x = v_{n+1}}}
  }
}{
  \triple{\sure{\tru}}{(x \gets S') \nd (x \coloneqq v_{n+1})}{ (\textstyle\bignd_{i=1}^n \sure{x=v_n}) \nd \sure{x = v_{n+1}}}
}
\] 
\end{proof}

\section{Analyzing Loops}


\subsection{The Zero-One Law}
\label{app:zero-one}

\begin{definition}[Minimum Termination Probability]
Let the minimum termination probability of some set $S \in \C(X)$ be defined as:
\[
  \minterm(S) \triangleq \inf_{\mu \in S} (1 - \mu(\bot))
\]
\end{definition}

The following lemma (similar to Lemma 7.3.1 from \citet{mciver2005abstraction}) states that if some invariant holds for each iteration of a loop, then it will hold at the end of the program execution with probability at least $t$, where $t$ is the minimum probability of termination for the loop.
\begin{lemma}[Invariant Reasoning]
\label{lem:invariant}
Suppose that $\varphi\Rightarrow \sure{e}$ and $\psi\Rightarrow \sure{\lnot e}$ and $\vDash\triple\varphi{C}{\varphi\nd\psi}$. Now take any $\mu\vDash\varphi$ and let $t = \minterm(\dem{\whl eC}{\upset\mu})$ be the minimum probability of termination. Then $\nu \vDash \psi \oplus_t \top$ for any $\nu\in\dem{\whl eC}{\upset\mu}$.
\end{lemma}
\begin{proof}
Let:
\[
F_n(S) \triangleq \Phi_{\langle C,e\rangle}^\alpha\left(\bot_\C^\bullet\right)^\dagger(S)
\]
We begin by proving a more general claim: for any $n\in
\mathbb N$ and any set $S\in\C(\Sigma)$ such that $\nu\vDash\varphi$ for all $\nu\in S$, $\xi \vDash \psi \oplus_{\minterm(F_n(S))}\top$ for any $\xi \in F_n(S)$. The proof is by induction on $n$.
\begin{itemize}
\item\textbf{Base Case}. If $n = 0$, then the claim holds trivially, since $F_0(S) = \mathord\uparrow\{ \delta_\bot\}$ and therefore $\minterm(F_0(S)) = 0$ and so $\psi\oplus_{\minterm(F_0(S))} \top \Leftrightarrow\top$ and anything satisfies $\top$.

\item\textbf{Successor Case}. Suppose the claim holds for $n$. Now, we have:
\begin{align*}
F_{n+1}(S)
&= \Phi_{\langle C,e\rangle}^{n+1}\left(\bot_\C^\bullet\right)^\dagger(S)
\\
\intertext{Note that since $\nu\vDash\varphi$ for all $\nu \in S$, and therefore $\de{e}(\sigma) = \tru$ for all $\sigma\in\supp(\nu)$, we can simplify as follows:}
&= \Phi_{\langle C,e\rangle}^n\left(\bot_\C^\bullet\right)^\dagger(\dem{C}S)
\\
\intertext{By \Cref{lem:dagger-upset}.}
&= \smashoperator{\bigcup_{\nu \in \dem{C}S}} \Phi_{\langle C,e\rangle}^n\left(\bot_\C^\bullet\right)^\dagger(\upset\nu)
\\
\intertext{Now, from $\vDash\triple\varphi{C}{\varphi\nd\psi}$, we know that for every $\nu \in \de{C}^\dagger_\bot(S)$ there is a probability $p$ such that $\nu = p\cdot\nu_1 + (1-p)\cdot\nu_2$ and $\nu_1\vDash\varphi$ and $\nu_2\vDash\psi$. Using \Cref{lem:upcl-prob}, we get:}
&= \smashoperator{\bigcup_{
  \nu_1\vDash\varphi,\nu_2\vDash\psi,p \mid p\cdot\nu_1 + (1-p)\cdot\nu_2 \in \dem{C}S
}} \Phi_{\langle C,e\rangle}^n\left(\bot_\C^\bullet\right)^\dagger_\bot(\upset{\nu_1} \oplus_p \upset{\nu_2})
\\
\intertext{By \Cref{lem:oplus_dist}.}
&= \smashoperator{\bigcup_{
  \nu_1\vDash\varphi,\nu_2\vDash\psi,p \mid p\cdot\nu_1 + (1-p)\cdot\nu_2 \in \dem{C}S
}} \Phi_{\langle C,e\rangle}^n\left(\bot_\C^\bullet\right)^\dagger(\upset{\nu_1}) \oplus_p \Phi_{\langle C,e\rangle}^n\left(\bot_\C^\bullet\right)^\dagger_\bot(\upset{\nu_2})
\\
\intertext{Since $\nu_2\vDash\psi$ and $\psi\Rightarrow\sure{\lnot e}$, we know that applying $\Phi_{\langle C,e\rangle}^n\left(\bot_\C^\bullet\right)$ to $\nu_2$ will have no effect.}
&= \smashoperator{\bigcup_{
  \nu_1\vDash\varphi,\nu_2\vDash\psi,p \mid p\cdot\nu_1 + (1-p)\cdot\nu_2 \in \dem{C}S
  }} F_{n}(\upset{\nu_1}) \oplus_p \upset{\nu_2}
\end{align*}
So, for any $\xi \in F_{n+1}$, we know that $\xi = p\cdot \xi' + (1-p)\cdot\nu_2$ such that $p\cdot\nu_1 + (1-p)\cdot\nu_2 \in\dem{C}S$ and $\nu_1 \vDash\varphi$ and $\nu_2\vDash\psi$, and $\xi'\in F_n(\nu_1)$. By the induction hypothesis, we know that $\xi' \vDash \psi \oplus_{\minterm(F_n(\{\nu_1\}))} \top$. Therefore, $\xi \vDash (\psi \oplus_{\minterm(F_n(\{\nu_1\}))} \top) \oplus_p \psi$ and  so $\xi \vDash \psi \oplus_{p\cdot \minterm(F_n(\{\nu_1\})) + (1-p)} \top$.

It remains only to show that:
\begin{align*}
\minterm(F_{n+1}(S))
&= \inf_{\xi \in F_{n+1}(S)} (1 - \xi(\bot)) \\
&= \inf_{\xi \in \bigcup_{\nu \in \dem{C}S} F_{n}(\{\nu_1\}) \oplus_p \{\nu_2\}} (1 - \xi(\bot))
\\
&= \inf_{\nu \in \dem{C}S}  \inf_{\xi\in F_{\alpha}(\{\nu_1\}) \oplus_p \{\nu_2\}} (1 - \xi(\bot))
\\
\intertext{Since we know that $\bot\notin\supp(\nu_2)$}
&= \inf_{\nu \in \dem{C}S} \inf_{\xi \in F_{n}(\{\nu_1\})} (p\cdot (1 - \xi(\bot)) + 1-p)
\\
&= \inf_{\nu \in \dem{C}S} p\cdot (\inf_{\xi \in F_{n}(\{\nu_1\})} (1 - \xi(\bot))) + 1-p
\\
&= \inf_{\nu \in \dem{C}S} p\cdot \minterm(F_{n}(\{\nu_1\})) + 1-p
\\
&\le p\cdot \minterm(F_{n}(\{\nu_1\})) + 1-p
\end{align*}
\end{itemize}

Now, since $\dem{\whl eC}{\upset\mu} = \sup_{n\in\mathbb N} F_n(\upset\mu) = \bigcap_{n\in\mathbb N} F_n(\upset\mu)$, we have:
\begin{align*}
\minterm(\dem{\whl eC}{\upset\mu})
&= \minterm(\bigcap_{n\in\mathbb N} F_n(\upset\mu))
\\
&= \inf_{\xi \in \bigcap_{n\in\mathbb N} F_n(\upset\mu)}(1-\xi(\bot))
\\
\intertext{Now since the $F_n(\upset\mu)$ terms form a chain, the probability of $\bot$ decreases monotonically as $n$ increases:}
&= \sup_{n \in \mathbb N} \inf_{\xi\in F_n(\upset\mu)} (1-\xi(\bot))
\\
&= \sup_{n\in\mathbb N} \minterm(F_n(\upset\mu))
\end{align*}
Now, take any $\xi \in \dem{\whl eC}{\upset\mu}$. We know that $\xi \in F_n(\upset\mu)$ for all $n\in\mathbb N$ and therefore $\xi \vDash \psi \oplus_{\minterm(F_n(\upset\mu))} \top$ for all $n\in\mathbb N$. Since we also know that $\minterm(\dem{\whl eC}{\upset\mu}) = \sup_{n\in\mathbb N} \minterm(F_n(\upset\mu))$, we get that $\xi\vDash \psi\oplus_{\minterm(\dem{\whl eC}{\upset\mu})}\top$.

\end{proof}

Building on the previous lemma, we will show that in fact the probability of termination can only be zero or one. More precisely, if there is any positive probability of termination, then the program must almost surely terminate.

\begin{lemma}[Almost Sure Termination]\label{lem:zo-ast}
Suppose that $\varphi\Rightarrow \sure{e}$ and $\psi\Rightarrow \sure{\lnot e}$ and $\vDash\triple\varphi{C}{\varphi\nd\psi}$ and there exists some probability $p>0$ such that $\minterm(\dem{\whl eC}{\upset\mu}) \ge p$ for any $\mu\vDash\varphi$. 
Then, for any $\mu\vDash\varphi$, $\minterm(\dem{\whl eC}{\upset\mu}) = 1$.
\end{lemma}
\begin{proof}
Given that $\vDash\triple\varphi{C}{\varphi\nd\psi}$, then after each execution of the loop there is some probability of ending up in a state satisfying $\varphi$, so that the loop will continue. Call each probability in this sequence $q_\alpha$, so the total probability of nontermination is:
\[
  \PP[\text{nonterm}]
  \qquad=\qquad
  q_1 \times q_2 \times q_3 \times \cdots \times q_\alpha \times q_{\alpha+1}\times \cdots
\]
Each time the execution returns to a distribution $\mu\vDash\varphi$, we know that:
\[
  \minterm(\dem{\whl eC}{\upset\mu}) \ge p
\]
So, any tail of the product above must be at most $1-p$.
\[
  \PP[\text{nonterm}]
  \qquad=\qquad
  q_1 \times q_2 \times q_3 \times \cdots \times \underbrace{ q_\alpha \times q_{\alpha+1}\times \cdots }_{\le 1-p}
\]
Let $\rho_k = \prod^{k}_{i=1} q_i$, the $\nth k$ partial product.
Since the $\rho_k$ are monotonically decreasing and bounded below by $0$, we have that $\lim_{k \to \infty} \rho_{k}$ exists, ensuring that the infinite product above exists.

Set $\rho = \lim_{k \to \infty} \rho_k = \PP[\text{nonterm}]$.
We shall show that $\rho = 0$, by contradiction.
Suppose that $\rho > 0$.
Then, taking logarithms, we can convert the infinite product into a (converging) infinite series.
Setting $s_k = \log(\rho_k)$ and $s = \log(\rho)$, then we have
\[s_k = \log\left(\prod^{k}_{i = 1} q_i\right) = \sum_{i=1}^{k} \log(q_i)\]
Since $\log$ is continuous, $\lim_{k \rightarrow \infty} s_k = s$.
Similarly from the fact that the tail products of the $q_i$ are always bounded by $1 - p$, and the fact that $\log$ is monotone, we have that for all $j$,
\begin{equation}
\label{eqn:tailprod}
\sum^\infty_{i=j} \log(q_i) \leq \log(1 -p)
\end{equation}
Let $\epsilon = |\log(1 - p)| = -\log(1 - p)$
By convergence of $s_k$, there exists some $N$ such that for all $m, n > N$, $|s_m - s_n| < \epsilon/2$.
Taking limits as $m \to \infty$, we have then that for $n > N$, $|s - s_n| \leq \epsilon/2 < \epsilon$. Hence,
\begin{align*}
-\log(1-p) = \epsilon
&> \left| s - s_n \right| \\
&= \left| \sum_{i = 1}^\infty \log(q_i) - \sum_{i=1}^{n} \log(q_i) \right| \\
&= \left| \sum_{i = n + 1}^\infty \log(q_i) \right| \\
&= - \sum_{i = n + 1}^\infty \log(q_i)
\end{align*}
Thus the above implies that $\log(1 - p) < \sum_{i=n}^{\infty} \log(q_i)$, contradicting \cref{eqn:tailprod}.
\end{proof}

\begin{lemma}[Zero One Law]\label{lem:zero-one}
The following inference is valid:
\[
\inferrule{
  \varphi \Rightarrow \sure{e}
  \quad
  \psi\Rightarrow \sure{\lnot e}
  \quad
  \triple{\varphi}C{\varphi \nd \psi}
  \quad
  \triple{\varphi}{\whl eC}{\sure{\lnot e} \oplus_p \top}
  \quad
  p >0
}{
  \triple{\varphi}{\whl eC}{\psi}
}{\ruleref{Zero-One}}
\]
\end{lemma}
\begin{proof}
Suppose that $\mu\vDash\varphi$ and take any $\nu \in \dem{\whl eC}{\upset\mu}$. By \Cref{lem:invariant}, we know that $\nu\vDash \psi\oplus_t \top$ where $t = \minterm(\dem{\whl eC}{\upset\mu})$.
Using \Cref{lem:zo-ast} and the premise that $\triple{\varphi}{\whl eC}{\sure{\lnot e} \oplus_p \top}$, we know that $t = 1$, therefore $\nu \vDash \psi$.
\end{proof}

\subsection{Variants and Ranking Functions}
\label{app:variants}

\begin{theorem}[The Bounded Integer Variant Rule]
Let $p > 0$ be some nonzero probability and $(\varphi_n)_{n=0}^N$ be a finite sequence of assertions such that $\varphi_0 \Rightarrow \sure{\lnot e}$ and $\varphi_n \Rightarrow \sure{e}$ for all $n \ge 1$.
\[
\inferrule{
  \forall n \in [1,N].\quad
  \triple{\varphi_n}C{\textstyle(\bignd_{k=0}^{n-1} \varphi_k) \oplus_p (\bignd_{k=0}^N \varphi_k)}
}{
  \triple{\textstyle\bignd_{k=0}^N \varphi_k}{\whl eC}{\varphi_0}
}{\textnormal{\ruleref{Bounded Variant}}}
\]
\end{theorem}
\begin{proof}
This rule is a straightforward application of the \ruleref{Zero-One} law.
Let $\varphi = \varphi_1 \nd \cdots \nd \varphi_N$ and $\psi = \varphi_0$, then it is relatively easy to see that $\varphi\Rightarrow \sure{e}$ and $\psi \Rightarrow \sure{\lnot e}$ and that:
\[
  \textstyle(\bignd_{k=0}^{n-1} \varphi_k) \oplus_p (\bignd_{k=0}^N \varphi_k)
  \quad\Rightarrow\quad
  \textstyle(\bignd_{k=0}^{n-1} \varphi_k) \nd (\bignd_{k=0}^N \varphi_k)
  \quad\Rightarrow\quad
  \bignd_{k=0}^N \varphi_k
  \quad\Rightarrow\quad
  \varphi \nd \psi
\]
So, we have:
\[
\inferrule*[right={\ruleref{ND Split}~$\times~ N$}]{
  \forall n \in [1,N].
  \quad
  \inferrule*[Right=\ruleref{Consequence}]{
      \triple{\varphi_n}C{\textstyle(\bignd_{k=0}^{n-1} \varphi_k) \oplus_p (\bignd_{k=0}^N \varphi_k)}
  }{
    \triple{\varphi_n}C{\varphi \nd\psi}
  }
}{
  \triple{\varphi}C{\varphi \nd \psi}
}
\]
So, $\varphi$ and $\psi$ are an invariant pair for $C$, giving us the first premise for the \ruleref{Zero-One} law. Now, it just remains to show that there is a nonzero probability of termination. This is simple, since with probability $p$, the index of the variant must decrease by at least 1 on each iteration, so after $N$ iterations the program must terminate with probability at least $p^N$. Since $p$ is nonzero and $N$ is finite, then $p^N$ is also finite, so we have.
\[
  \vDash\triple{\varphi}{\whl eC}{\sure{\lnot e} \oplus_{p^N} \top}
\]%
\qedhere
\end{proof}

We now show how to derive the bounded integer rule of \citet[Lemma 7.5.1]{mciver2005abstraction} and \citet[Theorem 6.7]{kaminski}.

\begin{lemma}
Let $P$ be some basic assertion (the invariant), $e_\mathsf{rank}$ be an integer-valued ranking expression, and $p>0$ be some nonzero probability. Also, suppose that $\sure{P\land e} \Rightarrow \sure{\ell \le e_{\mathsf{rank}} \le h}$. Then, the following inference rule is derivable:
\[
\inferrule{
  \forall z \in [ \ell, h ]. \quad \triple{\sure{P \land e \land e_\mathsf{rank} = z}}C{\sure{P \land e_\mathsf{rank} < z} \oplus_p \sure{P}}
}{
  \triple{\sure{P}}{\whl eC}{\sure{P\land\lnot e}}
}{\textnormal{\ruleref{Bounded Rank}}}
\]
\end{lemma}
\begin{proof}
Let $N = h - \ell + 1$, and let $(\varphi_n)_{n=0}^N$ be defined as follows:
\[
  \varphi_n = \left\{
    \begin{array}{ll}
      \sure{P\land \lnot e} & \text{if}~ n=0 \\
      \sure{P \land e \land e_\mathsf{rank} = \ell + n - 1} & \text{if}~ n \in [1,N]
    \end{array}
  \right.
\]
First, we argue that $\sure{P} \Rightarrow\varphi_0 \nd \cdots \nd \varphi_N$. Suppose that $\mu\vDash \sure{P}$, so $\sigma\in\sem{P}$ for each $\sigma\in\supp(\mu)$. It must be that $\de{e}(\sigma)$ is either $\tru$ or $\fls$. If it is $\tru$, then we also know that $\ell\le\de{e_\mathsf{rank}}(\sigma) \le h$ since $\sure{P\land e} \Rightarrow \sure{\ell \le e_{\mathsf{rank}} \le h}$, and therefore $\eta(\sigma) \vDash \varphi_{n} = \sure{P \land e \land e_{\mathsf{rank}} = \ell +n-1}$ where $n=\de{e_\mathsf{rank}}(\sigma) - \ell + 1$. Now suppose that $\de{e}(\sigma) = \fls$, then clearly $\eta(\sigma) \vDash \varphi_0 = \sure{P\land\lnot e}$. Since this is true for all $\sigma\in\supp(\mu)$, we have:
\[
  \sure{P}
  \quad\Rightarrow\quad
  \sure{P\land \lnot e} \nd \sure{P \land e \land e_\mathsf{rank} = \ell} \nd \cdots \nd \sure{P \land e \land e_\mathsf{rank} = h}
  \quad = \quad
  \varphi_0 \nd \cdots \nd \varphi_N
\]
Next, we show that $\sure{P \land e_\mathsf{rank} < \ell+n-1} \Rightarrow \varphi_0 \nd \cdots \nd\varphi_{n-1}$. Suppose that $\mu\vDash \sure{P \land e_\mathsf{rank} < \ell+n-1}$, meaning that $\sigma\in\sem{P \land e_\mathsf{rank} < \ell+n-1}$ for each $\sigma\in \supp(\mu)$. If $\de{e}(\sigma) = \tru$, then we know that $\ell \le \de{e_\mathsf{rank}}(\sigma)\le h$ since $\sure{P\land e} \Rightarrow \sure{\ell\le e_\mathsf{rank}\le h}$. We also know that $\de{e_\mathsf{rank}}(\sigma) < \ell + n - 1$, so there must be some $m < n$ such that $\de{e_\mathsf{rank}}(\sigma) = \ell + m -1$, so $\eta(\sigma)\vDash \varphi_m$. If instead $\de{e}(\sigma) = \fls$, then $\eta(\sigma)\vDash \varphi_0$. Therefore, $\mu \vDash \varphi_0 \nd \cdots \nd \varphi_{n-1}$.

We now complete the derivation as follows:
\[
\inferrule*[right=\ruleref{Bounded Variant}]{
  \forall n \in [1,N].
  \quad
  \inferrule*[Right=\ruleref{Consequence}]{
    \inferrule*{
      \forall z\in[\ell,h].\quad
      \triple{\sure{P \land e \land e_\mathsf{rank} = z}}C{\sure{P\land e_\mathsf{rank} < z} \oplus_p \sure{P}}
    }{
      \triple{\sure{P \land e \land e_\mathsf{rank} = \ell + n -1}}C{\sure{P\land e_\mathsf{rank} < \ell + n -1} \oplus_p \sure{P}}
    }
  }{
    \triple{\sure{P \land e \land e_\mathsf{rank} = \ell + n -1}}C{\textstyle(\bignd_{k=0}^{n-1} \varphi_k) \oplus_p (\bignd_{k=0}^{N} \varphi_k)}
  }
}{
  \triple{\sure{P}}{\whl eC}{\sure{P\land\lnot e}}
}
\]
\end{proof}

\subsection{The Progressing Rank Rule}
\label{app:new-rule}

In this section, we present our propositional variant of a more complex and powerful almost sure termination rule that was originally developed by \citet{mciver2018new} and later refined by \citet[\S 6.2.3]{kaminski}. Similar to the \ruleref{Bounded Rank} rule, the new rule relies on an invariant $P$ and a ranking function $e_\mathsf{rank}$. However, this time $e_\mathsf{rank}$ is unbounded and is nonnegative-rational-valued rather than integer-valued. In addition, two more antitone functions are needed as well:
\begin{align*}
  \de{e_\mathsf{rank}} &\colon \Sigma \to \mathbb{Q}_{\ge0}
  && \textsf{(Ranking Function)}
  \\
  p & \colon  \mathbb{Q}_{\ge0} \to (0,1]
  && \textsf{(Probability)}
  \\
  d & \colon  \mathbb{Q}_{\ge0} \to  \mathbb{Q}_{\ge0}
  && \textsf{(Decrease)}
\end{align*}
Antitonicity means that $p(r) \ge p(s)$ if $r \le s$, and the same for $d$. Additionally, the loop guard $e$ must be false iff $e_\mathsf{rank} = 0$.

The intuition behind $p$ and $d$ is that as the rank gets closer to $0$ (implying termination), the probability of decreasing the rank and the magnitude of such a decrease both increase. The premise of the rule is that if the rank is initially $k$, then after one iteration of the loop, rank must decrease by at lease $d(k)$ with probability at least $p(k)$. In addition, the expected rank must decrease monotonically as well. The rule is shown below:
\[
\inferrule{
  \triple{\sure{P\land e \land e_{\mathsf{rank}}=k}}C{
     \sure{P\land e_{\mathsf{rank}} \le k-d(k)} \oplus_{p(k)} \sure{P\land e_{\mathsf{rank}} \le k + {\textstyle\frac{p(k)}{1-p(k)}}d(k)}
   }
}{
  \triple{\sure{P}}{\whl eC}{\sure{P\land\lnot e}}
}{\ruleref{Progressing Rank}}
\]
Note that the expected decrease in rank is guaranteed by the fact that there are two probabilistic outcomes, and in the weighted average we have $e_\mathsf{rank} \le k$ (more on that soon, when we discuss soundness).

The weakest pre-expectation version of this rule presented in \citet[\S 6.2.3]{kaminski} has 3 different premises, which involve different concepts such as \emph{sub-invariants} and \emph{super-invariants}, as well as both the \emph{demonic} and \emph{angelic} $\wpre$ calculi. We contend that the Outcome Logic style version presented here is simpler because all three of those premises have been coalesced, using just a single concept: \emph{demonic outcome triples}.

\subsubsection*{Soundness of the Rule}
We first establish that $\varphi \triangleq \sure{P \land e}$ and $\psi \triangleq \sure{P\land \lnot e}$ is an invariant pair for the program. That is, we must prove that $\vDash\triple{\sure{P\land e}}C{\sure{P\land e}\nd \sure{P\land \lnot e}}$ is a valid triple.

Suppose that $\mu\vDash \sure{P \land e}$. Then we can partition $\mu$ into a countable number of components based on the value of $e_\mathsf{rank}$. That is, there must be some distribution $\xi \in \D(\mathbb{Q}_{\ge 0})$ and distributions $(\mu_k)_{k\in \supp(\xi)}$ such that $\mu = {\sum_{k \in \supp(\xi)}} \xi(k) \cdot \mu_k$ and $\de{e_\mathsf{rank}}(\sigma) = k$ for each $k\in\supp(\xi)$ and $\sigma\in\supp(\mu_k)$. Clearly, we have $\mu_k \vDash \sure{P \land e \land e_\mathsf{rank} = k}$ for each $k \in \supp(\xi)$, so by the premise of \ruleref{Progressing Rank}, we know that:
\begin{align*}
  \dem{C}{\upset{\mu_k}}
  \vDash& \sure{P\land e_{\mathsf{rank}} \le k-d(k)} \oplus_{p(k)} \sure{P\land e_{\mathsf{rank}} \le k + {\textstyle\frac{p(k)}{1-p(k)}}d(k)}
  \\
  \Rightarrow & \; \sure{P}
  \\
  \Rightarrow & \; \sure{P \land e} \nd \sure{P\land \lnot e}
\end{align*}
Therefore, by \Cref{lem:convex_sum_idem}, $\mu \vDash \sure{P\land e} \nd \sure{P\land \lnot e}$ and by \Cref{lem:invariant}, $\sure{P\land \lnot e}$ applies to all the terminating outcomes. All that remains is to show that the program almost surely terminates. For this, we will appeal to \citet[Theorem 6.8]{kaminski}, which relies on the weakest pre-expectation $\mathsf{wp}\de{C}(e)$ as defined in \citet[\S 4.1]{kaminski} and is equal to the minimum expected value of the expression $e$ after running the program $C$:
\begin{equation}\label{eq:wp}
  \mathsf{wp}\de{C}(e)(\sigma) = \inf_{\mu \in \de{C}(\sigma)} \mathbb E_\mu[e]
  \qquad\text{where}\qquad
  \mathbb{E}_\mu[e] \triangleq \smashoperator{\sum_{\tau\in\supp(\mu)}} \mu(\tau) \cdot \de{e}(\tau)
\end{equation}
In addition, let the Iverson brackets $[e]$ evaluates to 0 if $e$ is false and 1 if $e$ is true. To invoke Theorem 6.8 of \citet{kaminski}, we must show that the four following premises hold:
\begin{enumerate}
\item[\textsc{a.}] $P$ is a \textsf{wp}-subinvariant \cite[Definiton 5.1b]{kaminski} of $\whl eC$:
\[
  [P] \le [\lnot e] \cdot [P] + [e]\cdot \mathsf{wp}\de{C}([P])
\]
If $e$ is false, then this reduces to $[P] \le [P]$, which is trivially true. If not, it reduces to $[P]\le \mathsf{wp}\de{C}([P])$, which is implied by the fact the $P$ is an invariant of the loop, as we just showed.

\item[\textsc{b.}] $e_\mathsf{rank} = 0$ implies termination. This was an assumption.

\item[\textsc{c.}] $e_\mathsf{rank}$ is an \textsf{awp}-superinvariant \cite[Definiton 5.1c]{kaminski} of $\whl eC$:
\[
  e_\mathsf{rank} \ge [\lnot e]\cdot e_\mathsf{rank} + [e]\cdot \mathsf{awp}\de{C}(e_\mathsf{rank})
\]
If $e$ is false, then this reduces to $e_\mathsf{rank} \ge e_\mathsf{rank}$, which is trivially true. If not, then it reduces to $e_\mathsf{rank} \ge\mathsf{awp}\de{C}(e_\mathsf{rank})$. Here $\mathsf{awp}$ is the \emph{angelic} weakest pre, which maximizes expected values instead of minimizing them. Formally, the definition is the same as (\ref{eq:wp}), but with the $\inf$ replaced by a $\sup$. It is used since we are now bounding the value from above. This condition simply states that the expected value of $e_\mathsf{rank}$ must not increase after executing the loop body.

To show that this is true, take any state $\sigma \in \sem{P\land e}$ and let $k = \dem{e_\mathsf{rank}}\sigma$\footnote{
  Note that condition \textsc{c} does not require $P$ to hold, however we can modify $e_\mathsf{rank}$ such that it is equal to the constant 0 whenever $P$ is false without affecting the truth of any of our other claims.
}.
We know from the premise of \ruleref{Progressing Rank} that $\mu\vDash \sure{P\land e_{\mathsf{rank}} \le k-d(k)} \oplus_{p(k)} \sure{P\land e_{\mathsf{rank}} \le k + {\textstyle\frac{p(k)}{1-p(k)}}d(k)}$ for any $\mu \in \de{C}(\sigma)$. That means that there are $\mu_1$ and $\mu_2$ such that $\mu = p(k) \cdot \mu_1 + (1-p(k))\cdot\mu_2$ and:
\[
  \supp(\mu_1)\subseteq\sem{P\land e_{\mathsf{rank}} \le k-d(k)}
  \quad\text{and}\quad
  \supp(\mu_2)\subseteq \sem{P\land e_{\mathsf{rank}} \le k + {\textstyle\frac{p(k)}{1-p(k)}}d(k)}
\]
From this, we can determine the expected value of $e_\mathsf{rank}$ according to $\mu$:
\begin{align*}
  \mathbb E_\mu[e_\mathsf{rank}]
  &= \smashoperator{\sum_{\tau\in\supp(\mu)}} \mu(\tau)\cdot \de{e_\mathsf{rank}}(\tau)
  \\
  &= p(k)\cdot\smashoperator{\sum_{\tau\in\supp(\mu_1)}} \mu_1(\tau)\cdot \de{e_\mathsf{rank}}(\tau)
    + (1-p(k))\cdot\smashoperator{\sum_{\tau\in\supp(\mu_2)}} \mu_2(\tau)\cdot \de{e_\mathsf{rank}}(\tau)
  \\
  &\le p(k)\cdot\smashoperator{\sum_{\tau\in\supp(\mu_1)}} \mu_1(\tau)\cdot (k-d(k))
    + (1-p(k))\cdot\smashoperator{\sum_{\tau\in\supp(\mu_2)}} \mu_2(\tau)\cdot \left(k + {\frac{p(k)}{1-p(k)}}\cdot d(k)\right)
    \intertext{Since $|\mu_1| = |\mu_2| = 1$:}
  &= p(k)\cdot (k-d(k))
    + (1-p(k))\cdot \left(k + {\frac{p(k)}{1-p(k)}}\cdot d(k)\right)
  \\
  &= (p(k) + 1-p(k)) \cdot k - p(k)\cdot d(k) + \cancel{(1-p(k))} \cdot \frac{p(k)}{\cancel{1-p(k)}}\cdot d(k)
  \\
  &= k
\end{align*}
So, we have shown that $\mathbb{E}_\mu[e_\mathsf{rank}] \le k$ for all $\mu\in \de{C}(\sigma)$, therefore the premise that the expected value must always decrease holds.

\item[\textsc{d.}] $e_\mathsf{rank}$ satisfies a progress condition; for any $k \in \mathbb{Q}_{\ge0}$:
\[
  p(k) \cdot [e] \cdot [P] \cdot [e_\mathsf{rank} = k] \le \mathsf{wp}\de{C}([e_\mathsf{rank} \le k - d(k)])
\]
This condition simply states that if $\sigma\in \sem{P\land e\land e_\mathsf{rank} =k}$, then $e_\mathsf{rank} \le k - d(k)$ holds after running $C$ with probability at least $p(k)$, which is implied directly by the premise of \ruleref{Progressing Rank}.
\end{enumerate}

\subsubsection*{Demonically Fair Random Walk}
We can prove almost sure termination of a more sophisticated random walk using the \ruleref{Progressing Rank} rule. This example is adapted from \citet[\S 5.2]{mciver2018new} and \citet[\S 6.2.4.1]{kaminski}.
The program below represents a random walk in which the agent steps towards the origin with probability $\frac12$, otherwise the adversary can choose whether or not to step away from the origin.
\[
\begin{array}{l}
\whl{x>0}{} \\
\quad x\coloneqq x-1 \oplus_{\frac12} (x \coloneqq x+1 \nd \skp)
\end{array}
\]
Now, we can instantiate the \ruleref{Progressing Rank} rule by letting:
\[
  P\triangleq x \ge 0
  \qquad\qquad
  e_\mathsf{rank} \triangleq x
  \qquad\qquad
  p(k) \triangleq \frac12
  \qquad\qquad
 d(k) \triangleq 1
\]
In addition, we will presume that $x$ is always integer-valued.
We now prove the premise of \ruleref{Progressing Rank} as follows:
\[\def\arraystretch{1.25}
\begin{array}{l}
\ob{\sure{x > 0 \land x = k}} \\
\left(
\begin{array}{l}
\ob{\sure{x > 0 \land x =k}} \\
\;\;x \coloneqq x-1 \\
\ob{\sure{x \ge 0 \land x = k-1}} \\
\end{array}
\right)
\oplus_\frac12
\left(
  \begin{array}{l}
    \ob{\sure{x > 0 \land x =k}} \\
    \;\;x \coloneqq x+1 \\
    \ob{\sure{x \ge 0 \land x = k+1}}
  \end{array}
  \nd
  \begin{array}{l}
      \ob{\sure{x > 0 \land x =k}} \\
      \;\; \skp \\
      \ob{\sure{x \ge 0 \land x = k}}
  \end{array}
\right)
\\
\ob{ \sure{x \ge 0 \land x = k-1} \oplus_\frac12 ( \sure{x \ge 0 \land x = k+1} \nd \sure{x\ge 0 \land x=k})}
\\
\ob{ \sure{x \ge 0 \land x \le k-1} \oplus_\frac12 \sure{x \ge 0 \land x \le k+1}}
\end{array}
\]
The consequence in the last step is justified since we can weaken the second nondeterministic possibility to be $x \le k+1$ and then use idempotence of $\nd$ to collapse the two possibilities.
After applying \ruleref{Progressing Rank}, we get the following triple:
\[
  \triple{\sure{x\ge 0}}{\whl{x>0}{x\coloneqq x-1 \oplus_{\frac12} (x \coloneqq x+1 \nd \skp)}}{\sure{x =0}}
\]

\section{Probabilistic SAT Solving by Partial Rejection Sampling}
\label{app:sat-solving}

In this section, we provide the technical details that were omitted from the main text in \Cref{sec:sat-solving}.

\subsection{Semantics of Lists}
\label{app:list-sem}

Before showing formal derivations for the subroutines, we discuss the semantics of lists used in the case study. A list is a finite sequence of values $\langle v_1, \ldots, v_n\rangle$. In order to manipulate lists in programs, we add the following expression syntax:
\[
  e \Coloneqq \cdots \mid \langle e_1, \ldots, e_n \rangle
  \mid
  e_1 [ e_2 ]
  \mid
  e_1 [ e_2 \mapsto e_3 ]
\]
The three new additions are list literals $\langle e_1, \ldots, e_n\rangle$, list accesses $e_1[e_2]$ (where $e_1$ is a list and $e_2$ is the index of the element being lookup up), and list updates $e_1 [ e_2 \mapsto e_3 ]$, which create a new list from $e_1$ with index $e_2$ now having value $e_3$.
The semantics of these expressions is below.
\begin{align*}
  \de{\langle e_1, \ldots, e_n \rangle}(\sigma) &\triangleq
    \langle \de{e_1}(\sigma), \ldots, \de{e_n}(\sigma) \rangle
  \\
  \de{e_1 [ e_2 ]}(\sigma) & \triangleq
    \left\{
      \begin{array}{ll}
        v_k & \text{if}~ \de{e_1}(\sigma) = \langle v_1, \ldots, v_n \rangle
        ~ \text{and}~ 
        \de{e_2}(\sigma) = k
        ~ \text{and}~ 
        1 \le k \le n
      \\
        0 & \text{otherwise}
      \end{array}
    \right.
  \\
  \de{e_1 [ e_2 \mapsto e_3 ]}(\sigma) & \triangleq
  \left\{
    \begin{array}{ll}
      \langle v_1, \ldots, v_{\max(k, n)} \rangle
      &\text{if}~ \de{e_1}(\sigma) = \langle v_1, \ldots, v_n\rangle
      \\ & \;\;\text{where}
        \de{e_2}(\sigma) = k
        ~ \text{and}~ 
        v_k = \de{e_3}(\sigma)
        \\&
        \;\; \text{and}~
        v_i = 0
        ~ \text{for}~ 
        n < i < k
      \\
      0 &\text{otherwise}
      \end{array}
    \right.
\end{align*}
Note that list operations never fail. Looking up an invalid index simply evaluates to 0. Additionally, if an update is out of bounds, then the intermediary values get filled in with zeros.
In addition, we provide some syntactic sugar in order to update lists in the style of arrays.
\[\arraycolsep=0pt
\begin{array}{lcll}
  x[ & e_1 & ] \coloneqq e_2
  &\qquad\triangleq\qquad
  x \coloneqq x [ e_1 \mapsto e_2 ]
  \\
  x[&e&] \coloneqq \mathsf{flip}(p)
  &\qquad\triangleq\qquad
  (x[e] \coloneqq \tru)
  \oplus_p
  (x [ e ] \coloneqq \fls)
\end{array}
\]

\subsection{$\selclause$ Derivation}

We now give the derivation for the $\selclause$ subroutine defined in \Cref{fig:solver-program}. Provided that the formula is not yet satisfied, $\selclause$ selects an arbitrary unsatisfied clause and stores its index in the variable $s$.

\begin{figure}
\[\def\arraystretch{1.2}
\small
\begin{array}{l}
\ob{\sure{\eval = \fls}}\\
\;\;  s \coloneq -1 \fse \\
\ob{\sure{\eval = \fls \land s = -1}}\\
\;\;  i \coloneq 1  \fse \\ 
\ob{\sure{\eval = \fls \land s = -1 \land i = 1}} \\
\lrob{\sure{\eval = \fls \land 1\le i} \land \bignd \left\{
    \begin{array}{l}
      \sure{s = -1 \land \textstyle\bigwedge_{k=0}^{i-1} \evalclause(k) = \tru}
      \\
      \sure{1 \le s \le M \land \evalclause(s) = \fls}
    \end{array}\right.} \\
\;\;  \mathsf{while} ~i \le M~ \mathsf{do} \\
\quad\lrob{\sure{\eval = \fls \land i = M-k} \land \bignd \left\{
    \begin{array}{l}
      \sure{s = -1 \land \textstyle\bigwedge_{k=0}^{i-1} \evalclause(k) = \tru}
      \\
      \sure{1 \le s \le M \land \evalclause(s) = \fls}
    \end{array}\right.} \\
\;\;  \quad \mathsf{if} ~\lnot  \evalclause(i)~ \mathsf{then} \\
\qquad\lrob{\sure{\eval = \fls \land i = M - k \land \evalclause(i) = \fls}  \land\bignd \left\{
    \arraycolsep=0em
    \begin{array}{l}
      \sure{s = -1 \land \textstyle\bigwedge_{k=0}^{i-1} \evalclause(k) = \tru}
      \\
      \sure{1 \le s \le M \land \evalclause(s) = \fls}
    \end{array}\right.} \\
\;\; \quad \quad \mathsf{if} ~s = -1~ \mathsf{then} \\
\qquad\quad\lrob{\sure{\eval = \fls \land i = M - k \land \evalclause(i) = \fls  \land
      s = -1 \land \textstyle\bigwedge_{k=0}^{i-1} \evalclause(k) = \tru}} \\
\;\; \quad \quad \quad s \coloneq i  \\
\qquad\quad\lrob{\sure{\eval = \fls \land s = i = M - k \land \evalclause(s) = \fls}} \\
\qquad\quad\lrob{\sure{\eval = \fls \land i = M - k \land 1 \le s \le M \land \evalclause(s) = \fls}} \\
\;\; \quad \quad \mathsf{else}  \\
\qquad\quad\lrob{\sure{\eval = \fls \land i = M - k \land \evalclause(i) = \fls  \land
      1 \le s \le M \land \evalclause(s) = \fls}} \\
\;\; \quad \quad \quad \skp \nd s\coloneqq i\fse \\
\qquad\quad\lrob{\sure{\eval = \fls \land i = M - k  \land
      1 \le s \le M \land \evalclause(s) = \fls}} \\
\qquad\lrob{\sure{\eval = \fls \land i = M - k  \land
      1 \le s \le M \land \evalclause(s) = \fls}} \\
\;\; \quad \mathsf{else}  \\
\qquad\lrob{\sure{\eval = \fls \land i = M - k \land \evalclause(i) = \tru}  \land\bignd \left\{
    \arraycolsep=0em
    \begin{array}{l}
      \sure{s = -1 \land \textstyle\bigwedge_{k=0}^{i-1} \evalclause(k) = \tru}
      \\
      \sure{1 \le s \le M \land \evalclause(s) = \fls}
    \end{array}\right.} \\
\;\;\quad\quad \skp \\
\qquad\lrob{\sure{\eval = \fls \land i = M - k} \land \bignd \left\{
    \arraycolsep=0em
    \begin{array}{l}
      \sure{s = -1 \land \textstyle\bigwedge_{k=0}^{i} \evalclause(k) = \tru}
      \\
      \sure{1 \le s \le M \land \evalclause(s) = \fls}
    \end{array}\right.} \\
\quad\lrob{\sure{\eval = \fls \land i = M-k} \land \bignd \left\{
    \begin{array}{l}
      \sure{s = -1 \land \textstyle\bigwedge_{k=0}^{i} \evalclause(k) = \tru}
      \\
      \sure{1 \le s \le M \land \evalclause(s) = \fls}
    \end{array}\right.} \\
\;\; \quad  i \coloneq i + 1 \\
\quad\lrob{\sure{\eval = \fls \land 1\le i \land M-i < k} \land \bignd \left\{
    \begin{array}{l}
      \sure{s = -1 \land \textstyle\bigwedge_{k=0}^{i-1} \evalclause(k) = \tru}
      \\
      \sure{1 \le s \le M \land \evalclause(s) = \fls}
    \end{array}\right.} \\
\ob{\sure{1 \le s \le M \land \evalclause(s) = \fls}}
\end{array}
\]
\caption{Derivation of the $\selclause$ subroutine.}
\label{fig:selclause-derivation}
\end{figure}

Termination of the loop is deterministic, but we can still analyze it using the \ruleref{Bounded Rank} rule. We use the ranking function $e_\mathsf{rank} \triangleq M-i$ so that as the iteration counter $i$ goes from 1 to $M$, the rank decreases towards 0. In addition, this decrease certainly occurs each iteration, so we let $p\triangleq 1$. The invariant is defined below:
\[
  P \triangleq \sure{\eval = \fls \land 1\le i} \land \bignd \left\{
    \begin{array}{l}
      \sure{s = -1 \land \textstyle\bigwedge_{k=0}^{i-1} \evalclause(k) = \tru}
      \\
      \sure{1 \le s \le M \land \evalclause(s) = \fls}
    \end{array}\right.
\]
The invariant has three components. First, it asserts that $\eval=\fls$, meaning that the formula is not satisfied, which will guarantee that the subroutine can find an unsatisfied clause. Next, $1\le i$, meaning that $i$ is either in bounds. Finally, either $s=-1$ and all the clauses up to $i$ are satisfied, or $1 \le s \le M$ and clause $s$ is not satisfied. Here, the $\nd$ plays a similar role to a disjunction; $\sure{Q_1}\nd \sure{Q_2}$ says that each element in the support of the underlying distribution satisfies either $Q_1$ or $Q_2$. Clearly $\sure{P \land i \le M}$ implies that the rank is bounded between $0$ and $M-1$:
\[
  \sure{P\land i\le M}
  \quad\Rightarrow\quad
  \sure{1\le i \le M}
  \quad\Rightarrow\quad
  \sure{0 \le M-i \le M-1}
\]
We now describe the derivation, shown in \Cref{fig:selclause-derivation}. Upon entering the loop, we presume that $e_\mathsf{rank} = k$ for some $1 \le k \le M-1$, which also means that $i = M-k$. We can then split each of the two $\nd$-clauses into two parts depending on whether $\evalclause(i)$ is true or false, since for any Boolean-valued expression $e$:
\[
  \sure{Q_1} \nd \sure{Q_2}
  \qquad\Rightarrow\qquad
  (\sure{Q_1} \nd \sure{Q_2}) \land \sure{e}
  \quad\nd\quad
  (\sure{Q_1} \nd \sure{Q_2}) \land \sure{\lnot e}
\]
We then use \Cref{lem:if-nd} to analyze the outer if statement. In the true branch, we again use \Cref{lem:if-nd} to split into the outcomes where $s=-1$ and $1 \le s \le M$. If $s=-1$, then we assign $s \coloneqq i$, which establishes the second case of the invariant. If $s \neq -1$, then we nondeterministically choose to do either nothing or assign $s \coloneqq i$, both of which establish the second case of the invariant. 

Nothing happens in the false branch of the outer if statement, although we do know in this case that $\evalclause(i) = \tru$, so we can increment the upper limit of the conjunction in the first case of the invariant. Joining the two branches and incrementing $i$ reestablishes the invariant and also decreases the rank.

Exiting the loop, we have $P \land i > M$, which makes the first case of the invariant impossible, since it would say that all the clauses are satisfied, despite the fact that $\eval=\fls$. So, we know the second case of the invariant must hold for the entire support.

\subsection{$\resampleclause$ Derivation}

In this section, we give the derivation for the $\resampleclause$ subroutine, which decreases the Hamming distance between the current solution $\vars$ and the sample solution $\xgood$ with probability at least $1/8$. The derivation is shown in \Cref{fig:resample}.

\begin{figure}
\[\def\arraystretch{1.25}
\begin{array}{l}
\ob{\sure{0\le i < M \land \dist(\vars, \xgood) = k \land \evalclause(i) = \fls}}\\
\ob{\sure{\sum_{j \in J} \big[\vars[j] \neq \xgood[j]\big]  < k}} \\
\;\; \vars[\clvars[i][0]] \coloneq \flip{\tfrac12} \fse  \\
\ob{\sure{\sum_{j \in J} \big[\vars[j] \neq \xgood[j]\big]  < k \land \vars[c_0] = \xgood[c_0]} \oplus_\frac12 \sure{\tru}} \\
\;\; \vars[\clvars[i][1]] \coloneq \flip{\tfrac12} \fse  \\
\ob{\sure{\sum_{j \in J} \big[\vars[j] \neq \xgood[j]\big]  < k
    \land \vars[c_0] = \xgood[c_0]
    \land \vars[c_1] = \xgood[c_1]    
    } \oplus_\frac14 \sure{\tru}} \\
\;\; \vars[\clvars[i][2]] \coloneq \flip{\tfrac12} \\
\ob{\sure{\sum_{j \in J} \big[\vars[j] \neq \xgood[j]\big]  < k
    \land \vars[c_0] = \xgood[c_0]
    \land \vars[c_1] = \xgood[c_1]    
    \land \vars[c_2] = \xgood[c_2]    
    } \oplus_\frac18 \sure{\tru}} \\
\ob{\sure{\dist(\vars, \xgood) < k} \oplus_{\frac18} \sure{\tru}}
\end{array}
\]
\caption{Derivation for the $\resampleclause$ subroutine.}
\label{fig:resample}
\end{figure}

At the beginning, we presume that $1\le i\le M$, $\dist(\vars,\xgood) = k$, and $\evalclause(s) = \fls$. Unfolding the definition of Hamming distance, we get $\sum_{j=1}^{N} \big[ \vars[j] \neq \xgood[j] \big] = k$.
Now let $c_j \triangleq \clvars[s][j]$ be the index of the $\nth{j}$ variable in the $\nth{s}$ clause and let $J \triangleq \{ j \mid 1 \le j \le N, j\notin \{ c_1, c_2, c_3 \} \}$ be the set of indices of variables that do not appear in the $\nth{s}$ clause. Since $\evalclause(s) = \fls$, it must be the case that $\vars[c_j] \neq \xgood[c_j]$ for some $j \in \{1,2,3\}$, therefore summing over $J$ instead of over all the indices must yield a value strictly less than $k$.

Each coin flip gives us $\vars[c_j] = \xgood[c_j]$ with probability $\frac12$, so all three are equal with probability $\frac18$. A sketch of the details of those derivations is given below, where $Q$ is the precondition of the command from \Cref{fig:resample}.
\[
\inferrule*[right=\ruleref{Consequence}]{
  \inferrule*[Right=\ruleref{ND Split}]{
    \inferrule*[right=\ruleref{Prob}]{
      \inferrule*[right=\ruleref{Assign}]{\;}{
        \triple{\sure{Q\land \xgood[e] = \tru}}{\vars \coloneqq \vars[ e \mapsto \tru]}{\sure{Q \land \vars[e] = \xgood[e]}}
      }
      \quad
      \vdots
    }{
      \triple{\sure{Q\land \xgood[e] = \tru}}{\vars[e] \coloneqq \flip{\tfrac12}}{\sure{Q \land \vars[e] = \xgood[e]} \oplus_\frac12 \sure{\tru}}
    }
    \quad
    \vdots
  }{
    \triple{\sure{Q\land \xgood[e] = \tru} \nd \sure{Q\land \xgood[e] = \fls}}{\vars[e] \coloneqq \flip{\tfrac12}}{\sure{Q \land \vars[e] = \xgood[e]} \oplus_\frac12 \sure{\tru}}
  }
}{
  \triple{\sure{Q}}{\vars[e] \coloneqq \flip{\tfrac12}}{\sure{Q \land \vars[e] = \xgood[e]} \oplus_\frac12 \sure{\tru}}
}
\]
The first step is to use the rule of \ruleref{Consequence} to split into the two cases where $\xgood[e]$ is either true or false---both lead to the same postcondition, which is collapsed using idempotence of $\nd$. Next, we use \ruleref{ND Split} to analyze the program for each value of $\xgood[e]$. Since $\vars[e]\coloneqq\flip{\frac12}$ is syntactic sugar for a probabilistic choice between two assignments, the derivation is completed with \ruleref{Prob} and \ruleref{Assign}.

Note that the same variable could appear multiple times in the clause ($c_j = c_\ell = n$ where $j\neq \ell$), but this does not affect our argument, since we only consider what happens when the coin flip matches $\xgood$. It does not matter if the first sample gave us $\vars[n] \neq \xgood[n]$ and then the second one gives us $\vars[n] = \xgood[n]$, since all we said in that case is $\tru$, which vacuously holds even if the variable is later resampled. As such, $\frac18$ is only a \emph{lower bound} on reducing the Hamming distance, but it is a tight enough bound to prove almost sure termination.

At the end of the routine, on the left side of the $\oplus_{\frac18}$, we have $\vars[c_j] = \xgood[c_j]$ for all $j\in\{1,2,3\}$, so expanding the bounds of the sum to be over all indices instead of over $J$ does not affect its value, since we are just adding zeros. We therefore get that $\dist(\vars,\xgood) < k$ with probability at least $\frac18$.

\fi

\end{document}